\documentclass[11pt,a4paper,onehalfspacing]{article}
\usepackage{rotating}
\usepackage{amsfonts}
\usepackage{amsmath}
\DeclareMathOperator*{\argmax}{arg\,max}

\usepackage{amsthm}
\usepackage{amssymb}
\usepackage{caption}
\usepackage{float}
\usepackage{graphicx}
\usepackage{xurl}
\usepackage[colorlinks=true, urlcolor=blue, pdfborder={0 0 0}]{hyperref}
\usepackage{enumerate}
\usepackage{multicol}
\usepackage{array}
\usepackage{bbm}
\usepackage{mathtools}
\usepackage{leftidx}
\usepackage[para]{footmisc}
\usepackage[nolist]{acronym}
\usepackage{scalerel}

\newtheorem{theorem}{Theorem}
\newtheorem{proposition}[theorem]{Proposition}
\newtheorem{lemma}[theorem]{Lemma}
\newtheorem*{lemma*}{Lemma}
\newtheorem*{assumption*}{Assumption}
\newcounter{assumptionc}
\newtheorem{assumption}[assumptionc]{Assumption}

\newtheorem{corollary}[theorem]{Corollary}

\numberwithin{equation}{section}
\numberwithin{theorem}{section}

\newtheorem{subassumption}{Assumption}[assumptionc]

\newenvironment{assumption+}
{\ifnum\value{subassumption}=0 \stepcounter{assumptionc}\fi\subassumption}
{\endsubassumption}

\def\sign{\mbox{sign}}

\makeatletter
\newcommand\footnoteref[1]{\protected@xdef\@thefnmark{\ref{#1}}\@footnotemark}
\makeatother

\def\text#1{\hbox{#1}}

\def\build #1_#2{\mathrel{\mathop{\kern 0pt #1}\limits_\zs{#2}}}
\newcommand{\zs}[1]{{\mathchoice{#1}{#1}{\lower.25ex\hbox{$\scriptstyle#1$}}
	{\lower0.25ex\hbox{$\scriptscriptstyle#1$}}}}
	
	\numberwithin{equation}{section}

	\usepackage[top=1.5in, bottom=1.5in, left=1in, right=1in]{geometry}

	\newtheorem*{example*}{Example}

	\newcommand{\hochkomma}{$^{,}$}

	\newcommand{\PP}{\mathbb{P}}
	\newcommand{\QQ}{\mathbb{Q}}
	\newcommand{\EX}{\mathbb{E}}
	\newcommand{\Real}{\mathbb{R}}

	\newcommand{\1}{\mathbbm{1}}
	
	\newcommand{\diff}{\mathrm{d}}

	\makeatletter
	\newenvironment{myproof}[1][\proofname]{%
\par\pushQED{\qed}\normalfont%
\topsep6\p@\@plus6\p@\relax
\trivlist\item[\hskip\labelsep\bfseries#1\@addpunct{.}]%
\ignorespaces
}{%
\popQED\endtrivlist\@endpefalse
}
\newcommand{\myitem}[1]{%
\item[#1]\protected@edef\@currentlabel{#1}%
}
\makeatother

\def\dfrac{\displaystyle\frac}
\def\tfrac{\textstyle\frac}

\def\FF{\mathcal{F}}

\usepackage{placeins}
\usepackage{accents}

\allowdisplaybreaks

\usepackage{setspace}
\begin{document}

\title{Optimal Capital Structure for Life Insurance Companies Offering Surplus Participation\footnote{Declarations of interest: none}
}
\author{Felix Fie{\ss}inger\footnote{Ulm University, Institute of Insurance Science and Institute of Mathematical Finance, Faculty of Mathematics and Economics, Ulm, Germany. Email: felix.fiessinger@uni-ulm.de} \, and Mitja Stadje\footnote{Ulm University, Institute of Insurance Science and Institute of Mathematical Finance, Faculty of Mathematics and Economics, Ulm, Germany. Email: mitja.stadje@uni-ulm.de} \hochkomma \footnote{Corresponding author}}
\date{\today}
\maketitle
\begin{abstract}
	This manuscript develops a dynamic capital structure model of life insurance companies offering participating contracts. Specifically, we explain why life insurers offer guaranteed payments either with or without policyholder surplus participation and derive conditions under which participating features are endogenously preferred. We show that surplus participation helps mitigate the asset substitution effect but is not always optimal. Depending on the tax rate and the recovery rate in bankruptcy, the insurer may optimally choose a non-participating contract.
\end{abstract}






\section{Introduction}
This paper studies why life insurers voluntarily share surplus with policyholders and how this contractual feature affects shareholders’ risk-taking incentives. We develop a dynamic capital structure model in which insurers choose between contracts that combine guaranteed payments with or without policyholder participation in surplus. The participation decision is endogenous.

Surplus participation has an important corporate-finance effect. By transferring part of the upside in the insurer’s asset value to policyholders, participation changes the curvature of the equity payoff and weakens shareholders’ incentives to increase asset risk. As a result, participating contracts can substantially mitigate the asset substitution problem. Participation, however, is not always optimal. We derive conditions under which insurers optimally choose participating or non-participating contracts. 
Participation is optimal when the tax benefits associated with the insurer’s liability structure are sufficiently large, whereas low tax rates or high competing payout obligations can make a non-participating contract optimal.
The model therefore rationalizes the coexistence of participating and non-participating contracts observed in life insurance markets.



Our basic framework adapts Leland's model from the insurance company's perspective, incorporating surplus participation into the original model. Our analysis shows that surplus participation fundamentally changes the endogenous default problem. Whereas prominent Leland-type models select the largest economically admissible solution of the smooth-pasting condition when multiple roots arise, in our setting the optimal bankruptcy trigger is given by the smallest such solution. This reversal reflects the altered shape of the equity value induced by policyholder surplus participation.

Leland's model, introduced in 1994 \cite{leland1994corporate}, was designed to determine the optimal leverage of companies by deriving the optimal capital structure between equity and debt. Initially simple, this model has closed-form solutions. A key feature of this setting is that equity holders are not assumed to be tied to their investments; rather they can liquidate the business if the company's asset value falls too low and obligations to debt holders become unsustainable. In this strand of literature, equity holders shut down the company if the value of the equity becomes negative. Leland then derives a bankruptcy-triggering value, such that for all asset values higher than this threshold, equity remains non-negative. By construction, the bankruptcy-triggering value has to be determined endogenously in the analysis. Various generalizations and adaptions of this model exist. Initially, Leland and Toft \cite{leland1996optimal} adapted the original framework to finite-time debt. In this paper, we build upon this adaptation, as we also consider contracts with finite durations. Further generalizations include Leland \cite{leland1998agency}, who incorporated capital restructuring, Goldstein et al. \cite{goldstein2001ebit}, who allowed the company to increase their debt level, and Manso et al. \cite{manso2010performance}, who introduced performance-dependent coupon levels. He and Milbradt \cite{he2016dynamic} examined the debt structure within a dynamic framework, where the firm has the flexibility to adjust its debt maturity structure in response to evolving market conditions. Hilpert et al. \cite{hilpert2024information} extended the model by incorporating asymmetric information between the firm and debt holders. They also introduced learning dynamics in the market over time and considered performance-sensitive debt. 

We show that the bankruptcy-triggering value is uniquely determined through a non-linear equation and is monotonically decreasing in the tax rates, but monotonically increasing in the surplus participation and the guarantee rate. 
We give sufficient conditions for the participation and the guarantee rate to be strictly positive. In particular, if the tax rate is sufficiently high, it is always beneficial for the insurance company, from a capital structure perspective, to offer a surplus participation. This provides a possible explanation for the peculiar structure of the life insurance market, where guarantees which match the policyholders' preference for safety, go often hand in hand with surplus participation. In the numerical analysis, we demonstrate that these conditions are typically satisfied. However, if other payout obligations, like the lump sum payment, are too large, or the tax rates (and therefore the tax benefits) are too low, then and only then it is advantageous for the insurance company not to offer a positive surplus participation. A sensitivity analysis indicates that the surplus participation is mainly exposed to changes in the dividend payout of the insurance company, the contract duration, and the tax rate.
Finally, we explore the so-called asset substitution effect, which is a type of agency costs. This effect describes the tendency of equity holders to increase the riskiness of a company's investment decisions, leading to a transfer of value from liabilities to equity. This phenomenon was first identified by Black and Scholes \cite{black1973pricing} and Jensen and Meckling \cite{meckling1976theory}. Subsequently, Merton \cite{merton1974pricing} and Barnea et al. \cite{barnea1980rationale} expanded upon this issue, identifying the core issue as the treatment of equity as a call option. However, when additional features such as guaranteed payments, taxes, and bankruptcy costs are incorporated, equity is no longer a classical call option, and the asset substitution effect weakens, particularly for shorter contract durations. Barnea et al. \cite{barnea1980rationale} already proposed that shorter durations diminish shareholders' incentives to increase investment risk. In our framework, we demonstrate that surplus participation further mitigates the asset substitution effect, rendering it a negligible factor. 
The reason is that the equity value $E$ is no longer convex with respect to the asset value, due to the participation costs becoming more significant for higher values of $V$. 
As a result, the rate of increase of $E$ in $V$ decreases, and $E$ adopts a concave form as $V$ grows 
In the absence of convexity, the "option-like" nature of equity diminishes thereby weakening shareholders’ incentives to increase asset risk.

We observe that for reasonable contract maturities, up to $50$ years, the asset-substitution effect virtually disappears when surplus participation is incorporated. 
Therefore, agency costs associated with asset substitution become economically negligible when such contracts are offered.

Participating policies are widespread in the life insurance sector. Profit participation is typically paid during favorable economic conditions, with a proportional participation, as in this paper, being the common example. In this sector, which also includes pension and health insurance, according to the European Insurance Overview 2023 \cite{EIOPAreport}, published by the European Insurance and Occupational Pensions Authority (EIOPA), approximately a quarter of the total gross premium in Europe is accounted for by contracts with some form of profit participation. In countries like Croatia, Italy, or Belgium, this proportion exceeds 50 \%. Research in this area is also ongoing, with many contributions, see, for instance, Bryis and de Varenne \cite{briys1997risk}, Bacinello and Persson \cite{rita2002design}, Gatzert and Kling \cite{gatzert2007analysis}, Schmeiser and Wagner \cite{schmeiser2015proposal}, Lin et al. \cite{lin2017optimal}, Mirza and Wagner \cite{mirza2018policy}, Nguyen and Stadje \cite{nguyen2020nonconcave}, He et al. \cite{he2020weighted}, Dong et al. \cite{dong2020optimal}, or Fie{\ss}inger and Stadje \cite{fiessinger2024mean}. To the best of our knowledge, we are the first to incorporate policyholder surplus participation into an endogenous dynamic capital structure framework providing a possible capital-based explanation for the peculiar structure of the life insurance market, where guaranteed interest is often but not always combined with surplus participation. According to, e.g., Kling et al. \cite{kling2007impact}, the combination of these two obligations to the policyholders are typical in the design of insurance products with surplus participation, whereby this combination is mostly studied focusing on managing the risk of the insurer, see, for instance, Kling et al. \cite{kling2007impact,kling2007interaction}, Hieber et al. \cite{hieber2015analyzing} or Schmeiser and Wagner \cite{schmeiser2015proposal}. Starting in the early 2000s, several publications have analyzed surplus participation products in the context of potential insolvency of life insurance companies, beginning with Grosen and J{\o}rgensen \cite{grosen2002life}. Subsequent works, such as those by Bernard et al. \cite{bernard2005market}, Ballotta et al. \cite{ballotta2006guarantees}, and Cheng and Li \cite{cheng2018early}, further explored this topic. However, these studies did not address the determination of the optimal bankruptcy-triggering value or the optimal capital structure, which are the pillars of dynamic capital models.

In the insurance market, there are several products offered with a surplus participation on the financial market result. 
The Principal Protected Notes offered by JPMorgan Chase provide one of many examples of such a structure.  
Another application of such a model is given in the context of occupational pension schemes, based on defined contribution (DC) plans. 
In some countries, a hybrid model exists that adds a guarantee to a DC plan. Due to these guarantees, employees, on the other hand, do not fully benefit from the returns of the fund, resulting in a product that offers both a guaranteed interest rate and a share in the fund's performance. 

The structure of the paper is as follows: Section \ref{chapter: basic and model setup} introduces participating life insurance contracts, outlines the basic model, and discusses the modeling of participation rates using option theory. In Section \ref{chapter: analysis}, we present the liability structure and the insurance company's value, followed by the derivation of the bankruptcy-triggering value. Section \ref{chapter: optimal rates} derives formulas for the optimal participation rate and the guarantee rate. In Section \ref{chapter: numerical}, we perform numerical analysis based on the result from the preceding section. Section \ref{chapter: conclusion} concludes the paper. All proofs are included in the Appendix.
A table~\ref{tab:variables} summarizing the notation used for the different variables is also provided.

\section{Basic Concepts and Model Setup} \label{chapter: basic and model setup}

\subsection{Participating life insurance contracts} \label{subsection: participating life insurance}

In the life insurance sector, contracts with surplus participation, or more broadly, contracts with profit participation, are commonplace. As illustrated in Figure \ref{fig: diagram}, which presents the market share of various lines of business in the life sector in 2022, approximately one-quarter of the total gross premiums are allocated to insurance contracts with a participation component. Furthermore, in countries such as Croatia, Belgium, and Italy, more than half of all gross premiums are invested in contracts featuring some form of participation.

\begin{figure}[!htb]
	\centering
	\includegraphics[trim= 22mm 98mm 28mm 28mm, clip,width=\textwidth]{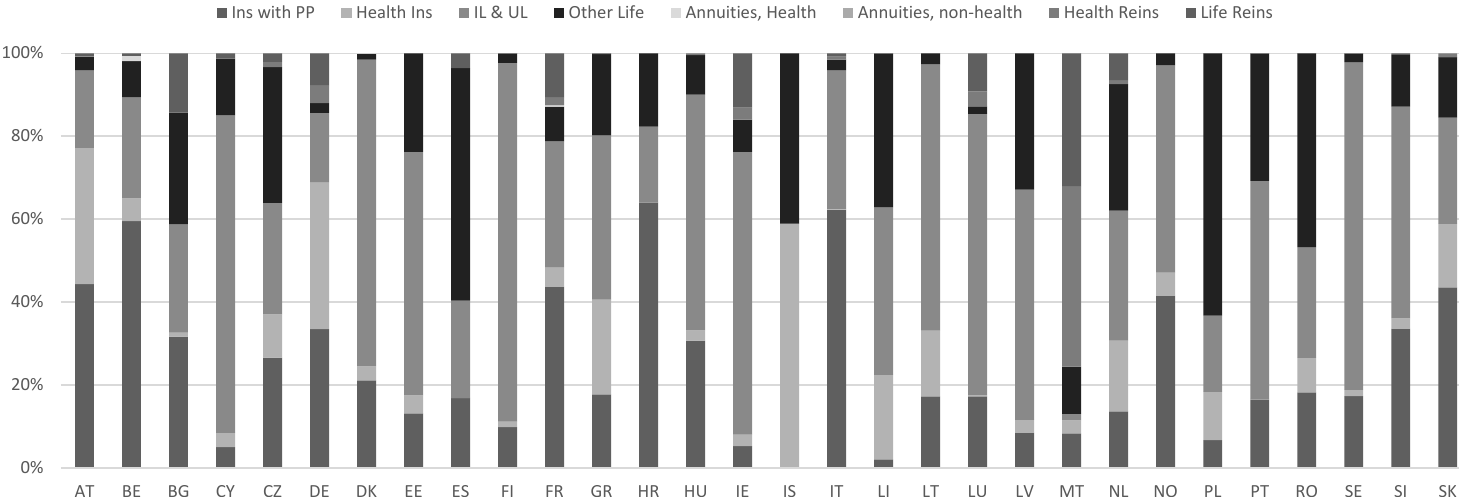}
	\caption{Market share in 2022 of the gross premium separated by the line of business in the life sector. Data source: European Insurance Overview from the EIOPA \cite{EIOPAreport}}
	\label{fig: diagram}
\end{figure}	

In countries like Germany, for instance, profit participation is legally mandated, as stipulated in the so-called ``Mindestzuf{\"u}hrungsverordnung'' (minimum allocation decree), requiring insurance companies to distribute cost-, risk-, and investment-surpluses to policyholder (if any) for all life insurance policies. While such legal requirements are less stringent in many other countries, many contracts still include profit participation, as can be seen from Figure \ref{fig: diagram}. 


In this paper, we focus on an insurance company that offers a single type of insurance contract. This contract includes a constant, guarantee rate $g \geq 0$, a deterministic lump sum payment at maturity, and an additional participation on the surplus exceeding a pre-determined threshold $k \geq 0$, with a participation rate $\alpha \in [0,1]$. The modeling of this participation presents some mathematical complexities, which we address in Subsection \ref{subsection: barrier} following the setup of the basic model.

\subsection{Model Setup} \label{subsection: model setup}

Let $(\Omega,\FF,(\FF)_{t \in [0,T]},\PP)$ be a filtered probability space with time horizon $T>0$, where the Brownian Motion $W$ generates the filtration, satisfying the usual conditions. Additionally, let $\QQ$ denote the pricing measure. We consider a market with a default-free, risk-free asset $B$ offering an interest rate $r > 0$. Following the approach of Leland and Toft \cite{leland1996optimal} (and its various generalizations, such as Ju et al. \cite{ju2005optimal}, Liu et al. \cite{liu2006personal}, or Hennessy and Tserlukevich \cite{hennessy2008taxation}), we model the asset value of the insurance company by the stochastic differential equation:
\begin{align*}
	\diff V_t = (\mu_t (V_t) - \nu) V_t \diff t + \sigma V_t \diff W_t,
\end{align*}
where $\mu_t$ represents the insurance company's total expected rate of return, $\nu>0$ is the constant fraction paid out (to equity holders and policyholders together), and $\sigma>0$ is the (constant) volatility. Moreover, we consider a constant default-triggering value $V_B \in [0,V_0]$, which determines the threshold at which the insurance company decides to default, if the asset value $V$ falls below $V_B$. 
We denote by $f^{V_0}$ the density and by $F^{V_0}$ the cumulative distribution function of the first passage time of the asset value $V$ to the bankruptcy-triggering value $V_B$ under risk-neutral valuation, i.e., under $\QQ$, where the drift rate of $V$ is given by $r-\nu$. For clarity, we explicitly state the dependence of $f$ and $F$ on the initial value $V_0$ as an upper index.

Furthermore, we let $\tau_1$ (resp. $\tau_2$) represent the tax rate of the insurance company on the guaranteed payment (resp. the surplus participation), and $\rho$ denotes the fraction of the asset value that is lost in the event of bankruptcy. Throughout the paper, we ignore that, apart from policyholders, there might be additional debt holders. Note that if only the company's profits are taxed, then it holds that $\tau_1 = \tau_2$ and both are equal to the corporate tax rate. The total value of the insurance company is denoted by $v$, which is partitioned into the equity value $E$ and the liability value $L$. The liability value $L$ encompasses the value of the payments to policyholders, including guaranteed rates, lump sum payments and surplus participation. It is important to note that the total value does not equate to the asset value of the insurance company due to tax benefits and bankruptcy costs. In event of bankruptcy, the remaining asset value, $(1-\rho)V_B$, will be distributed solely among the policyholders.

\subsection{Modeling participation rates with barrier options} \label{subsection: barrier}

By the construction of the participation part in this contract, as described in Subsection \ref{subsection: participating life insurance}, the surplus participation at maturity is given by $\alpha (V_T-k)_+$, where $(x)_+ := \max\{x,0\}$. This surplus participation is solely paid if the insurance company remains solvent, i.e., if $V_t \geq V_B$ for all $t \in [0,T]$, or equivalently $\min_{t \in [0,T]} V_t \geq V_B$. First, we observe that the payout of the surplus participation, without considering the bankruptcy condition, is equivalent to the payout of a call option. For now assume that the bankruptcy-triggering value $V_B$ is constant (which will be shown in Subsection \ref{subsection: derivation of vb}). Then the value of the surplus participation can be modeled as a so-called Down-and-Out-Call option. Specifically, the barrier is set as $V_B$, the strike price as $k$, and the asset value as $V$. Consequently, the value of the Down-and-Out Call option, corresponding to the surplus participation, with maturity $T$ and dividend rate $\nu$, is given by:
\begin{align*}
	c_{do} (V_0,k,V_B,T) = e^{-rT} \EX^{\QQ} [(V_T-k)_+ \1_{\{\min_{s \in [0,T]} V_s \geq V_B\}}].
\end{align*}
For further details on barrier options, we refer to Appendix \ref{subsection: barrier appendix} or to Hull \cite{hull2017options}.

\section{Insurance company setup and determination of the bankruptcy-triggering value} \label{chapter: analysis}

For the following analysis, we assume that the insurance company continuously issues contracts with identical characteristics, so that the portfolio remains stationary over time. Specifically, as long as the insurer remains solvent, the value of maturing contracts at any given time equals the value of newly issued contracts. We further assume that contract maturities are uniformly distributed over each interval $[s, s+T]$, and---without loss of generality---set $s = 0$ in the formulas. These assumptions follow Leland's framework for finite-maturity debt\footnote{$T < \infty$}, as in \cite{leland1996optimal}, \cite{ju2005optimal}, \cite{perrakis2012structural}, and \cite{palmowski2020leland}. In practice, perpetual bonds are rare, so it is common for bond maturities to differ from the planning horizon. In insurance, maturities $T$ of ten, twenty, or---less commonly---thirty years are typical.

In the cited literature, the total liability value at time zero equals the accumulated payments from $0$ to $T$, since at that time the liability consists solely of bonds maturing at or before $T$. Consequently, although these bonds are continuously retired and refinanced to maintain a stationary portfolio, no liability valuation beyond $T$ is required as all bonds in the portfolio mature within the next $T$ time periods. By contrast, the tax benefit persists under perpetual refinancing throughout the entire planning horizon, which requires integration from $0$ to $\infty$.



Let $G$ denote the total amount of guaranteed payments per year, and $P>0$ the total amount of lump sum payments at maturity. Based on the assumptions above, both $G$ and $P$ are time-independent, in contrast to the surplus participation component. The constant guarantee rate is given by $g = \frac{G}{T}$, the constant yearly lump sum payment rate is $p = \frac{P}{T}$, and the insurance company pays out a total of $G+\frac{P}{T}$ along with the random surplus participation per year. In practice, the model remains valid even if the insurer pays out the rates at later times (e.g., monthly or yearly) and invests the money in the risk-free asset during the intermediate time period.

In the analysis, we restrict the parameters to those that reflect reasonable market conditions for equity. By ``reasonable market conditions'', we refer to a market in which limited liability holds, i.e., the equity value, $E(V)$, cannot be negative. Although without restrictions on the parameters, excessively high surplus participation promises could lead to such a scenario, an insurance company offering such products would face challenges in attracting investors. Thus, for the remainder of this paper, we assume that equity is non-decreasing in the asset value and non-negative.

Additionally, we assume that the liability value of the guaranteed payment exceeds the tax benefit associated with this payment\footnote{\label{footnote: global local}From a mathematical perspective, this excess is not required to hold globally. It suffices for this condition to be valid when the insurance company's asset value is close to the bankruptcy-triggering value.}. This assumption is also necessary in the absence of surplus participation, such as in the basic model of Leland and Toft \cite{leland1996optimal}, where it is implicitly assumed, even though not explicitly stated.\footnote{In the case of no surplus participation, we later demonstrate (see \ref{eq: vb alpha=0 main text}) that the optimal bankruptcy-triggering value, $V_B$, is affine-linear in the guaranteed payment $G$. Thus, this assumption ensures that $V_B$ is non-decreasing and, consequently, non-negative.} If this assumption is not satisfied, we find that $V_B$ is decreasing in $G$, leading to the paradoxical situation that it would be advantageous for the policyholders, the insurance company, and the equity holders to increase the guaranteed payments to infinity due to the high tax benefit. 
Therefore, we exclude this possibility for the remainder of the paper. Numerical analysis indicates that this assumption may not hold if the tax rate is close to $100 \%$ and the lump sum payment is small. However, if the contract duration, $T$, approaches infinity\footnote{Life insurance companies often hold long-duration bonds with durations between 10 to 30 years to match long time liabilites like life policies or annuities.}, the present value of the liability associated with the guaranteed payment will exceed the tax benefit derived from it, regardless of the parameter choices.

We also make a similar assumption regarding surplus participation: namely, that the liability value of the surplus participation exceeds the tax benefit associated with it\footnoteref{footnote: global local}. If this assumption does not hold, we again encounter paradoxical situations where increasing surplus participation to infinity would be advantageous for the policyholders, the insurance company, and the equity holders due to the high tax benefit. Our analysis indicates that this assumption may not hold true if, for instance, the tax rate is excessively high. However, as long as the contract duration is finite, the assumption is always valid. As such, we exclude this possibility from further analysis in this paper.

Finally, we assume that \(k \geq V_0\); that is, the threshold at which surplus 
participation begins is greater than or equal to the initial portfolio value. 
This assumption is typically satisfied in surplus participation contracts, as 
the baseline level for surplus determination is often chosen to be the initial 
value augmented by suitable safety margins.

For simplicity, from this point forward, we will replace the initial value $V_0$ with $V$ in the notation.

\subsection{Liability Structure} \label{subsection: liability}

In this subsection, we derive formulas for the liability associated with a fixed bankruptcy-triggering level $V_B$. Additionally, we restate the results from Leland and Toft \cite{leland1996optimal} concerning the non-participation component of the liabilities. That $V_B$ is in the optimum chosen constantly follows from the stationarity assumption and will be demonstrated later\footnote{See in particular Subsection \ref{subsection: derivation of vb}.}. Before defining the liabilities for the entire portfolio, we begin by considering the liability stemming from a portfolio with maturity $t$, denoted by $l$. We let $\tau$ represent the stopping time at which the asset value of the insurance company $V_t$ hits the bankruptcy triggering value $V_B$, i.e., $\tau := \inf_{s \geq 0} \{ V_s \leq V_B \}$:
\begin{align}
	l(V;V_B,t) =&\, \EX^\QQ \left[\int_0^t e^{-rs} g \1_{\{\min_{r \in [0,s]} V_r \geq V_B\}} \diff s \right] + \EX^\QQ \left[ e^{-rt} p \1_{\{\min_{s \in [0,t]} V_s \geq V_B\}} \right] \notag \\
	&+ \EX^\QQ \left[e^{-r \tau} (1-\rho) V_B \1_{\tau \leq t} \right] \notag\\
	&+ \EX^\QQ \left[ e^{-rt} \alpha (V_t-k)_+ \1_{\{\min_{s \in [0,t]} V_s \geq V_B\}} \right] \label{eq: def liability for small t}\\
	=&\, \int_0^t e^{-rs} g (1-F^V(s)) \diff s + e^{-rt} p (1-F^V(t)) + \int_0^t e^{-rs} (1-\rho)V_B f^V(s) \diff s \notag\\
	&+ \alpha c_{do} (V,k,V_B,t) \notag\\
	=&\, \frac{g}{r} + e^{-rt} \left( p - \frac{g}{r} \right) (1-F^V(t)) + \left( (1-\rho)V_B - \frac{g}{r} \right) G^V(t) + \alpha c_{do} (V,k,V_B,t), \notag
\end{align}
where $G^V(t) := \int_0^t e^{-rs} f^V(s) \diff s$. The last step follows from integration by parts of the first term. In this equation, the first term corresponds to the guaranteed payment, the second term represents the final lump sum payment, the third term reflects the remaining asset value in event of bankruptcy, and the fourth term accounts for the surplus participation. 

Now, we proceed with the liability of the entire portfolio with maturity $T$, i.e., the total liability value $L$, which is given by:
\begin{align}
	L(V;V_B,T) =&\, \int_0^T l(V;V_B,t) \diff t \label{eq: liability value L}\\
	=&\, \tfrac{G}{r} + ( P - \tfrac{G}{r} ) ( \tfrac{1-e^{-rT}}{rT} - I_1^V(T) ) + \left( (1-\rho)V_B - \frac{G}{r} \right) I_2^V(T) + \alpha \int_0^T c_{do} (V,k,V_B,t) \diff t, \notag
\end{align}
where $I_1^V (T) := \tfrac{1}{T} \int_0^T e^{-rt} F^V(t) \diff t$ and $I_2^V(T) := \tfrac{1}{T} \int_0^T G^V(t) \diff t$. The functions $F^V$, $G^V$, $I_1^V$, and $I_2^V$ admit explicit formulas, which are provided in Appendix \ref{app liability}.

\subsection{Total value and equity value}
\label{subsection:total value equity}

As outlined in the model setup in Subsection \ref{subsection: model setup}, the total value of the insurance company, $v$, is the sum of the actual value and the tax benefit $TB$, minus the lost value in the event of bankruptcy, denoted by $BC$. Specifically, we have:
\begin{align*}
	v = V + TB - BC.
\end{align*}
By definition of the model, the tax benefit $TB$ can be decomposed into two components: $TB = TB_1 + TB_2$, where $TB_1$ represents the tax benefit arising from the guaranteed payment, and $TB_2$ corresponds to the tax benefit from the participation component.
The tax benefit $TB$ and the (potential) bankruptcy costs $BC$ persist under perpetual refinancing throughout the entire planning horizon, which requires integration from $0$ to $\infty$.
In this context, $TB_1$ is given by $\tau_1 G$, where $G$ is the total guaranteed payment, as long as the insurance company remains solvent, i.e., as long as $\min_{s \in [0,t]} V_s \geq V_B$. A similar interpretation applies to $TB_2$, with $\tau_2$ replacing $\tau_1$ and the value of the participation component substituting the value of the guarantee. Finally, $BC$ represents the value of the bankruptcy costs. Thus, we get with $S := \inf\{t>0: V_t \leq V_B\}$:
\begin{align}
	TB_1 &= \tau_1 \EX^\QQ \left[ \int_0^\infty e^{-rt} G \1_{V_s \geq V_B \, \forall s \in [0,t]} \diff t \right], \label{eq: def TB1}\\
	TB_2 &= \tau_2 \EX^\QQ \left[ \int_0^\infty e^{-rt} \alpha(V_t-k)_+ \1_{V_s \geq V_B \, \forall s \in [0,t]} \diff t \right], \label{eq: def TB2} \\
	BC &= \rho \EX^\QQ \left[ e^{-rS}  V_B \1_{S < \infty} \right]. \label{eq: def BC}
\end{align}

For the non-participation terms $TB_1$ and $BC$, and applying Tonelli's theorem for $TB_2$, one obtains:
\begin{align} \label{eq: firm value v}
	v(V;V_B) &= V + TB_1 + TB_2 - BC \notag \\
	&=V + \tau_1 \tfrac{G}{r} (1-(\tfrac{V_B}{V})^{\lambda_2+\lambda_3}) + \tau_2 \EX^\QQ \left[ \int_0^\infty e^{-rt} \alpha(V_t-k)_+ \1_{V_s \geq V_B \, \forall s \in [0,t]} \diff t \right] - \rho V_B (\tfrac{V_B}{V})^{\lambda_2+\lambda_3} \notag \\
	&= V + \tau_1 \tfrac{G}{r} (1-(\tfrac{V_B}{V})^{\lambda_2+\lambda_3}) + \tau_2 \alpha \int_0^\infty c_{do} (V,k,V_B,t) \diff t - \rho V_B (\tfrac{V_B}{V})^{\lambda_2+\lambda_3}.
\end{align}
The value of the equity $E$ is simply given by:
\begin{align}
	\label{equity}
	E(V;V_B,T) = v(V;V_B) - L(V;V_B,T).
\end{align}

\subsection{Derivation of the bankruptcy-triggering value $V_B$} \label{subsection: derivation of vb}

To choose the equilibrium bankruptcy-triggering value $V_B$ which maximizes equity, we employ the smooth-pasting condition, also known as low-contact rule 
which is given by the smallest solution of
\begin{align} \label{eq: smooth pasting condition}
	\dfrac{\partial E(V;V_B,T)}{\partial V} \Big|_{V=V_B} = 0.
\end{align}
This condition maximizes both the equity and the insurance company's value with respect to $V_B$, under the condition of the limited liability for equity holders (i.e., equity holders can always walk away), ensuring that $E(V) \geq 0$ for all $V \geq V_B$. Furthermore, it holds that $E(V)=0$ for all $V \leq V_B$. 
Using the smooth-pasting condition, $V_B$ is chosen endogenously ex post via a maximization. $E(V) \geq 0$ for all $V \geq V_B$ guarantees that $\tfrac{\partial^2}{\partial V^2} E(V;V_B,T) \big|_{V=V_B} \geq 0$ and that $\tfrac{\partial E(V;V_B,T)}{\partial V_B}=0$ for any level of $V$. Note that the solution to equation \eqref{eq: smooth pasting condition} is independent of time, $t$, i.e., $V_B$ is constant in the analysis. If multiple solutions exist, we select the smallest solution, as this is the only one consistent with the limited liability of equity\footnote{The choice of the smallest solution, $V_B$, satisfying the smooth-pasting condition is somewhat surprising, and stands in contrast to typical results in related capital structure optimization problems, where the maximal solution must be chosen. However, we later show that there exists at most a single solution at which \(V \mapsto E(V;V_B)\) attains a minimum at \(V = V_B\), and that in total there are at most two solutions resp. solution intervals. The second solution, if it exists, corresponds to a local maximum of \(V \mapsto E(V;V_B)\) at \(V = V_B\) (as minima and maxima alternate), which is not admissible due to the requirements of non-negative equity and \(E(V_B;V_B) = 0\). Hence, there is at most one admissible solution.}. This implies that the equity value $E$ is increasing in the insurance company's asset value $V$ for all $V \geq V_B$. 

For a more detailed derivation of the smooth pasting condition, and its equivalence to $\tfrac{\partial E(V;V_B,T)}{\partial V_B}=0$, we refer to Merton \cite{merton1973theory}, Dixit \cite{dixit1991simplified}, Dumas \cite{dumas1991super}, Leland \cite{leland1994corporate}, and He and Milbradt \cite{he2016dynamic}. 
The main idea is to note first that by stationarity, $V_B=V_B^*$ if and only if for $V(0)=V_B^*$ the continuation value is negative. Hence, $V_B^*$ must be independent of time $t$ and of the asset value $V$. 
	
	Setting $V_B=V_B^*$ in the sequel, we derive \eqref{eq: smooth pasting condition} by observing that $h(V) := E(V,V) = 0$ for all $V \geq 0$. Hence
	\[
	\frac{\partial E(V;V_B,T)}{\partial V} \Big|_{V=V_B}
	\;+\;
	\frac{\partial E(V;V_B,T)}{\partial V_B} \Big|_{V=V_B}
	= \frac{\partial}{\partial V} h(V) \Big|_{V=V_B}
	= 0,
	\]
	where the smoothness of $E$ holds by Lemma \ref{lemma: cdo c1} and Lemma \ref{lemma: new} in the Appendix.
	Next, let us provide a brief discussion of the intuition behind the smooth-pasting condition \eqref{eq: smooth pasting condition}: The insurance company would set $V_B$ as low as possible in order to maximize its value and prefers to avoid bankruptcy (because of the bankruptcy costs). Conversely, equity holders want to ensure that the equity value is always non-negative. Due to their limited liability, they will liquidate the insurance company (i.e., stop payments) if the equity becomes negative. The equity holders determine the level of $V_B$ after the insurance company has finalized its liability structure. It is important to note that if $V_B$ is (in theory) set too low, the equity value would become negative if the insurance company's assets are low, as the guaranteed payment would become too costly. This also explains the underlying minimization problem in the smooth-pasting condition, as the equity holders minimize $V_B$ to the lowest possible value such that the equity capital remains non-negative. 
	Additionally, due to the absolute priority rule, the value of equity is (theoretically) $0$ for every $V \leq V_B$. The term ``low-contact rule'' refers to the boundary condition that for equity $E$, seen as a function of $V$ and $V_B$, the set where $V=V_B$ determines a boundary where the function $E$ is defined.

	From this point forward, we will assume, as previously discussed, that the liability value of the guaranteed payment and of the surplus participation exceeds the associated tax benefit. Consequently, based on the framework outlined in the preceding subsections, the following inequalities are assumed to hold throughout the rest of this paper:
	\begin{align}
		\int_0^T \EX^\QQ \left[\int_0^t e^{-rs} g \1_{\{\min_{r \in [0,s]} V_r \geq V_B\}} \diff s \right] \diff t &\geq TB_1, \label{eq: ass guarantee excess} \\
		\int_0^T c_{do} (V,k,V_B,t) \diff t &\geq \tau_2 \int_0^\infty c_{do} (V,k,V_B,t) \diff t. \label{eq: ass surplus excess}
	\end{align}
	
	In the following theorem, we apply this smooth-pasting condition and provide a formula where the solution yields the bankruptcy-triggering value $V_B$. Before proceeding, we introduce the following shorthand notations with $\lambda_1$ as defined in \eqref{eq: def d12 lambda1} and $\lambda_2,\lambda_3$ as defined in \eqref{eq: def lambda23}:
	\begin{align}
		A_1 &:= \tfrac{\lambda_2-\lambda_3}{2} + \lambda_3 \Phi (\lambda_3 \sigma \sqrt{T}) - \lambda_2 e^{-rT} \Phi(\lambda_2 \sigma \sqrt{T})>0, \label{eq: a1} \\
		A_2 &:= \tfrac{\lambda_2-\lambda_3}{2} - \tfrac{1}{2 \lambda_3 \sigma^2 T} + (\lambda_3+\tfrac{1}{\lambda_3 \sigma^2 T} ) \Phi (\lambda_3 \sigma \sqrt{T}) + \tfrac{\varphi (\lambda_3\sigma\sqrt{T})}{\sigma \sqrt{T}}>0, \label{eq: a2} \\
		A_3 &:= \tfrac{\lambda_1}{\nu} + \tfrac{1}{\sigma \nu} \sqrt{\lambda_1^2 \sigma^2 + 2\nu}>0, \label{eq: a3} \\
		A_4 &:= \tfrac{\lambda_1}{\nu} (1- 2e^{-\nu T} \Phi(\lambda_1\sigma\sqrt{T})) + \tfrac{1}{\sigma \nu} \sqrt{\lambda_1^2 \sigma^2 + 2\nu} (2 \Phi(\sqrt{\lambda_1^2\sigma^2+2\nu}\sqrt{T})-1)>0, \label{eq: a4} \\
		\bar{\alpha} &:= \begin{cases}
			\nu (1-\tau_2-e^{-\nu T})^{-1}, & \text{if } 1-\tau_2-e^{-\nu T} >0, \\
			\infty, & \text{else,}
		\end{cases} \label{eq: baralpha} \\
		\tilde{\alpha} &:= \begin{cases}
			\tfrac{1+\rho(\lambda_2+\lambda_3)+2(1-\rho)A_2}{A_4-\tau_2A_3}, & \text{if } A_4-\tau_2A_3 >0, \\
			\infty, & \text{else.}
		\end{cases} \label{eq: tildealpha}
	\end{align}
	where $\Phi$ denotes the cumulative distribution function of a standard normal distribution.
	$A_1$ is the scaling factor of the derivative of the function $I_1^V(T)$ (defined below \eqref{eq: liability value L}) at $V_B$, i.e., a scaling factor of the sensitivity of the discounted guaranteed payments w.r.t. the company's value at bankruptcy.
	$A_2$ is the scaling factor of the derivative of the function $I_2^V(T)$ (defined below \eqref{eq: liability value L}) at $V_B$, i.e., a scaling factor of the sensitivity of the discounted guaranteed lump sum payment w.r.t. the company's value at bankruptcy.
	$A_3$ is the sensitivity of the Down-and-Call option, corresponding to the surplus participation over an infinite time horizon w.r.t. changes in the company's value.
	$A_4$ is the sensitivity of the option over an finite time horizon $T$ w.r.t. changes in the company's value.
	$A_1$, $A_2$, $A_3$, and $A_4$ are indeed non-negative, as shown in Lemma \ref{lemma: a1,a2,a3,a4>0}. Additionally, we adopt the convention that $[0,\bar{\alpha}] = [0,\infty)$ when $\bar{\alpha}=\infty$.
	
	\begin{theorem} \label{th: VB determination}
		The bankruptcy-triggering value $V_B$ is determined as the minimum of $V_0$ and the smallest solution of the following formula:
		\begin{align}
			0 =&\, 1 + \rho(\lambda_2+\lambda_3)+2(1-\rho)A_2 - \tfrac{1}{V_B} \Big( \tfrac{2(P-\frac{G}{r})A_1}{rT} + 2 \tfrac{G}{r} A_2 - \tau_1 \frac{G}{r}(\lambda_2+\lambda_3) \Big) \notag \\
			&+ \tau_2 \alpha \int_0^\infty \tfrac{\partial  c_{do} (V,k,V_B,t) }{\partial V} \big|_{V=V_B} \diff t - \alpha \displaystyle\int_0^T \tfrac{\partial c_{do} (V,k,V_B,t)}{\partial V} \big|_{V=V_B} \diff t, \label{eq: formula for vb}
		\end{align}
		where $\lambda_2,\lambda_3$ are as in \eqref{eq: def lambda23},
		\begin{align}
			\tfrac{\partial}{\partial V} c_{do} (V,k,V_B,t) \big|_{V=V_B} =&\, 2e^{-\nu t} \Big(  \lambda_1 \Phi(d_1(\min\{\tfrac{V_B}{k},1\},t))+ \frac{\varphi(d_1(\min\{\frac{V_B}{k},1\},t))}{\sigma \sqrt{t}} \Big) \notag \\
			&- \tfrac{2ke^{-rt}}{V_B} \Big(\lambda_2  \Phi(d_2(\min\{\tfrac{V_B}{k},1\},t)) + \frac{ \varphi(d_2(\min\{\frac{V_B}{k},1\},t))}{\sigma \sqrt{t}} \Big), \label{eq: dv cdo v=vb}
		\end{align}
		and $d_1$ as in \eqref{eq: def d12 lambda1}. In particular, this formula is well-defined and it holds that $V_B > 0$. 
	\end{theorem}
	
	Note that the two (lengthy) terms in the first line of equation \eqref{eq: formula for vb} are positive, as demonstrated in Lemma \ref{lemma: long term >0}. This theorem ensures that a solution for the bankruptcy-triggering value always exists. However, if the solution of \eqref{eq: formula for vb} is larger than $V_0$ (or no solution exists), we can set $V_B = V_0$, because if $V_B \geq V_0$, the company declares bankruptcy immediately. Furthermore, the theorem reveals that the bankruptcy-triggering value depends on the chosen contract maturity, which aligns with the basic model of Leland and Toft \cite{leland1996optimal}. However, in models with flow-based bankruptcy or a positive net worth covenant, the bankruptcy-triggering value is independent of the maturity, as shown in works by, e.g., Kim et al. \cite{kim1993does}, Longstaff and Schwartz \cite{longstaff1995simple}, or Ross \cite{ross2005capital}. 
	
	In the following figure, Figure \ref{fig: rvb}, we provide a graphical illustration of the solution of formula \eqref{eq: formula for vb}. In the left plot, we depict the right-hand side of this formula, where the intersection of the graph with the horizontal axis at $0$ corresponds to $V_B$. In the right plot, we show the first line and the negative of the second line from \eqref{eq: formula for vb}. The intersection of these two lines represents $V_B$. This plot illustrates the most typical scenario, where a unique solution to \eqref{eq: formula for vb} exists.
	
	\begin{figure}[!htb]
		\centering
		\begin{minipage}[t]{0.48\textwidth}
			\includegraphics[trim= 1mm 6mm 10mm 8mm, clip,width=\textwidth]{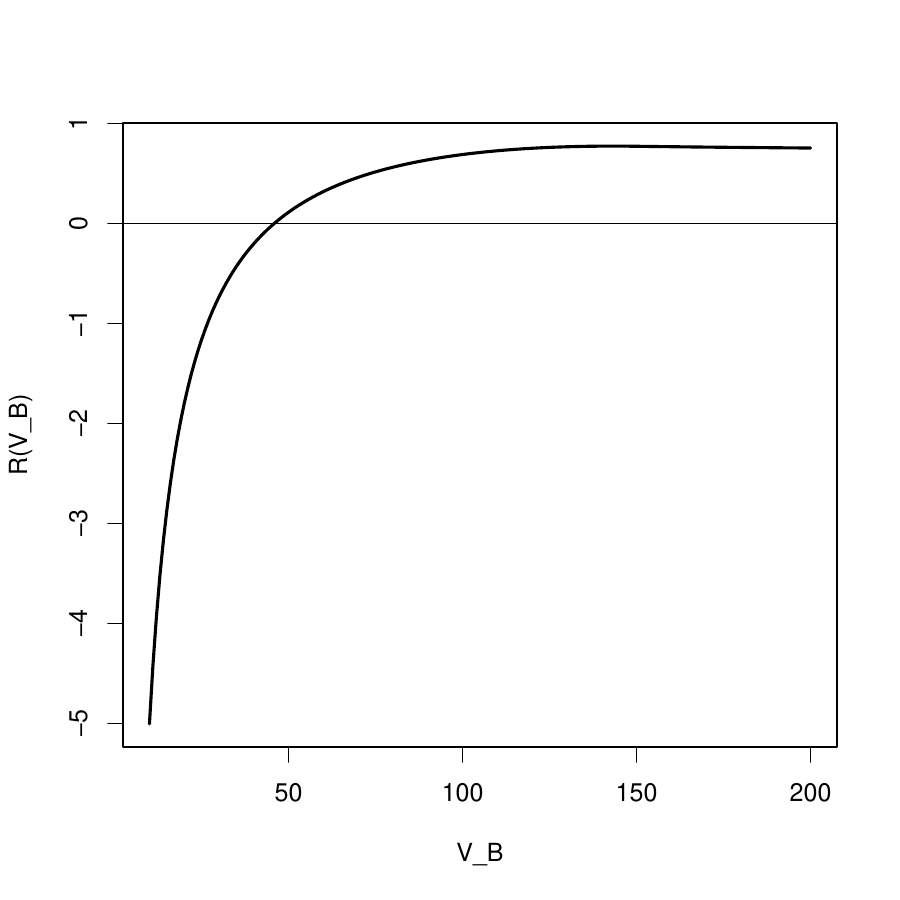}
		\end{minipage}
		\quad
		\begin{minipage}[t]{0.48\textwidth}
			\includegraphics[trim= 1mm 6mm 10mm 8mm, clip,width=\textwidth]{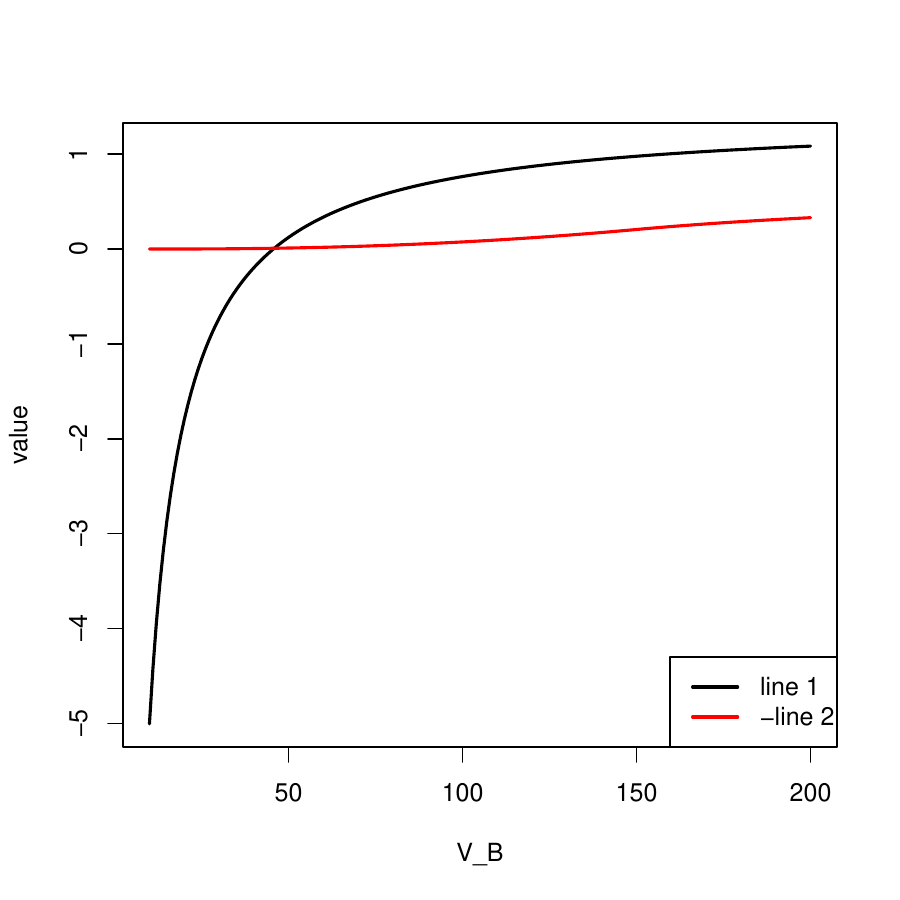}
		\end{minipage}
		\caption{Plot of the right hand side of \eqref{eq: formula for vb} (left) and of its first line and the negative of its second line (right). Note that the intersection with $0$ (left) resp. of the two lines (right) is $V_B$.}
		\label{fig: rvb}
	\end{figure}
	
	In the following propositions, we explore the influence of several parameters on the optimal bankruptcy-triggering value and provide a sufficient mathematical condition for the assumption that $V \to E(V)$ is non-negative.
	
	\begin{proposition} \label{prop: vb monoton tau T}
		The optimal bankruptcy-triggering value $V_B$ is monotonically decreasing in the tax rates $\tau_1$ and $\tau_2$. Moreover, if $P-\tfrac{G}{r} \leq 0$, we find that $V_B$ is monotonically increasing in the contract maturity $T$.
	\end{proposition}
	
	This result seems reasonable, as larger tax rates increase the equity value through a larger tax benefit, which leads to a lower bankruptcy-triggering value. On the other hand, a longer contract duration results in higher guarantee payments, which increases the liability value and consequently raises the bankruptcy-triggering value.
	
	\begin{proposition} \label{prop: vb monoton alpha g}
		The optimal bankruptcy-triggering value $V_B$ is monotonically increasing and left-continuous in both the surplus participation rate $\alpha$ and the guaranteed payment $G$. Furthermore, the right-limits exist.
	\end{proposition} 
	
	As both a larger surplus participation and a larger guaranteed payment increase the payment obligations to the policyholders, equity holders will opt for a larger bankruptcy-triggering value to offset the increased liabilities.
	
	The following proposition demonstrates that $V \to E(V)$ is actually non-negative (as assumed) as long as the participation rate is not unreasonably high.
	
	\begin{proposition} \label{prop: alpha bar competitive market}
		A sufficient condition for our assumption that $V \to E(V)$ being non-negative is that $\alpha<\bar{\alpha}$, where $\bar{\alpha}$ is as defined in \eqref{eq: baralpha}.
	\end{proposition}
	
	Proposition \ref{prop: alpha bar competitive market} also implies that the tax benefit associated with the tax rate $\tau_2$ significantly influences $\bar{\alpha}$. In particular, we observe that the limit $\bar{\alpha}$ increases with $\tau_2$ as long as $1-\tau_2-e^{-\nu T} > 0$, and is infinity beyond this point. This is intuitive, as a higher tax benefit makes the participation structure more advantageous, and it ensures that the equity remains increasing in the asset value (since the tax benefit is itself increasing in asset value). On the other hand, if $1-\tau_2-e^{-\nu T} \leq 0$, the increase in the value of the tax benefit surpasses the decrease in the value of the liabilities for the surplus participation component, even for arbitrarily high bankruptcy-triggering values. This seems unnatural and a numerical analysis reveals that this scenario does not occur for reasonable parameter values. However, our results cover this situation as well.
	
	The following corollary demonstrates a condition on the surplus participation such that the formula \eqref{eq: formula for vb} has a unique solution:
	
	\begin{corollary} \label{cor: unique solution of vb determination}
		A sufficient condition to have a unique solution of \eqref{eq: formula for vb} is that $\alpha<\tilde{\alpha}$, where $\tilde{\alpha}$ is as defined in \eqref{eq: tildealpha}.
	\end{corollary}
	
	In Figure \ref{fig: EValphas}, we illustrate how the equity value evolves as a function of the asset value under two conditions: when $V \to E(V)$ is non-negative (left) and when it is not (right), i.e., when the participating rate is set too high. 
	If there is no surplus participation, i.e., $\alpha=0$, $V \to E(V)$ is, of course, increasing.
	
	\begin{figure}[!htb]
		\centering
		\begin{minipage}[t]{0.48\textwidth}
			\includegraphics[trim= 1mm 6mm 10mm 20mm, clip,width=\textwidth]{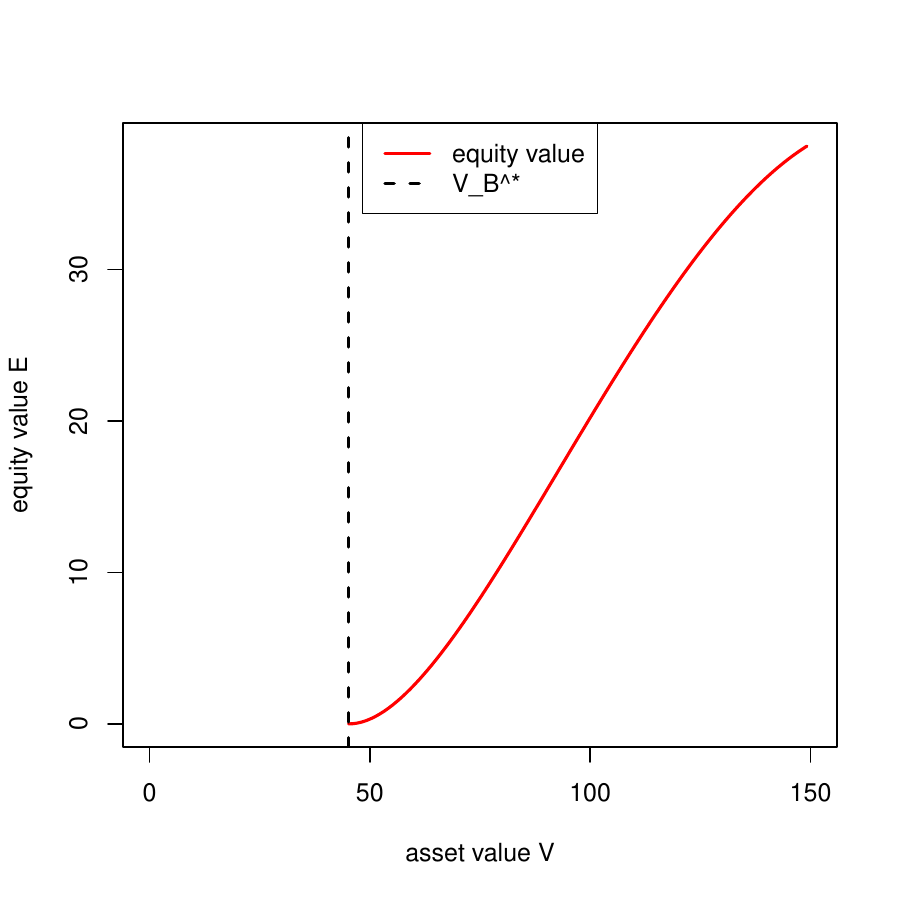}
		\end{minipage}
		\quad
		\begin{minipage}[t]{0.48\textwidth}
			\includegraphics[trim= 1mm 6mm 10mm 20mm, clip,width=\textwidth]{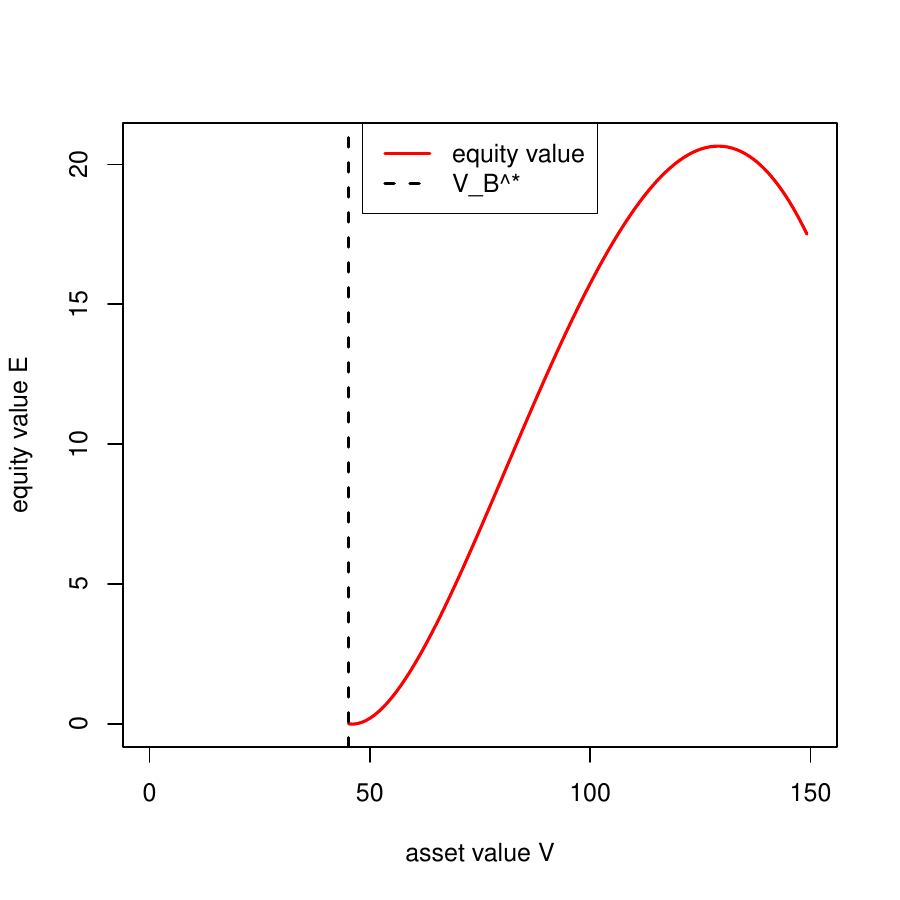}
		\end{minipage}
		\caption{Equity value as a function of the asset value with an $\alpha$ such that $V \to E(V)$ is non-decreasing (left) and if it does not hold (right). In both plots is $V_B$ chosen for the left case.}
		\label{fig: EValphas}
	\end{figure}
	
	In most cases, an analytical solution for the bankruptcy-triggering value $V_B$ as defined in Theorem \ref{th: VB determination} is not available. However, if there is no surplus participation, the proof of Theorem \ref{th: VB determination} (or Leland and Toft \cite{leland1996optimal}) yields the following formula:
	\begin{align} \label{eq: vb alpha=0 main text}
		V_B^{*} = \frac{\tfrac{2(P-\frac{G}{r})A_1}{rT} + 2 \tfrac{G}{r} A_2 - \tau_1 \frac{G}{r}(\lambda_2+\lambda_3)}{1 + \rho(\lambda_2+\lambda_3)+2(1-\rho)A_2},
	\end{align}
	where $\lambda_2,\lambda_3$ are as defined in \eqref{eq: def lambda23}.
	Moreover, if the portfolio parameters or the market situation are such that the bankruptcy-triggering value lies above the threshold for surplus participation, an analytical solution is provided by the following corollary:
	
	\begin{corollary} \label{cor: vb determination}
		Define
		\begin{align}
			\hat{V}_B = \frac{\tfrac{2(P-\frac{G}{r})A_1}{rT} + 2 \tfrac{G}{r} A_2 - \tau_1 \frac{G}{r}(\lambda_2+\lambda_3) + \tau_2\alpha kA_5-\alpha kA_6}{1 + \rho(\lambda_2+\lambda_3)+2(1-\rho)A_2 + \tau_2\alpha A_3 - \alpha A_4}, \label{eq: vb geq k}
		\end{align}
		where $A_1$, $A_2$, $A_3$, and $A_4$ are defined as in \eqref{eq: a1}, \eqref{eq: a2}, \eqref{eq: a3}, and \eqref{eq: a4}, $\lambda_2,\lambda_3$ are as in \eqref{eq: def lambda23}, and 
		\begin{align*}
			A_5 &:= \tfrac{\lambda_2}{r} + \tfrac{1}{\sigma r} \sqrt{\lambda_2^2 \sigma^2 + 2r},\\
			A_6 &:= \tfrac{\lambda_2}{r}(1- 2e^{-r T} \Phi(\lambda_2\sigma\sqrt{T})) + \tfrac{1}{\sigma r} \sqrt{\lambda_2^2 \sigma^2 + 2r}(2\Phi(\sqrt{\lambda_2^2 \sigma^2 + 2r}\sqrt{T})-1).
		\end{align*}
		If $\hat{V}_B \geq k$ and $\alpha<\tilde{\alpha}$, then $\hat{V}_B$ is the unique solution of \eqref{eq: smooth pasting condition} and therefore $V_B = \hat{V}_B$.
	\end{corollary}
	
	\section{Optimal rates} \label{chapter: optimal rates}
	
	In this section, we derive formulas for the optimal participation rate and the optimal guarantee rate. Providing existence results, we begin by fixing one of the two parameters and then conclude with the derivation of the joint optimal values. 
	
	Based on the definition of $\bar{\alpha}$ in \eqref{eq: baralpha}, it follows that $\bar{\alpha} \geq 0$. Throughout this section, we restrict the analysis to the case in which the company does not declare bankruptcy at inception, since in that situation the analysis of optimal participation rates would be redundant. In particular, this assumption ensures that equation \eqref{eq: formula for vb} admits at least one solution (see the proof of Theorem \ref{th: VB determination}) and that the bankruptcy-triggering value $V_B$ is continuous in $g$ and $\alpha$, as the relevant solution is unique (again, see the proof of Theorem \ref{th: VB determination}).

	\subsection{Derivation of the optimal participation rate with a pre-determined guarantee rate} \label{subsection: optimal alpha}
	
	In this subsection, we derive the optimal participation rate $\alpha^*$ when the guarantee rate $g$ is fixed in advance, such that the total insurance company value $v$ in \eqref{eq: firm value v} is maximized. Therefore, we consider $V_B$ as the smallest solution of \eqref{eq: formula for vb} and as a function of $\alpha$. Note that, in general, $V_B(\alpha)$ does not have an explicit form. Only when $\alpha=0$, we obtain an explicit form $V_B(0)= V_B^*$ as in \eqref{eq: vb alpha=0 main text}. Hence, (with a slight abuse of notation) our optimization problem is formulated as follows: We seek the optimal participation rate $\alpha^*$ given by:
	\begin{align} \label{eq: optimization problem alpha}
		\alpha^* = \argmax_{\alpha \in [0,1]} v(V;V_B(\alpha)),
	\end{align}
	where $v$ is defined as in \eqref{eq: firm value v}.
	
	\begin{proposition} \label{prop: alphastar exists}
		There exists an optimal participation rate $\alpha^* \in [0,1]$.
	\end{proposition}
	
	The previous proposition asserts that there is an optimal participation rate $\alpha^*$. The next natural question is under which conditions $\alpha^*>0$, i.e., when is it advantageous for the insurance company to offer contracts with surplus participation? The following theorem provides an answer to this question. 
	
	%
	
	\begin{theorem} \label{th: alpha*}
		There exist $\bar{\tau}, \bar{\bar{\tau}} \in (0,1)$ with $\bar{\bar{\tau}}\leq \bar{\tau}$ such that it is optimal to choose $\alpha^*>0$ if $\tau_2 \in (\bar{\tau},1]$, and it is optimal to choose $\alpha^*=0$ if $\tau_2 \in [0,\bar{\bar{\tau}})$. 
		
		Finally, if the following equation \eqref{eq: solution alpha*} admits a solution in $\alpha$, then this solution is equal to the optimal participation rate $\alpha^*$:
		\begin{align}
			0 =&\, - \tfrac{\tau_1 \frac{G}{r} (\lambda_2+\lambda_3)}{V} (\tfrac{V_B(\alpha)}{V})^{\lambda_2+\lambda_3-1} V_B'(\alpha) - \rho (\lambda_2+\lambda_3+1) (\tfrac{V_B(\alpha)}{V})^{\lambda_2+\lambda_3} V_B'(\alpha) \notag \\
			&+ \tau_2 \int_0^\infty c_{do} (V,k,V_B(\alpha),t) \diff t + \alpha \tau_2 \int_0^\infty \tfrac{\partial c_{do} (V,k,V_B(\alpha),T)}{\partial \alpha} \diff t, \label{eq: solution alpha*}
		\end{align}
		where $\tfrac{\partial c_{do} (V,k,V_B(\alpha),T)}{\partial \alpha}$ is given in \eqref{eq: dalpha cdo<} resp. \eqref{eq: dalpha cdo>} in the Appendix. For an explicit formula of $V_B'(\alpha)$ see \eqref{eq: formula dvb dalpha}.
		%
	\end{theorem}
	
	The existence of a threshold value $\bar{\tau}$, as stated in the theorem, is plausible, as offering surplus participation becomes more attractive when the tax benefit associated with it is higher. The proof of the theorem further provides an equation for determining $\bar{\tau}$ (by finding the zero root of \eqref{eq: dv dalpha alpha0}). Interestingly for $\tau_2 < \bar{\tau}$, it is more advantageous for the insurance company to refrain from offering surplus participation rather than offering a small rate which might explain why many contracts do not have any participation element. 
	
	From our numerical analysis, we observe that, in general, $\bar{\bar{\tau}}=\bar{\tau}$, $\tau_2>\bar{\tau}$ (so that $\alpha^*>0$) and that $\alpha^*$ increases with $\tau_2$. However, if the final lump sum payment $P$ or the guaranteed payment $G$ are too high, $\bar{\tau}$ may actually exceed $\tau_2$ and no participation is offered to policyholders (i.e., $\alpha^*=0$). The reason is that higher lump sum payments or a larger guarantee rate result in more costly liabilities, while a lower tax rate reduces the tax benefit and thus decreases equity. Consequently, the contract's liabilities must not be too expensive compared to the equity to ensure that a positive participation rate remains optimal. However, typically when optimizing $G$, the resulting liabilities do not impose an excessive cost on equity, which reinforces that $\tau_2>\bar{\tau}$ usually holds.

	\subsection{Derivation of the optimal guarantee rate with a pre-determined participation rate} \label{subsection: optimal g}
	
	In this subsection, we derive the optimal guarantee rate $g^*$ when the participation rate $\alpha$ is fixed in advance, such that the total insurance company value $v$ (as defined in \eqref{eq: firm value v}) is maximized. Therefore, we consider $V_B$ as the smallest solution of \eqref{eq: formula for vb} and as a function of $g$. Note that, in general, $V_B(g)$ does not have an explicit form. Also, recall that $g = \frac{G}{T}$. Hence, our optimization problem (with a slight abuse of notation) is as follows: We seek the optimal guarantee rate $g^*$ given by:
	\begin{align} \label{eq: optimization problem g}
		g^* = \argmax_{g \in [0,\infty)} v(V;V_B(g)),
	\end{align}
	where $v$ is defined in \eqref{eq: firm value v}.

	\begin{proposition} \label{prop: gstar exists}
		There exists an optimal guarantee rate $g^* \in [0,\infty)$.
	\end{proposition}
	
	Next, we address the question of which conditions ensure that $g^*>0$, i.e., when contracts with a guarantee rate are better for the insurance company than those without such guarantees? This question is answered in the upcoming theorem.
	
	\begin{theorem} \label{th: g*}
		Under the technical Assumption~\ref{ass: g*>0} in the Appendix it is optimal to choose $g^*>0$, i.e., it is optimal to provide a contract with a positive guarantee rate. 
		
		Moreover, if the following equation \eqref{eq: solution g*} admits a solution for $g$, then that solution is the optimal guarantee rate $g^*$:
		\begin{align}
			0 =&\, \tfrac{\tau_1 T}{r} (1-(\tfrac{V_B}{V})^{\lambda_2+\lambda_3})  - \tfrac{\tau_1 g T (\lambda_2+\lambda_3)}{Vr} (\tfrac{V_B}{V})^{\lambda_2+\lambda_3-1} V_B'(g) - \rho (\lambda_2+\lambda_3+1) (\tfrac{V_B}{V})^{\lambda_2+\lambda_3} V_B'(g) \notag \\
			&+ \alpha \tau_2 \int_0^\infty \tfrac{\partial c_{do} (V,k,V_B,T)}{\partial g} \diff t. \label{eq: solution g*}
		\end{align}
		where $\tfrac{\partial c_{do} (V,k,V_B(g),T)}{\partial g}$ is as in \eqref{eq: dg cdo<} resp. \eqref{eq: dg cdo>} in the Appendix. For an explicit formula of $V_B'(g)$ see \eqref{eq: formula dvb dg}.
	\end{theorem}
	
	Under some technical conditions, equation \eqref{eq: solution g*} always admits a solution (see Proposition \ref{prop: solution g* existence}). 
	To conclude this subsection, we discuss the numerical observations regarding Assumption \ref{ass: g*>0} in more detail:
	
	As in the previous Subsection \ref{subsection: optimal alpha}, our numerical analysis also shows that Assumption \eqref{eq: assumption optimal g} is typically fulfilled when optimizing $G$. However, if the lump sum payment $P$ is too high, or if the tax rate $\tau_1$ is too low, the assumption might not hold. This suggests that when the liabilities become too costly in comparison to equity, it becomes difficult to provide additional promises to policyholders. Additionally, while an arbitrary high participation rate $\alpha$ could theoretically cause \eqref{eq: assumption optimal g} to fail, the restriction $\alpha \in [0,\bar{\alpha}]$ ensures that in most situations, this range is not sufficient for the assumption to be violated.
	
	\subsection{Derivation of the optimal guarantee rate and participation rate}
	
	In this subsection, we derive the optimal participation rate $\alpha^*$ and the optimal guarantee rate $g^*$ simultaneously, such that the total insurance company value $v$ (as defined in \eqref{eq: firm value v}) is maximized. Therefore, we consider $V_B$ as the smallest solution of \eqref{eq: formula for vb}, and make its dependence on $(\alpha,g)$ explicit. Note that, in general, $V_B(\alpha,g)$ does not have an explicit form. Hence, our optimization problem is then (with a slight abuse of notation) formulated as follows: We seek the optimal rate vector $(\alpha^*,g^*)$ given by:
	\begin{align} \label{eq: optimization problem alpha g}
		(\alpha^*,g^*) = \argmax_{(\alpha,g) \in [0,1] \times [0,\infty)} v(V;V_B(\alpha,g)),
	\end{align}
	where $v$ is defined in \eqref{eq: firm value v}.
	
	\begin{proposition} \label{prop: alphastar and gstar exists}
		There exists an $(\alpha^*,g^*) \in [0,1]\times[0,\infty)$ which is the optimal pair of rates for $(\alpha,g) \in [0,1]\times[0,\infty)$.
	\end{proposition}
	
	\begin{proposition} \label{prop: alpha* and g*}
		If the non-linear equation system consisting of equations \eqref{eq: solution alpha*} and \eqref{eq: solution g*} admits a solution, then the solution is given by $(\alpha^*,g^*)$.
	\end{proposition}
	
	This proposition directly follows from Theorems \ref{th: alpha*} and \ref{th: g*}. From these theorems, we also obtain sufficient conditions for ensuring that $\alpha^*>0$ and $g^*>0$.
	

	\section{Numerical Results} \label{chapter: numerical}
	
	In this section we show that the asset subsititution effect is largely diminished when incorporating the possibility to issue surplus participation. 
	We also provide examination of the optimal participation rate $\alpha^*$ and the optimal guarantee rate $g^*$ and show that the assumptions for Theorems \ref{th: alpha*} and \ref{th: g*} are satisfied in typical numerical examples.
	
	For this analysis, we use a basic setting for both the financial market and the contract conditions. Unless stated otherwise, the parameters take the following values: For the financial market, we use a risk-free interest rate $r=1\%$, a dividend rate $\nu = 5\%$ and a volatility of $\sigma=20\%$. The contract length is $T=30$ with lump sum payment $P=95$, guarantee rate $\tfrac{G}{P}=2\%$, initial asset value $V_0=100$, and a surplus participation starting at $k=150$ with participation rate $\alpha=5\%$, i.e., we have the liability rate $\tfrac{P}{V_0} = 95 \%$ and the surplus initiation rate $\tfrac{k}{V_0} = 150 \%$. The tax rates are $\tau_1=35\%=\tau_2$, and we use a loss fraction at bankruptcy of $\rho=50\%$. These values for the lump sum payment and the dividend are typical for large insurance companies. The high liability capital is also not uncommon for life insurance companies, see, for instance the balance sheets of \cite[p.150]{Allianzreport} and \cite[p.30]{AllianzLifereport}. The tax rates and the loss fraction are taken from Leland and Toft \cite{leland1996optimal}, whereas Chen et al. \cite{chen2019constrained} utilized slightly lower values ($20 \%$ or $25 \%$). It is worth noting that, as discussed, e.g., by Kling et al. \cite{kling4744873intertemporal,kling2024intertemporal}, insurance companies commonly employ return smoothing mechanisms in their payments to policyholders. This can be modeled by reducing the volatility $\sigma$ to $50-75 \%$ of its original value. As shown in the sensitivity analysis in Figure \ref{fig: sensi}, this would significantly increase the surplus participation rate. However, we maintain $\sigma=20\%$, as this effect is not typically considered in the majority of the related literature.
	
	We have checked the pre-conditions from Theorems \ref{th: VB determination}, \ref{th: alpha*}, and \ref{th: g*} for all cases presented, and they are satisfied with $\bar{\alpha} \approx  0.12$ and $\bar{\tau}=\bar{\bar{\tau}}$ in the basic setting. However, $\bar{\tau}$ exceeds $\tau_2$, i.e., $\alpha^*=0$, if, ceteris paribus, the lump sum payment value or the guarantee rate becomes too high ($\tfrac{P}{V_0}\geq 150 \%$ resp. $\tfrac{G}{P} \geq 7 \%$), or if the tax rate on the participation is too low ($\tau_2 \leq 8\%$). Similarly, equation \eqref{eq: assumption optimal g} from Assumption \ref{ass: g*>0} does not hold, i.e., $g^*=0$, if, ceteris paribus, the lump sum payment value gets too high ($\tfrac{P}{V_0}\geq 260 \%$) or if the tax rate on the guaranteed payment is too low ($\tau_1 \leq 0.1\%$). These results align with the discussions at the end of Subsections \ref{subsection: optimal alpha} and \ref{subsection: optimal g}. Changes in other parameters, however, generally maintain the two conditions, ensuring that $\alpha^*>0$ resp. $g^*>0$. For these parameters, we provide an overview in Figure \ref{fig: test}, where we plot the results when varying the parameters individually. The left (resp. middle) plot shows the region such that $\alpha^*>0$ (resp. $g^*>0$) holds, depending on the liability ratio $\tfrac{P}{V_0}$ and the the tax benefit $\tau=\tau_1=\tau_2$. The right plot shows the region such that $\alpha^*>0$, depending on the liability ratio $\tfrac{P}{V_0}$ and the guarantee rate $\tfrac{G}{P}$. Since we vary $G$ in the right plot, it is not meaningful to impose the condition $g^*>0$ in this context. In all plots, a white square indicates that the respective optimal rate is positive, whereas a black square indicates that it is zero in the optimal case. All other parameters, including the respective other optimal rate, are fixed as in the basic parametrization. From the plots, we confirm the statements made in the discussions at the end of Subsections \ref{subsection: optimal alpha} and \ref{subsection: optimal g}, which suggest that the optimal rates remain positive as long as the liabilities are not excessively costly relative to the equity value. This indicates that it becomes challenging to offer additional promises to policyholders when the original promises, like the lump sum payment, are already too costly compared to the receiving tax benefit.
	
	
	\begin{figure}[!htb]
		\centering
		\begin{minipage}[t]{0.3\textwidth}
			\includegraphics[trim= 1mm 6mm 10mm 8mm, clip,width=\textwidth]{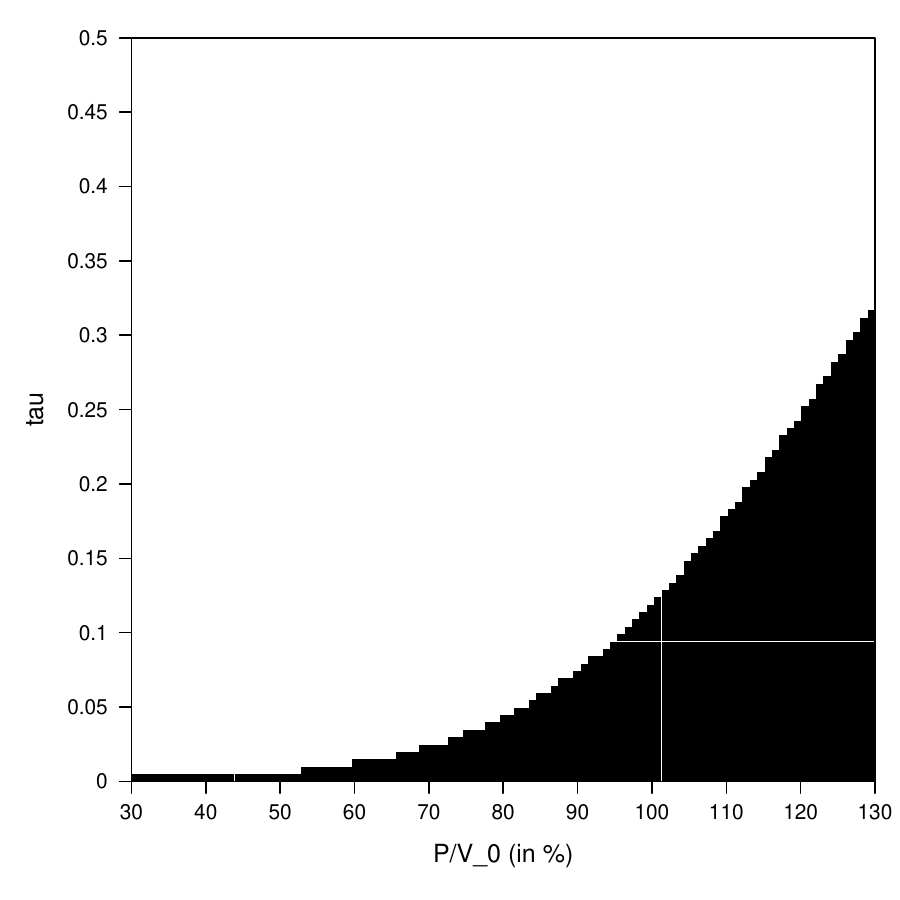}
		\end{minipage}
		\quad
		\begin{minipage}[t]{0.3\textwidth}
			\includegraphics[trim= 1mm 6mm 10mm 8mm, clip,width=\textwidth]{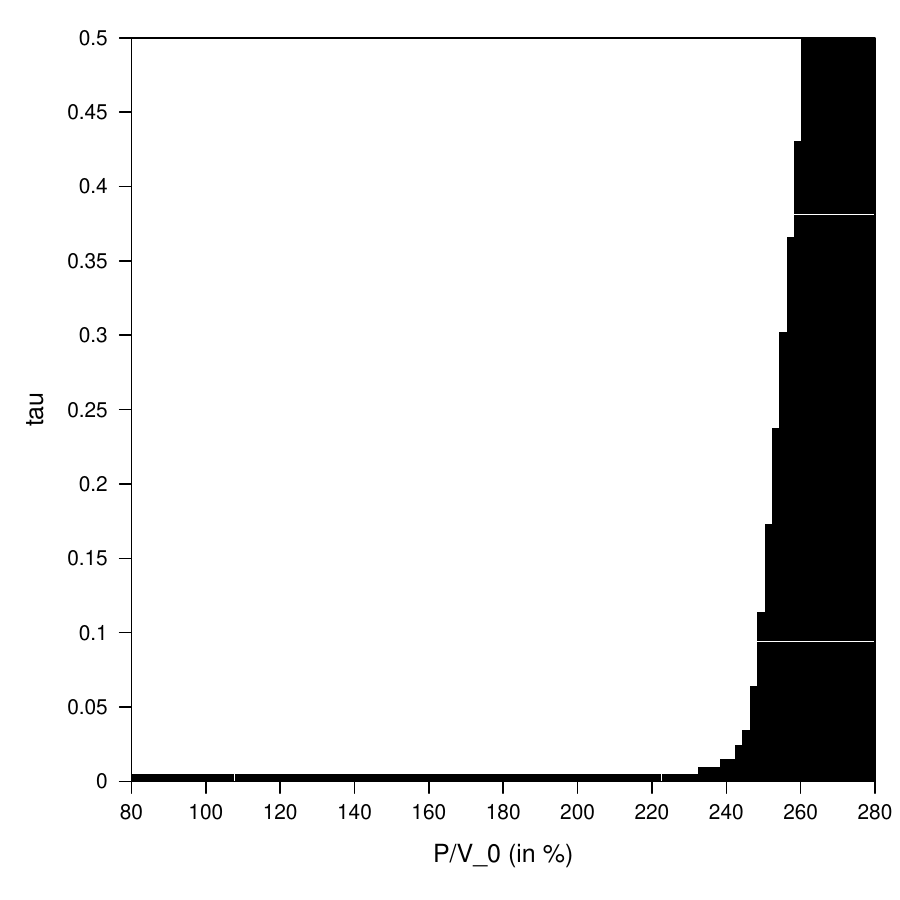}
		\end{minipage}
		\quad
		\begin{minipage}[t]{0.3\textwidth}
			\includegraphics[trim= 1mm 6mm 10mm 8mm, clip,width=\textwidth]{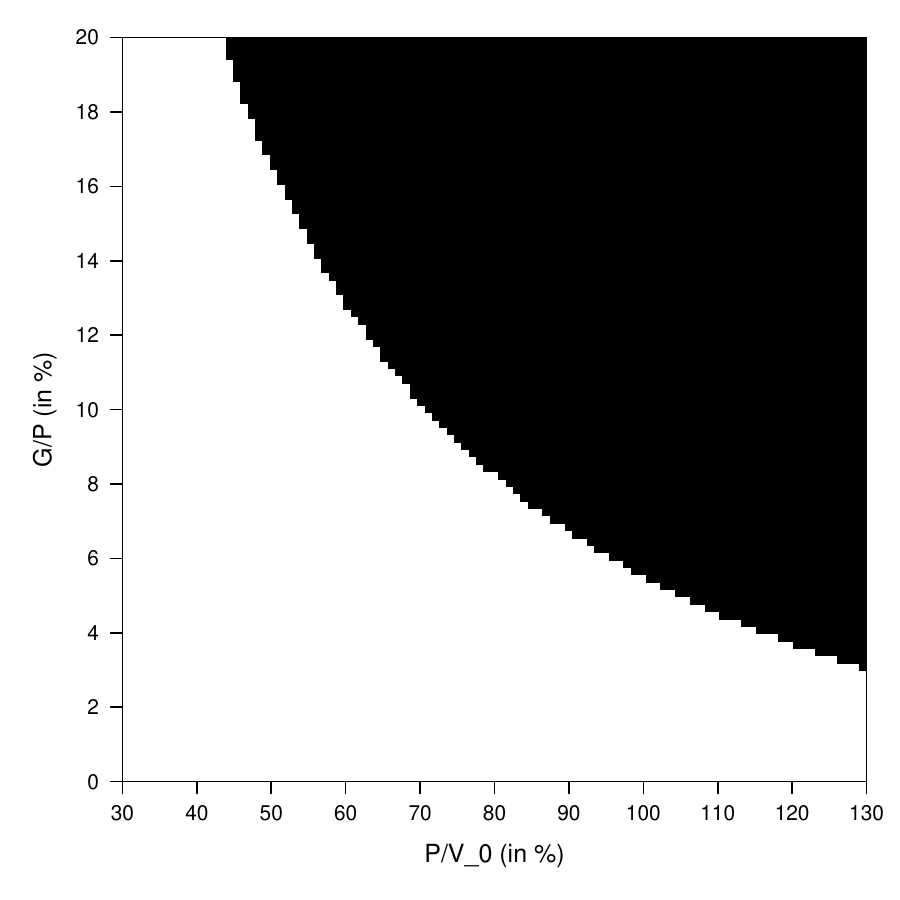}
		\end{minipage}
		\caption{Scatterplot if the optimal rates are positive as functions of the liability ratio $\tfrac{P}{V_0}$ (in \%) and the tax benefit $\tau=\tau_1=\tau_2$ (left) resp. the guarantee rate $\tfrac{G}{P}$ (in \%) (right). A white square indicates that $\alpha^*>0$ in the left and the right plot resp. that $g^*>0$ holds in the middle plot.}
		\label{fig: test}
	\end{figure}
	
	In Figure \ref{fig: VB variation}, we illustrate how the bankruptcy-triggering value $V_B$ varies as a function of the participation rate $\alpha$ (left plot) and the guarantee rate $\tfrac{G}{P}$ (right plot). For the left plot, we restrict the participation rate $\alpha<\bar{\alpha}$, ensuring that a solution to equation \eqref{eq: formula for vb} exists, i.e., the bankruptcy-triggering value is not set to $V_0$ which leads to immediate bankruptcy. In both plots, we observe that an increase in either the participation rate or the guarantee rate results in a higher bankruptcy-triggering value $V_B$. The reason is that both higher participation rates and higher guarantee rates to policyholders increase the total liabilities, which, in turn, brings the equity holders to default earlier in order not to have negative equity. Additionally, we notice that the bankruptcy-triggering value is more sensitive to changes in the guarantee rate than to changes in the participation rate. This can be explained by the fact that the guarantee rate is a fixed obligation, meaning it must be paid regardless of asset performance. In contrast, the participation rate only affects payments when the asset value exceeds a certain threshold (e.g., in the basic parametrization starting at $\tfrac{k}{V_0}=150 \%$). Given that the probability of surpassing this threshold is rather low when asset values are close to $V_B$, the participation rate, in this case, has less impact on the bankruptcy-triggering value. From the right plot, we can also observe that when the accumulated guarantee rate $\tfrac{G}{P}$ reaches approximately $11.1 \%$, the bankruptcy-triggering value $V_B$ exceeds the initial value $V_0 = 100$. In this case, equity holders would be forced to declare bankruptcy immediately, indicating that the guarantee rate has been set unfeasibly high.
	
	\begin{figure}[!htb]
		\centering
		\begin{minipage}{0.48\textwidth}
			\includegraphics[trim= 1mm 6mm 7mm 8mm, clip,width=\textwidth]{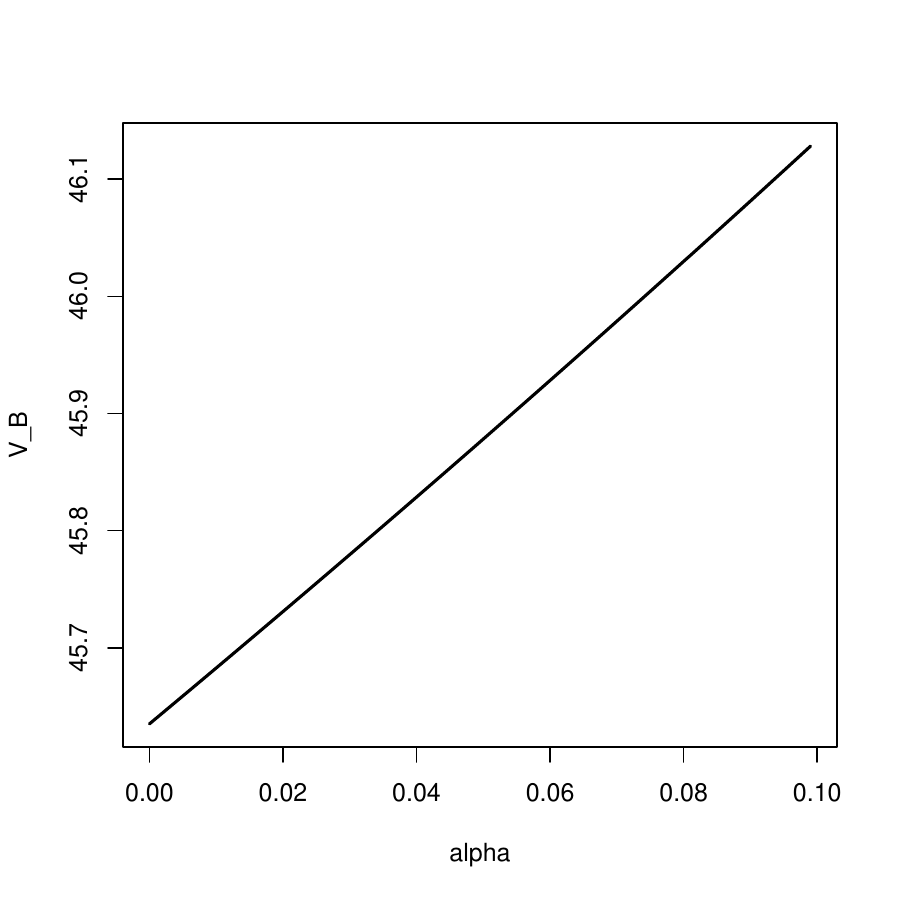}
		\end{minipage}
		\quad
		\begin{minipage}{0.48\textwidth}
			\includegraphics[trim= 1mm 6mm 7mm 8mm, clip,width=\textwidth]{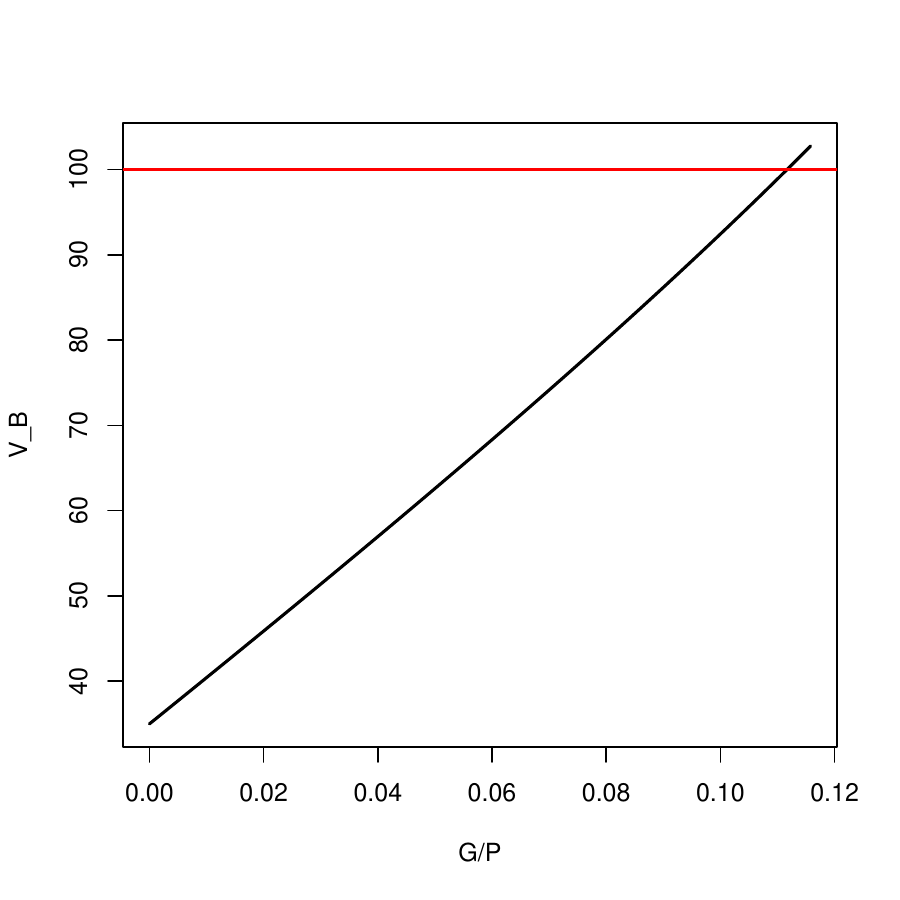}
		\end{minipage}
		\caption{Variation of the bankruptcy-triggering value $V_B$ for different values of the participation rate $\alpha$ (with $\alpha<\bar{\alpha}$) and the guaranteed payment rate $\tfrac{G}{P}$.}
		\label{fig: VB variation}
	\end{figure}
	
	In Figure \ref{fig: G}, we show the value of the insurance company $v$, the equity value $E$, and the liability value $L$ as functions of the guaranteed payment $G$. From the plot, we see that in the basic parametrization, the optimal guarantee rate is approximately $\tfrac{G^*}{P} \approx 1.91\%$ when optimizing for the insurance company value $v$. However, this optimal value does not correspond to the optimal guarantee rate from either the policyholder's or the equity holder's perspectives. This finding aligns with the result in the case of no surplus participation, as discussed by Lando \cite{lando1998cox}.
	
	In Figure \ref{fig: vED}, we present the insurance company value $v$, the equity value $E$, and the liability value $L$ as functions of the asset value $V$, for the optimal values $\alpha^*$, $G^*$, and $V_B^*$ in the basic parametrization. 
	We find an optimal participation rate of $\alpha^* \approx 0.099$, an optimal guarantee rate $\tfrac{G^*}{P} \approx 1.91\%$, and the corresponding bankruptcy-triggering value $V_B^* \approx 45.36$. From the plots in Figure \ref{fig: vED}, we observe that all three values, the insurance company value $v$, the equity value $E$ and the liability value $L$, increase with the asset value $V$, which is expected. However, it is important to note that the increase in the insurance company value is not linear. As the asset value $V$ increases, the growth in the company value starts to decrease slightly. 
	Additionally, we see that the equity value $E$ becomes zero at $V_B^*$, which is consistent with the smooth-pasting condition. 
	
	When comparing these results to Lando \cite{lando1998cox}, who analyzed a model without participation, we observe that the equity value $E$ is no longer convex with respect to the asset value. 
	This change is due to the participation costs becoming more significant at higher values of $V$. 
	As a result, the rate of increase of $E$ in $V$ decreases, and $E$ adopts a concave form as $V$ grows. 
	Moreover, in the absence of convexity, the "option-like" nature of equity diminishes. 
	This is crucial for mitigating the asset substitution effect, which will be discussed in greater detail later.
	
	\begin{figure}[!htb]
		\centering
		\begin{minipage}[t]{0.48\textwidth}
			\includegraphics[trim= 1mm 6mm 10mm 8mm, clip,width=\textwidth]{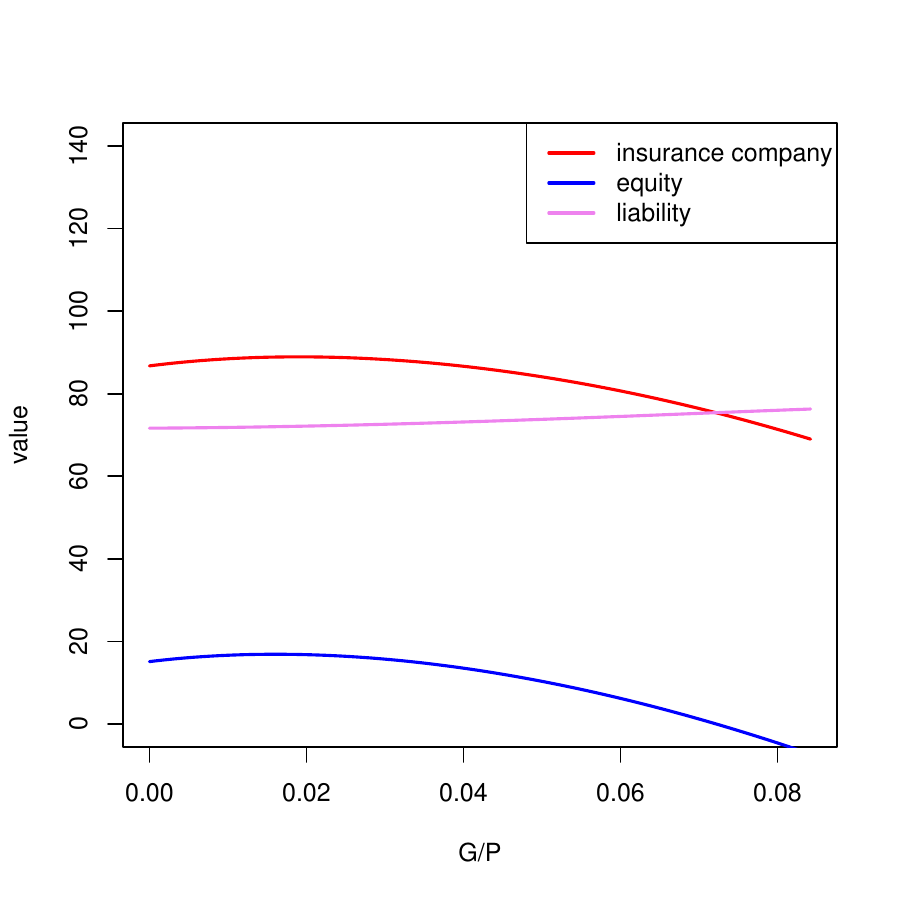}
			\caption{Insurance company value $v$, liability value $L$, and equity value $E$ as a function of the guarantee rate $\tfrac{G}{P}$ in the basic parametrization.}
			\label{fig: G}
		\end{minipage}
		\quad
		\begin{minipage}[t]{0.48\textwidth}
			\includegraphics[trim= 1mm 6mm 10mm 8mm, clip,width=\textwidth]{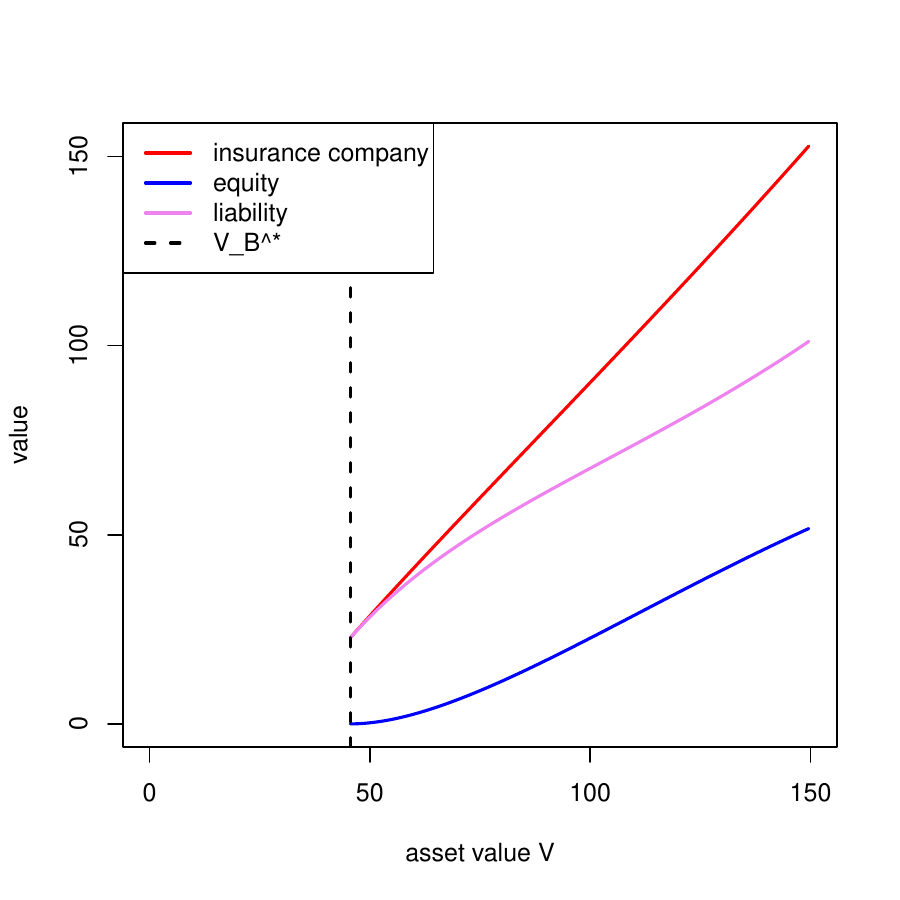}
			\caption{Value of the insurance company, equity and liability as a function of the asset value $V$. Here, we used the optimal participation rate $\alpha^* \approx 9.9 \%$ and the optimal guarantee rate $\tfrac{G^*}{P} \approx 1.91 \%$ to determine $\tfrac{V_B^*}{V_0} \approx 45.36 \%$.}
			\label{fig: vED}
		\end{minipage}
	\end{figure}
	
	In Figure \ref{fig: sensi}, we present some results of a sensitivity analysis on the optimal participation rate $\alpha^*$. The plots display the effects of variations in the dividend rate $\nu$ on the left, the contract duration $T$ in the middle, and the tax rate $\tau_2$ on the right. We find that the optimal participation rate is strongly influenced by all three parameters. As the dividend rate increases, the optimal participation rate also increases. This is economically reasonable, as a higher dividend rate offered by the insurance company diminishes the long-term performance of the company's asset process, thereby making it more cost-effective for the insurer to offer a higher participation rate. On the other hand, a longer contract duration decreases the optimal participation rate. This occurs because the longer duration increases the likelihood of experiencing a high surplus participation, assuming positive expected returns over time. Lastly, an increase in the tax rate $\tau_2$ results in higher optimal participation rates. This is intuitive, as higher tax rates enhance the value of equity, thus making participation more attractive. Conversely, the effect of the tax rate $\tau_1$ on the optimal participation rate is minimal.
	
	\begin{figure}[!htb]
		\centering
		\begin{minipage}{0.3\textwidth}
			\includegraphics[trim= 1mm 6mm 10mm 18mm, clip,width=\textwidth]{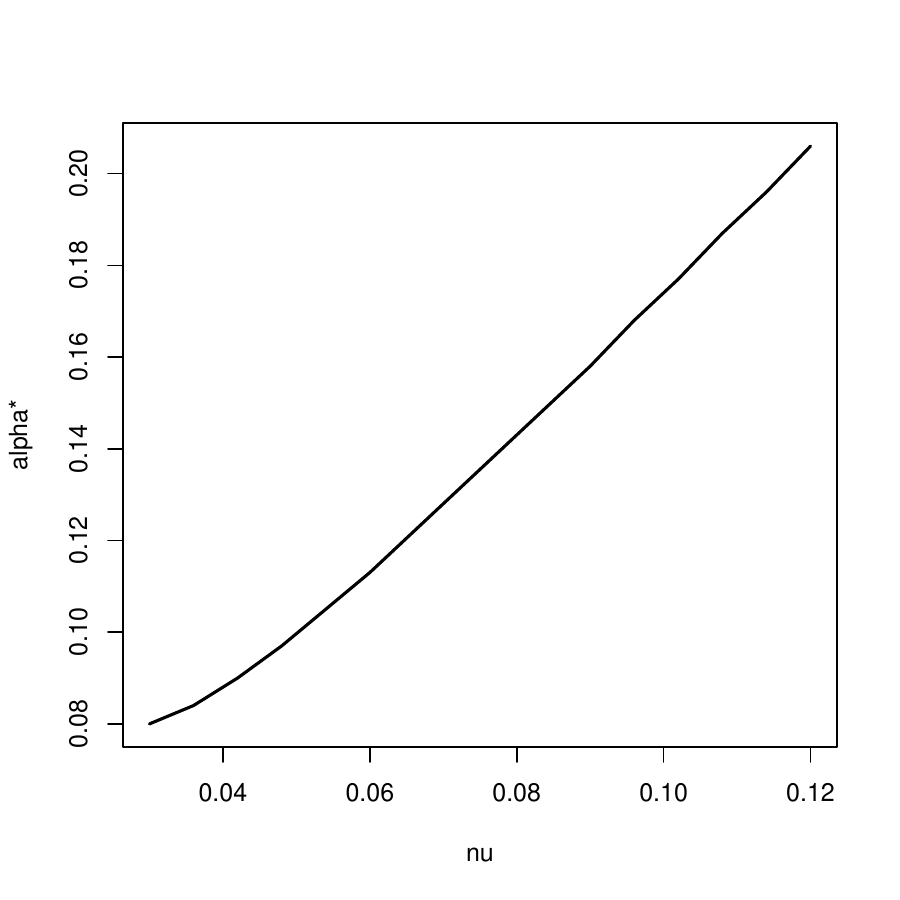}
		\end{minipage}
		\quad
		\begin{minipage}{0.3\textwidth}
			\includegraphics[trim= 1mm 6mm 10mm 18mm, clip,width=\textwidth]{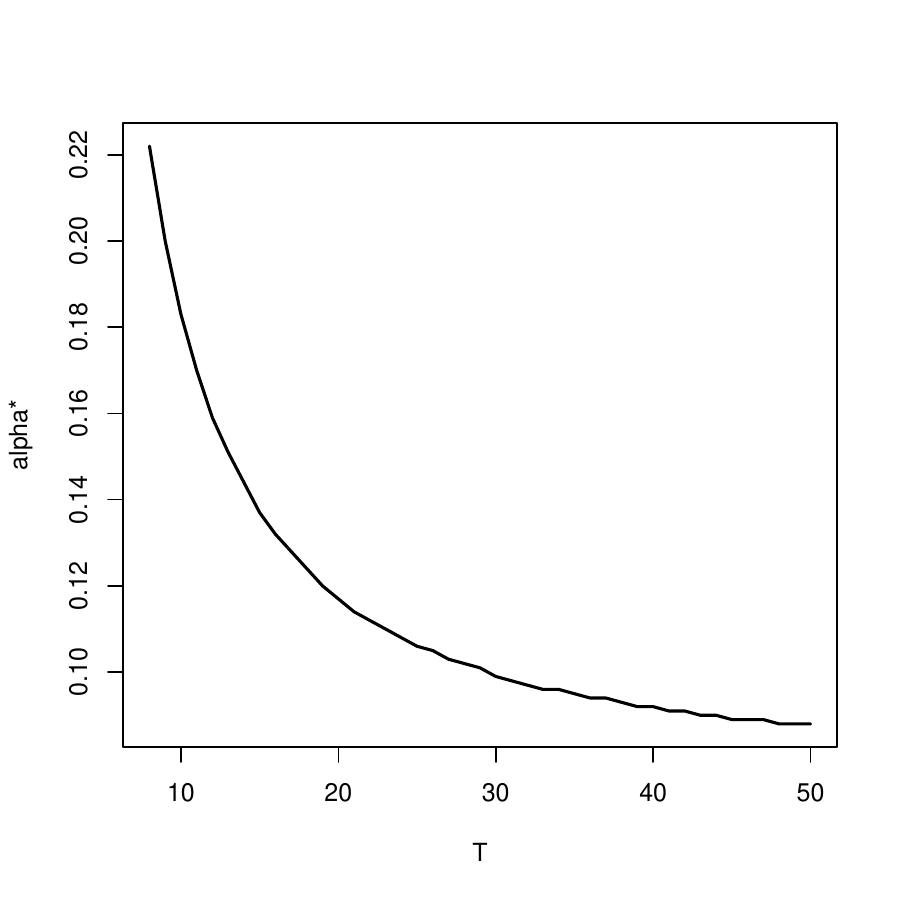}
		\end{minipage}
		\quad
		\begin{minipage}{0.3\textwidth}
			\includegraphics[trim= 1mm 6mm 10mm 18mm, clip,width=\textwidth]{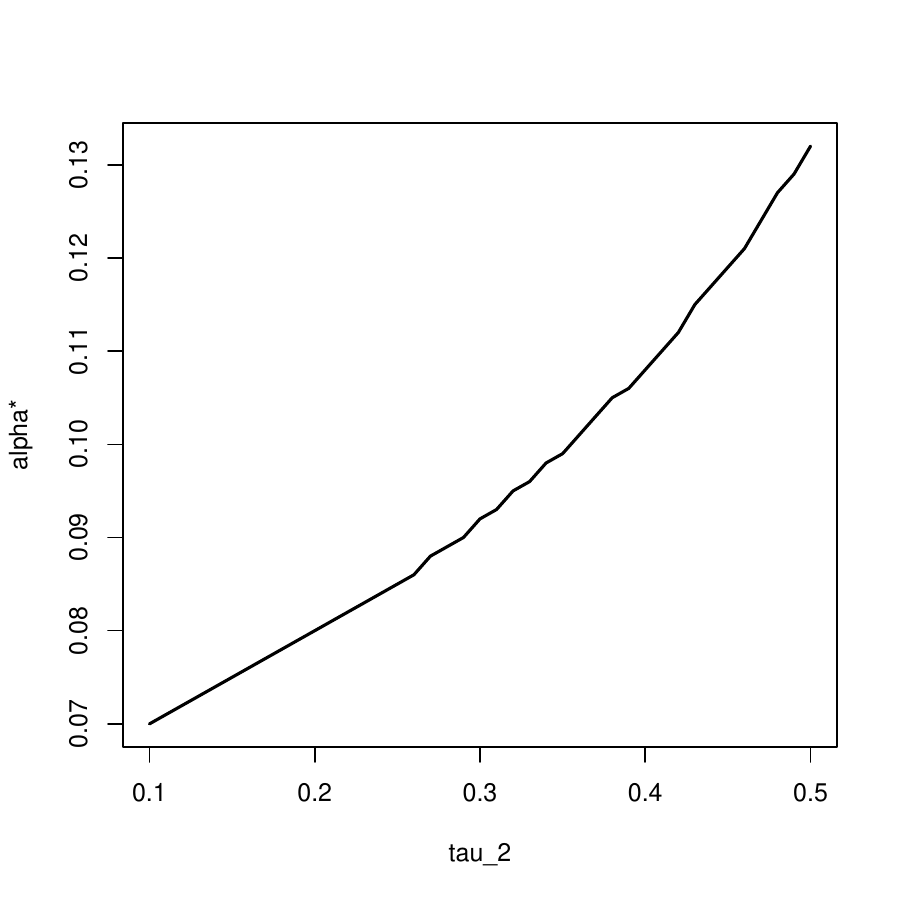}
		\end{minipage}
		\caption{Sensitivity analysis of the optimal participation rate $\alpha^*$ for the dividend rate $\nu$ (left), the contract duration $T$ (middle), and the tax rate $\tau_2$ (right).}
		\label{fig: sensi}
	\end{figure}
	
	In the concluding paragraph of this numerical analysis, we examine the asset substitution effect, which describes the tendency of equity holders to increase the riskiness of a company's investment decisions, leading to a transfer of value from liabilities to equity. Figure \ref{fig: assetsub} illustrates the partial derivatives of equity and liability with respect to asset volatility, across different surplus participation rates (top row) and contract durations (bottom row). The asset substitution effect appears in regions where $\tfrac{\partial}{\partial \sigma} L <0$ and $\tfrac{\partial}{\partial \sigma} E >0$, meaning equity holders seek to increase risk, while policyholders seek to reduce it. In the absence of participation, i.e., $\alpha = 0 \%$, we confirm previous findings (see the references in the introduction) that there is an asset substitution effect in a large region (starting at $75$ \% of the initial asset value for our parametrization). When surplus participation is introduced, transferring some of the incentives for risk-taking to policyholders, the asset substitution effect vanishes for a reasonable contract duration. However, as the contract duration increases (particularly beyond the lifespan of multiple generations), the influence of surplus participation diminishes, and the asset substitution effect reverts to the case without surplus participation, where the asset substitution effect is present. This last result is plausible, as a longer contract duration delays surplus payments to policyholders, and in the limit ($T = \infty$), no surplus payment occurs at finite time points. Furthermore, Leland and Toft \cite{leland1996optimal} observe that even in the absence of surplus participation, longer maturities exacerbate the asset substitution effect. They contend that, although the option analogy (presented in the introduction) may not be entirely accurate, the adverse incentives linked to longer maturities are indeed magnified. The impacts of parameter changes align with the case of no participation. 
	
	\begin{figure}[!htb]
		\centering
		\begin{minipage}{0.3\textwidth}
			\includegraphics[trim= 1mm 6mm 10mm 8mm, clip,width=\textwidth]{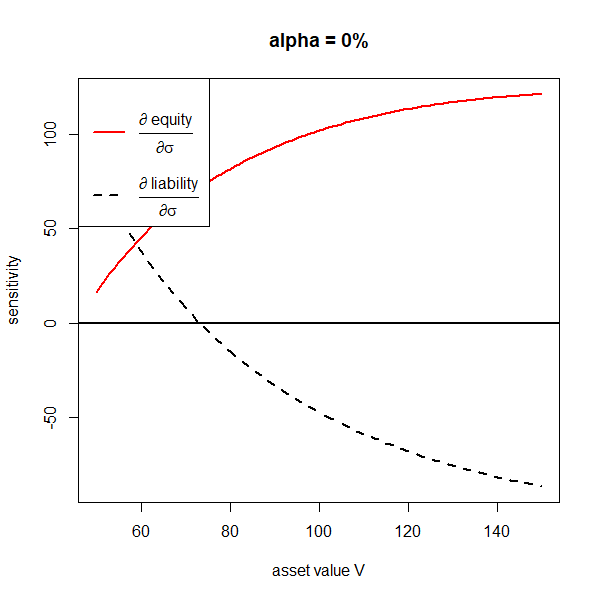}\\[1ex]	
			
			\includegraphics[trim= 1mm 6mm 10mm 8mm, clip,width=\textwidth]{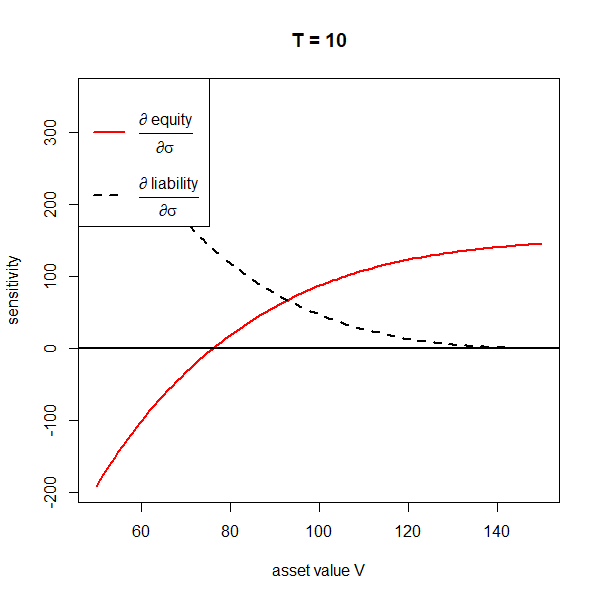}
		\end{minipage}
		\quad
		\begin{minipage}{0.3\textwidth}
			\includegraphics[trim= 1mm 6mm 10mm 8mm, clip,width=\textwidth]{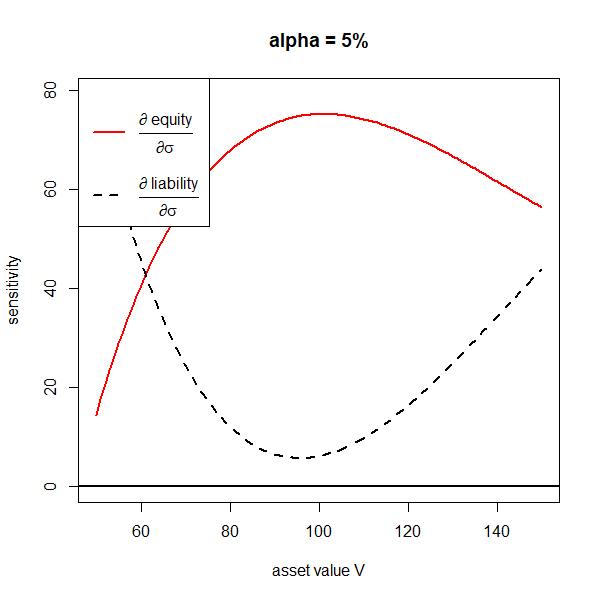}\\[1ex]
			
			\includegraphics[trim= 1mm 6mm 10mm 8mm, clip,width=\textwidth]{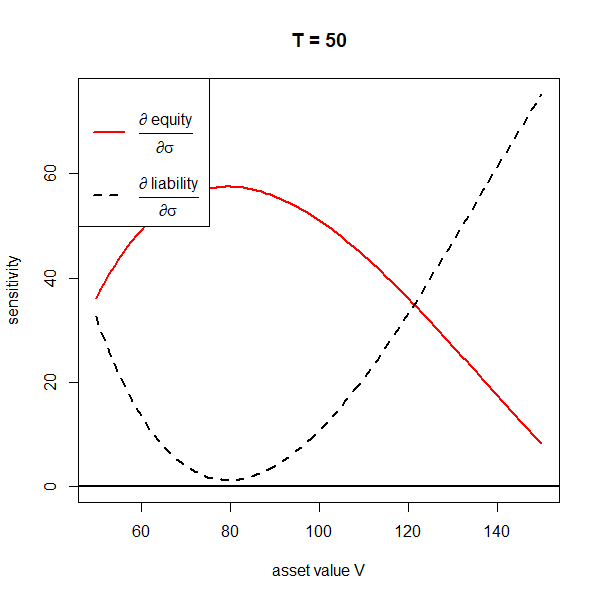}
		\end{minipage}
		\quad
		\begin{minipage}{0.3\textwidth}
			\includegraphics[trim= 1mm 6mm 10mm 8mm, clip,width=\textwidth]{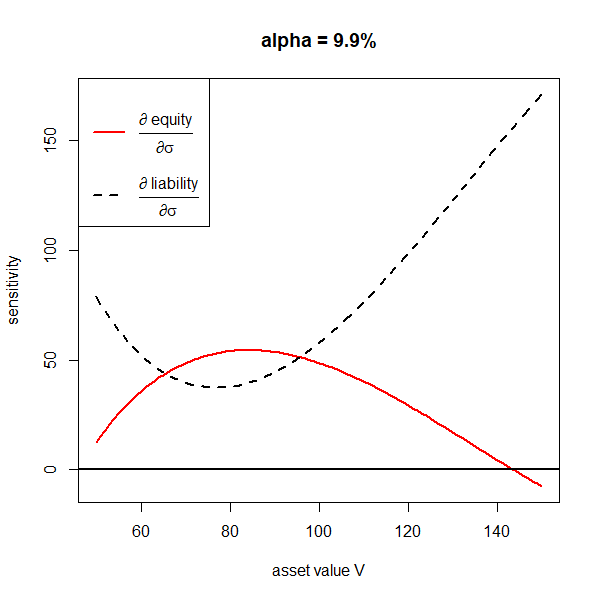}\\[1ex]
			
			\includegraphics[trim= 1mm 6mm 10mm 8mm, clip,width=\textwidth]{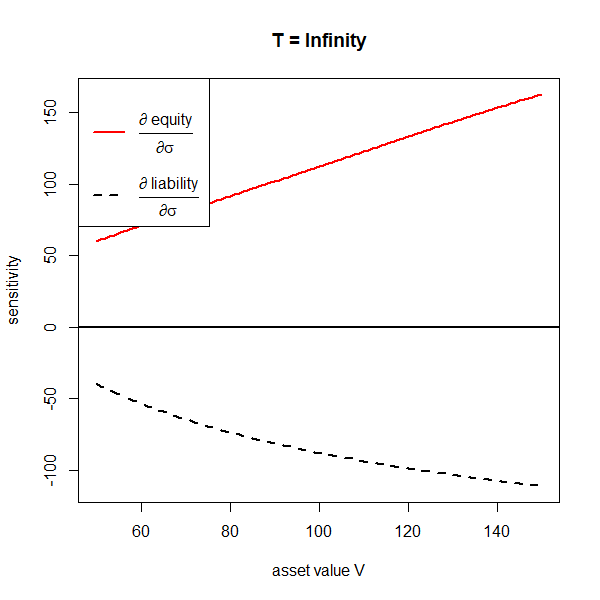}
		\end{minipage}
		\caption{Effect of an increase in the volatility on equity and liabilities when varying the contract duration (top row) and when varying the surplus participation rate (bottom row). The lines show the partial derivative with respect to the volatility.}
		\label{fig: assetsub}
	\end{figure}
	\FloatBarrier
		\section{Conclusion} \label{chapter: conclusion}

In this paper, we develop a dynamic capital structure model of life insurance companies with surplus participation to explain the endogenous emergence of hybrid contracts that combine guarantee and participation features. We show that surplus participation helps mitigate the asset substitution effect but is not always optimal. In contrast to prominent Leland-type models, in which the largest economically admissible root is typically selected, we show that the optimal bankruptcy trigger in our setting is given by the smallest economically admissible root. We derive formulas for the optimal participation rate and the optimal guarantee rate. Our numerical analysis demonstrates that the optimal participation rate is particularly sensitive to contract maturity and the tax rate associated with surplus participation.

	
	\section*{Funding}
	
	This research did not receive any specific grant from funding agencies in the public, commercial, or not-for-profit sectors.
	
	\section*{Declaration of generative AI and AI-assisted technologies in the writing process}
	During the preparation of this work, the authors used ChatGPT (OpenAI) to assist with language editing, wording, and restructuring of selected passages in the Introduction and Conclusion, and to assist with a preliminary analysis of the monotonicity of \(\eta\) in the proof of Lemma \ref{lemma: uniqueness preparation lemma} in the Appendix. The underlying research questions, model, analytical results, economic interpretations, and conclusions were developed by the authors. All AI-assisted content was independently reviewed and verified by the authors, who take full responsibility for the final manuscript.

	\newpage
	\setcounter{section}{0}
	\renewcommand{\thesection}{\Alph{section}}
	\appendix
	\section*{Definition of variables}
		\begin{table}[h!]
		\scriptsize
		\centering
		\caption{Overview of used variables}
		\label{tab:variables}
		\begin{tabular}{p{1.5cm}|p{2.5cm}|p{9cm}}
			\textbf{Variable} & \textbf{Definition} & \textbf{Intuition} \\ \hline\hline
			$G$ & Section~\ref{chapter: analysis} & total amount of guaranteed payments per year \\
			$P$ & Section~\ref{chapter: analysis} & total amount of lump payments at maturity \\
			$g$ & Section~\ref{subsection: participating life insurance} & guarantee rate \\
			$p$ & Section~\ref{chapter: analysis} & lump sum payment rate \\
			$E$ & Equation~\eqref{equity} & equity value of the insurance company \\
			$\QQ$ & Subsection~\ref{subsection: model setup} & risk-neutral measure \\
			$\PP$ & Subsection~\ref{subsection: model setup} & real-world measure \\
			$V_t$ & Subsection~\ref{subsection: barrier}
			& insurance company's asset value \\
			$V_0$ & Subsection~\ref{subsection: barrier} & initial asset value \\
			$r$ & Subsection~\ref{subsection: model setup} & risk-free interest rate \\
			$\mu$ & Subsection~\ref{subsection: model setup} & expected rate of return of the asset value \\
			$\sigma$ & Subsection~\ref{subsection: model setup} & volatility of the insurance company's asset returns \\
			\(\{W_t\}_{t \ge 0}\) & Subsection~\ref{subsection: model setup} & Brownian motion (noise generator) of $V_t$ \\
			$\alpha$ & Subsection~\ref{subsection: participating life insurance} & participation rate \\
			$V_{B}$ & Subsection~\ref{subsection: barrier}, Equation~\eqref{equity} & barrier value triggering bankruptcy \\
			$c_{do}$ & Subsection~\ref{subsection: barrier} & value of the Down-and-Out Call option corresponding to the surplus participation \\
			\(\bar{\alpha}\) & Equation~\eqref{eq: baralpha} & upper bound for participation rates such that equity stays non-negative \\
			$\rho$ & Subsection~\ref{subsection: model setup} & fraction of the asset value that is lost in the event of bankruptcy \\
			$v$ & Subsection~\ref{subsection: model setup} & total value of the insurance company \\
			$\tau_1$ & Subsection~\ref{subsection: model setup} & tax rate on the guaranteed payment \\
			$\tau_2$ & Subsection~\ref{subsection: model setup} & tax rate on the surplus participation \\
			$\lambda_1$ & Equation~\eqref{eq: def d12 lambda1} & excess return adjusted for volatility, also used in the Black–Scholes model \\
			$\lambda_2$ & Equation~\eqref{eq: def lambda23} & classical exponent from Leland’s model (also used in the Black-Scholes setting) \\
			$\lambda_3$ & Equation~\eqref{eq: def lambda23} & positive constant used to discount perpetual payment streams taking bankruptcy risk into account \\
			$A_1$ & Equation~\eqref{eq: a1} & scaling factor of the derivative of the function $I_1^V(T)$ (defined below \eqref{eq: liability value L}) at $V_B$, i.e., scaling factor of the sensitivity of the discounted guaranteed payments w.r.t. the company's value at bankruptcy \\
			$A_2$ & Equation~\eqref{eq: a2} & scaling factor of the derivative of the function $I_2^V(T)$ (defined below \eqref{eq: liability value L}) at $V_B$, i.e., scaling factor of the sensitivity of the discounted guaranteed lump sum payment w.r.t. the company's value at bankruptcy \\
			$A_3$ & Equation~\eqref{eq: a3} & sensitivity of the option over an infinte time horizon w.r.t. changes in the company's value \\
			$A_4$ & Equation~\eqref{eq: a4} & sensitivity of the option over an finite time horizon $T$ w.r.t. changes in the company's value \\
		\end{tabular}
	\end{table}
	
	\newpage
	\setcounter{section}{0}
	\renewcommand{\thesection}{\Alph{section}}
	\appendix
	\section*{Appendix}
	
	\section{Barrier options and mathematical details on the liability structure and the optimal rates} \label{app general infos}
	
	In this section, we offer deeper insights into barrier options, and we provide  more details on the construction and mathematical foundations used to determine the liability structure and the optimal rates. 
	
	\subsection{Barrier options} \label{subsection: barrier appendix}
	
	The surplus participation component constructed in Subsection \ref{subsection: participating life insurance} is modeled as a so-called Down-and-Out-Call Option. Unlike traditional European options, barrier options are path-dependent, meaning their value depends on whether the underlying asset’s path reaches a pre-defined barrier, which in this case is the bankruptcy-triggering value, $V_B$. Barrier options are classified into two types: ``knock-in'' and ``knock-out'' options, see Hull \cite{hull2017options}. A ``knock-in'' option only pays out if the barrier is breached, while a ``knock-out'' option only pays if the barrier is not hit. Additionally, barrier options are further categorized as ``up'' or ``down'' depending on whether the barrier is above or below the initial asset value. In our framework, the value of the surplus participation is equivalent to a Down-and-Out Call option with a barrier at $V_B$ and a strike price of $k$, where the bankruptcy triggering value $V_B$ is lower than the initial insurance company value $V_0$. If the asset value hits the barrier $V_B$, bankruptcy is triggered, and all contracts terminate, meaning no further surplus participation will be paid. Notably, a Down-and-Out Call option is always cheaper than a standard Call option. The pricing formula for barrier options depends on whether the strike price is larger or smaller than the barrier. However, when the strike price equals the barrier, the pricing formulas for both cases coincide.
	
	Now, returning to our setting: Let the barrier be represented by $V_B$ and the strike price by $k$. By Hull \cite{hull2017options}, for the asset value $V$, the values of a classical call option $c$, of a Down-and-Out Call option $c_{do}^{V_B \leq k}$, when the barrier is below the strike, and of a Down-and-Out Call option $c_{do}^{V_B \geq k}$, when the barrier is above the strike, all with maturity $T$ and dividend rate $\nu$, are given by the following formulas:
	\begin{align}
		c(V_0,k,T) =&\, e^{-rT}\EX^\QQ [(V_T-k)_+] \notag \\
		=&\, V_0 e^{-\nu T} \Phi (d_1 (\textstyle\frac{V_0}{k},T)) - k e^{-rT} \Phi(d_2 (\textstyle\frac{V_0}{k},T)), \notag \\
		c_{do}^{V_B \leq k} (V_0,k,V_B,T) :=&\, e^{-rT}\EX^\QQ [(V_T-k)_+ \1_{\{\min_{s \in [0,T]} V_s \geq V_B\}}] \label{eq: cd0 vb<k formula} \\
		=&\, c(V_0,k,T) - V_0 e^{-\nu T} (\textstyle\frac{V_B}{V_0})^{2\lambda_1} \Phi (d_1 (\textstyle\frac{V_B^2}{V_0 k},T))+ke^{-rT} (\textstyle\frac{V_B}{V_0})^{2\lambda_1-2} \Phi (d_2 (\textstyle\frac{V_B^2}{V_0 k},T)), \notag \\
		c_{do}^{V_B \geq k} (V_0,k,V_B,T) :=&\, e^{-rT}\EX^\QQ [(V_T-k)_+ \1_{\{\min_{s \in [0,T]} V_s \geq V_B\}}] \label{eq: cd0 vb>k formula}\\
		=&\, V_0 \Phi (d_1 (\textstyle\frac{V_0}{V_B},T)) e^{-\nu T} - k e^{-rT} \Phi(d_2 (\textstyle\frac{V_0}{V_B},T)) - V_0 e^{-\nu T}(\textstyle\frac{V_B}{V_0})^{2\lambda_1} \Phi (d_1 (\textstyle\frac{V_B}{V_0},T)) \notag \\
		&+ke^{-rT} (\textstyle\frac{V_B}{V_0})^{2\lambda_1-2} \Phi (d_2 (\textstyle\frac{V_B}{V_0},T)), \notag
	\end{align}
	where $\Phi$ denotes the cumulative distribution function of a standard normal distribution and
	\begin{align} \label{eq: def d12 lambda1}
		d_{1/2} (x,t) &= \dfrac{\ln x + (r - \nu \pm \frac{\sigma^2}{2})t}{\sigma \sqrt{t}}, &
		\lambda_1 &= \dfrac{r-\nu +\frac{\sigma^2}{2}}{\sigma^2}.
	\end{align}
	By substituting $V_B=k$, we find that $c_{do}^{V_B \leq k}$ and $c_{do}^{V_B \geq k}$ yield the same value when $V_B=k$. Thus, we can express this as:
	\begin{align} \label{eq: def cdo}
		c_{do} (V_0,k,V_B,T) := \begin{cases}
			c_{do}^{V_B \leq k} (V_0,k,V_B,T) & \text{if $V_B \leq k$}, \\
			c_{do}^{V_B \geq k} (V_0,k,V_B,T) & \text{if $V_B > k$},
		\end{cases}
	\end{align}
	which is a continuous function in $V_B$. In particular, note that 
	\begin{align*}
		c_{do} (V_0,k,V_B,T) = \EX^{\QQ} [(V_T-k)_+ \1_{\{\min_{s \in [0,T]} V_s \geq V_B\}}].
	\end{align*}
	
	
	\subsection{Mathematical details on the liability structure} \label{app liability}
	
	In this brief paragraph, we provide explicit formulas for the functions $F^V$, $G^V$, $I_1^V$, and $I_2^V$ defined in Subsection \ref{subsection: liability} (for the proofs, we refer to Harrison \cite{harrison1990brownian}, Rubenstein and Reiner \cite{rubinstein1991breaking}, and Leland and Toft \cite{leland1996optimal}):
	\begin{align}
		F^V(t) =&\, \Phi(-d_3(\tfrac{V}{V_B},t)) + (\tfrac{V_B}{V})^{2\lambda_2} \Phi(-d_4(\tfrac{V}{V_B},t)), \label{eq: def of F}\\
		G^V(t) =&\, (\tfrac{V_B}{V})^{\lambda_2-\lambda_3} \Phi (-d_5(\tfrac{V}{V_B},t)) + (\tfrac{V_B}{V})^{\lambda_2+\lambda_3} \Phi (-d_6(\tfrac{V}{V_B},t)), \label{eq: def of G} \\
		I_1^V (T)=&\, \tfrac{1}{rT} (G^V(T)-e^{-rT}F^V(T)), \label{eq: def of i1}\\
		I_2^V(T) =&\, \tfrac{1}{\lambda_3 \sigma \sqrt{T}} \left( (\tfrac{V_B}{V})^{\lambda_2-\lambda_3} \Phi(-d_5(\tfrac{V}{V_B},T))d_5(\tfrac{V}{V_B},T) - (\tfrac{V_B}{V})^{\lambda_2+\lambda_3} \Phi (-d_6(\tfrac{V}{V_B},T)) d_6(\tfrac{V}{V_B},T) \right), \label{eq: def of i2}
	\end{align}
	where 
	\begin{align}
		d_{3/4} (x,t) &= \dfrac{\ln x \pm \lambda_2 \sigma^2 t}{\sigma \sqrt{t}}, &d_{5/6} (x,t) &= \dfrac{\ln x \pm \lambda_3 \sigma^2 t}{\sigma \sqrt{t}}, \notag \\
		\lambda_2 &= \dfrac{r-\nu-\tfrac{\sigma^2}{2}}{\sigma^2} (=\lambda_1-1), &\lambda_3 &= \dfrac{\sqrt{(\lambda_2 \sigma^2)^2+2r\sigma^2}}{\sigma^2}. \label{eq: def lambda23}
	\end{align}
	
	\subsection{Formulas for determining the optimal rates}
	
	In this subsection, we provide detailed formulas for terms presented in the results of Chapter \ref{chapter: optimal rates}. The correctness of these formulas is demonstrated in the proofs of Theorem \ref{th: alpha*} and \ref{th: g*}. 
	
	For the terms stated in Theorem \ref{th: alpha*}, where $V_B$ is expressed as a function of $\alpha$, we obtain:
	\begin{align}
		V_B'(\alpha) =& \frac{\int_0^T \frac{\partial c_{do} (V,k,V_B(\alpha),t)}{\partial V} \big|_{V=V_B(\alpha)} \diff t - \tau_2 \int_0^\infty \frac{\partial  c_{do} (V,k,V_B(\alpha),t) }{\partial V} \big|_{V=V_B(\alpha)} \diff t}{\frac{1}{V_B^2(\alpha)} \Big( \frac{2(P-\frac{G}{r})A_1}{rT} + 2 \frac{G}{r} A_2 - \tau_1 \frac{G}{r}(\lambda_2+\lambda_3) \Big) + \tau_2 \alpha \int_0^\infty \partial c_{do}(t) \diff t - \alpha \int_0^T \partial c_{do}(t) \diff t}, \label{eq: formula dvb dalpha}\\
		\partial c_{do}(t) :=&\, \tfrac{\partial}{\partial V_B(\alpha)}[\tfrac{\partial  c_{do} (V,k,V_B(\alpha),t) }{\partial V} \big|_{V=V_B(\alpha)}] \notag\\
		=&\, \tfrac{2ke^{-rt}}{V_B^2(\alpha)} \Big(\lambda_2  \Phi(d_2(\min\{\frac{V_B(\alpha)}{k},1\},t)) + \frac{ \varphi(d_2(\min\{\frac{V_B(\alpha)}{k},1\},t))}{\sigma \sqrt{t}} \Big), \label{eq: partial cdo}
	\end{align}
	with the special case
	\begin{align}
		V_B'(0) = (V_B(0))^2\frac{\scaleobj{1.333}{\int_{\scaleobj{0.75}{0}}^{\scaleobj{0.75}{T}}} \frac{\partial c_{do} (V,k,V_B(0),t)}{\partial V} \big|_{V=V_B(0)} \diff t-\tau_2 \scaleobj{1.333}{\int_{\scaleobj{0.75}{0}}^{\scaleobj{0.75}{\infty}}} \frac{\partial  c_{do} (V,k,V_B(0),t) }{\partial V} \big|_{V=V_B(0)} \diff t}{\Big( \frac{2(P-\frac{G}{r})A_1}{rT} + 2 \frac{G}{r} A_2 - \tau_1 \frac{G}{r}(\lambda_2+\lambda_3) \Big)}. \label{eq: partial vb0}
	\end{align}
	
	For the terms stated in Theorem \ref{th: g*}, with $V_B$ as a function of $g$, we find that $V_B(0)$ is the smallest solution of
	\begin{align}
		0 =&\, 1 + \rho(\lambda_2+\lambda_3)+2(1-\rho)A_2 - \tfrac{2PA_1}{V_B(0) rT} + \tau_2 \alpha \int_0^\infty \tfrac{\partial  c_{do} (V,k,V_B(0),t) }{\partial V} \big|_{V=V_B(0)} \diff t \notag \\
		& - \alpha \textstyle\int_0^T \tfrac{\partial c_{do} (V,k,V_B(0),t)}{\partial V} \big|_{V=V_B(0)} \diff t.  \label{eq: vb0 g}
	\end{align}
	Furthermore, we obtain:
	\begin{align}
		V_B'(g) =& \frac{\tfrac{T}{V_B(g) r} ( -\tfrac{2A_1}{rT} + 2A_2 - \tau_1(\lambda_2+\lambda_3))}{\frac{1}{V_B^2(g)} ( \frac{2(P-\frac{G}{r})A_1}{rT} + 2 \frac{G}{r} A_2 - \tau_1 \frac{G}{r}(\lambda_2+\lambda_3) ) + \tau_2 \alpha \int_0^\infty \partial c_{do}(t) \diff t - \alpha \int_0^T \partial c_{do}(t) \diff t}, \label{eq: formula dvb dg}\\
		\partial c_{do}(t) :=&\, \tfrac{\partial}{\partial V_B(g)}[\tfrac{\partial  c_{do} (V,k,V_B(g),t) }{\partial V} \big|_{V=V_B(g)}] \notag\\
		=&\, \tfrac{2ke^{-rt}}{(V_B(g))^2} \Big(\lambda_2  \Phi(d_2(\min\{\frac{V_B(g)}{k},1\},t)) + \frac{ \varphi(d_2(\min\{\frac{V_B(g)}{k},1\},t))}{\sigma \sqrt{t}} \Big), \label{eq: partial cdo g}
	\end{align}
	with the special case
	\begin{align}
		V_B'(0) =& \frac{\tfrac{T}{V_B(0) r} ( -\tfrac{2A_1}{rT} + 2A_2 - \tau_1(\lambda_2+\lambda_3) )}{\tfrac{2PA_1}{(V_B(0))^2rT} + \alpha (\tau_2 \int_0^\infty \partial_0 c_{do}(t)  \diff t - \int_0^T \partial_0 c_{do}(t) \diff t)}, \label{eq: partial vb0 g} \\
		\partial_0 c_{do}(t) :=& \tfrac{\partial}{\partial V_B}[\tfrac{\partial  c_{do} (V,k,V_B,t) }{\partial V} \big|_{V=V_B}] \big|_{g=0} \notag \\
		=&\, \tfrac{2ke^{-rt}}{(V_B(0))^2} \Big(\lambda_2  \Phi(d_2(\min\{\tfrac{V_B(0)}{k},1\},t)) + \frac{ \varphi(d_2(\min\{\frac{V_B(0)}{k},1\},t))}{\sigma \sqrt{t}}\Big). \label{eq: partial cdo g0}
	\end{align}
	
	\subsection{Assumption for Section 4.2} \label{subsection: Assumption for Section 4.2}
	
	We state an assumption made (solely) for section 4.2 (see Theorem~\ref{th: g*}):
	
	\begin{assumption} \label{ass: g*>0}
		Let 
		\begin{align} \label{eq: assumption optimal g}
			&\tfrac{2PA_1}{(V_B(0))^2rT} + \tau_2 \alpha \displaystyle\int_0^\infty \tfrac{\partial}{\partial V_B}[\tfrac{\partial  c_{do} (V,k,V_B,t) }{\partial V} \big|_{V=V_B}] \big|_{V_B=V_B(0)} \diff t \notag\\
			&\hspace{3cm}- \alpha \displaystyle\int_0^T \tfrac{\partial}{\partial V_B}[\tfrac{\partial c_{do} (V,k,V_B,t)}{\partial V} \big|_{V=V_B}] \big|_{V_B=V_B(0)} \diff t \neq 0, \notag \\
			&\tfrac{\tau_1 T}{r} (1-(\tfrac{V_B(0)}{V})^{\lambda_2+\lambda_3}) - \rho (\lambda_2+\lambda_3+1) (\tfrac{V_B(0)}{V})^{\lambda_2+\lambda_3} V_B'(0)	+ \alpha \tau_2 \int_0^\infty \tfrac{\partial c_{do} (V,k,V_B,t)}{\partial g} \big|_{g=0} \diff t >0,
		\end{align}
		where more explicit formulas are given in \eqref{eq: def lambda23}, \eqref{eq: vb0 g}, \eqref{eq: partial vb0 g}, \eqref{eq: dv2 cdo<} resp. \eqref{eq: dv2 cdo>}, and \eqref{eq: dg cdo<} resp. \eqref{eq: dg cdo>} in the Appendix with $g=G=0$.
	\end{assumption}
	
	A numerical analysis shows that this assumption already holds for small values of $\tau_1$ (greater than $0.1 \%$ in our basic setting).
	
	\section{Technical Lemmas} \label{lemmas}
	
	In this section, we present and prove several technical lemmas that are used in the proofs of the theorems and propositions discussed in the main text.
	
	\begin{lemma} \label{lem: int cdo < infty}
		It holds that $|\int_0^\infty c_{do} (V,k,V_B,t) \diff t| < \infty$ for all $V>0$, $k\geq0$, and $V_B\geq0$. 
	\end{lemma}
	
	\begin{myproof}
		It holds:
		\begin{align*}
			\left|\int_0^\infty c_{do} (V_0,k,V_B,t) \diff t \right| &= \int_0^\infty \EX^\QQ \left[e^{-rt}(V_t-k)_+ \1_{V_s \geq V_B \, \forall s \in [0,t]} \right] \diff t  \\
			&\leq \int_0^\infty e^{-rt} \EX^\QQ \left[V_t \right] \diff t 
			= \int_0^\infty e^{-rt} V_0 e^{(r-\nu)t} \diff t = V_0 \int_0^\infty e^{-\nu t} \diff t = \tfrac{V_0}{\nu} < \infty,
		\end{align*}
		since $\nu>0$ and where we could drop the absolute value, as everything is non-negative. We denoted $V=V_0$ in accordance with the notation in the main part.
	\end{myproof}
	
	\begin{lemma} \label{lemma: cdo c1} 
		It holds that $c_{do} (V,k,V_B,T)$ is continuously differentiable as a function of $V$.
	\end{lemma}
	
	\begin{myproof}
		It is evident from equations \eqref{eq: cd0 vb<k formula} and \eqref{eq: cd0 vb>k formula} that both $c_{do}^{V_B \leq k} (V,k,V_B,T)$ and $c_{do}^{V_B \geq k} (V,k,V_B,T)$ are continuously differentiable. Therefore, it suffices to verify whether the derivatives coincide at $V_B=k$. Differentiating equations \eqref{eq: cd0 vb<k formula} and \eqref{eq: cd0 vb>k formula} yields:
		\begin{align}
			\tfrac{\partial}{\partial V} c_{do}^{V_B \leq k} (V,k,V_B,T) =&\, e^{-\nu T} \Phi(d_1(\tfrac{V}{k},T))+Ve^{-\nu T} \varphi(d_1(\tfrac{V}{k},T)) \frac{1}{\sigma \sqrt{T} V} -ke^{-rT} \varphi(d_2(\tfrac{V}{k},T))\frac{1}{\sigma \sqrt{T} V} \notag \\
			&- e^{-\nu T} V_B^{2\lambda_1} (1-2\lambda_1) V^{-2\lambda_1} \Phi(d_1(\tfrac{V_B^2}{Vk},T)) - V e^{-\nu T} (\tfrac{V_B}{V})^{2\lambda_1} \varphi(d_1(\tfrac{V_B^2}{Vk},T)) \frac{-1}{\sigma \sqrt{T} V} \notag\\
			&+ ke^{-rT} V_B^{2\lambda_1-2}  (2-2\lambda_1) V^{1-2\lambda_1} \Phi(d_2(\tfrac{V_B^2}{Vk},T)) \notag\\
			&+ ke^{-rT} (\tfrac{V_B}{V})^{2\lambda_1-2} \varphi(d_2(\tfrac{V_B^2}{Vk},T)) \frac{-1}{\sigma \sqrt{T} V} \notag\\
			=&\, e^{-\nu T} \Phi(d_1(\tfrac{V}{k},T))+ \frac{e^{-\nu T} \varphi(d_1(\frac{V}{k},T))}{\sigma \sqrt{T}} - \frac{ke^{-rT} \varphi(d_2(\frac{V}{k},T))}{\sigma \sqrt{T} V} \notag\\
			&- (1-2\lambda_1)e^{-\nu T} (\tfrac{V_B}{V})^{2\lambda_1} \Phi(d_1(\tfrac{V_B^2}{Vk},T)) + \frac{e^{-\nu T} (\frac{V_B}{V})^{2\lambda_1} \varphi(d_1(\frac{V_B^2}{Vk},T))}{\sigma \sqrt{T}} \notag\\
			&+ (2-2\lambda_1) \tfrac{ke^{-rT}}{V} (\tfrac{V_B}{V})^{2\lambda_1-2} \Phi(d_2(\tfrac{V_B^2}{Vk},T)) - \frac{ke^{-rT} (\frac{V_B}{V})^{2\lambda_1-2} \varphi(d_2(\frac{V_B^2}{Vk},T))}{\sigma \sqrt{T} V}, \label{eq: derivative cdo vb<k}\\
			\tfrac{\partial}{\partial V} c_{do}^{V_B \geq k} (V,k,V_B,T) =&\, e^{-\nu T} \Phi(d_1(\tfrac{V}{V_B},T))+Ve^{-\nu T} \varphi(d_1(\tfrac{V}{V_B},T)) \frac{1}{\sigma \sqrt{T} V} -ke^{-rT} \varphi(d_2(\tfrac{V}{V_B},T))\frac{1}{\sigma \sqrt{T} V} \notag\\
			&- e^{-\nu T} V_B^{2\lambda_1} (1-2\lambda_1) V^{-2\lambda_1} \Phi(d_1(\tfrac{V_B}{V},T)) - V e^{-\nu T} (\tfrac{V_B}{V})^{2\lambda_1} \varphi(d_1(\tfrac{V_B}{V},T)) \frac{-1}{\sigma \sqrt{T} V} \notag\\
			&+ ke^{-rT} V_B^{2\lambda_1-2}  (2-2\lambda_1) V^{1-2\lambda_1} \Phi(d_2(\tfrac{V_B}{V},T)) \notag\\
			&+ ke^{-rT} (\tfrac{V_B}{V})^{2\lambda_1-2} \varphi(d_2(\tfrac{V_B}{V},T)) \frac{-1}{\sigma \sqrt{T} V} \notag\\
			=&\, e^{-\nu T} \Phi(d_1(\tfrac{V}{V_B},T))+ \frac{e^{-\nu T} \varphi(d_1(\frac{V}{V_B},T))}{\sigma \sqrt{T}} - \frac{ke^{-rT} \varphi(d_2(\frac{V}{V_B},T))}{\sigma \sqrt{T} V} \notag\\
			&- (1-2\lambda_1)e^{-\nu T} (\tfrac{V_B}{V})^{2\lambda_1} \Phi(d_1(\tfrac{V_B}{V},T)) + \frac{e^{-\nu T} (\frac{V_B}{V})^{2\lambda_1} \varphi(d_1(\frac{V_B}{V},T))}{\sigma \sqrt{T}}\notag \\
			&+ (2-2\lambda_1) \tfrac{ke^{-rT}}{V} (\tfrac{V_B}{V})^{2\lambda_1-2} \Phi(d_2(\tfrac{V_B}{V},T)) - \frac{ke^{-rT} (\frac{V_B}{V})^{2\lambda_1-2} \varphi(d_2(\frac{V_B}{V},T))}{\sigma \sqrt{T} V}.\label{eq: derivative cdo vb>k}
		\end{align}
		Hence, we obtain $\tfrac{\partial}{\partial V} c_{do}^{V_B \leq k} (V,k,V_B,T) \big|_{V_B=k} = \tfrac{\partial}{\partial V} c_{do}^{V_B \geq k} (V,k,V_B,T) \big|_{V_B=k}$ and the claim follows.
	\end{myproof}
	
	\begin{lemma} \label{lemma: varphi bounded}
		Let $a,b,y \in \Real$, and $x>0$. Then, the function $h(x):=\frac{\varphi(a \ln(x)+b)}{x^y}$ is continuous on $(0,\infty)$, and there exists a constant $C>0$ such that $|h(x)| \leq C$ for all $x \in (0,\infty)$. Moreover, we have the following limits: $\lim_{x \to 0} |h(x)| = 0$ and $\limsup_{x \to \infty} |h(x)| \leq C$.
	\end{lemma}
	
	\begin{myproof}
		The continuity of $h$ follows directly from the continuity of $\varphi$. Next, we show the existence of a constant $C>0$ such that $|h(x)| \leq C$. First, we note that $h(x)>0$ for all $x \in (0,\infty)$. If we can show that $h$ has a unique extreme point $x^* \in (0,\infty)$ which is a local maximum, then $x^*$ will also be the global maximum of $h$ in $(0,\infty)$, and the main claim follows with $C:= h(x^*)$.
		
		To prove the existence of a unique extreme point, we compute the first two derivatives of $h$:
		\begin{align*}
			h'(x) =&\, \dfrac{-(a \ln(x)+b)\varphi(a \ln(x)+b) \frac{a}{x}x^y-\varphi(a \ln(x)+b)yx^{y-1}}{x^{2y}} \\
			=&\, - (a^2 \ln(x)+ab+y) \dfrac{\varphi(a \ln(x)+b)}{x^{y+1}}, \\
			h''(x) =&\, -\frac{a^2}{x} \cdot \dfrac{\varphi(a \ln(x)+b)}{x^{y+1}} \\
			&- (a^2 \ln(x)+ab+y) \dfrac{-(a \ln(x)+b)\varphi(a \ln(x)+b) \frac{a}{x}x^{y+1}-\varphi(a \ln(x)+b)(y+1)x^{y}}{x^{2y+2}} \\
			=&\, \dfrac{\varphi(a \ln(x)+b)}{x^{y+2}} \left( -a^2+ (a^2 \ln(x)+ab+y)(a^2 \ln(x)+ab+y+1) \right),
		\end{align*}
		where we used that $\varphi'(x) = -x \varphi(x)$. Setting $h'(x) \overset{!}{=} 0$ is equivalent to the equation $(a^2 \ln(x)+ab+y)=0$ since $\varphi(\cdot)>0$ and $x>0$. Solving this equation for $x$ yields the unique solution $x^* = e^{-\frac{y}{a^2}-\frac{b}{a}}>0$. Thus, we have a unique extreme point. Plugging $x^*$ into $h''$, we obtain:
		\begin{align*}
			h''(x^*) &= \dfrac{\varphi(a \ln(x^*)+b)}{(x^*)^{y+2}} \left( -a^2+ (a^2 (-\tfrac{y}{a^2}-\tfrac{b}{a})+ab+y)(a^2 (-\tfrac{y}{a^2}-\tfrac{b}{a})+ab+y+1) \right) \\
			&= \dfrac{\varphi(a \ln(x^*)+b)}{(x^*)^{y+2}} \left( -a^2 + (-y-ab+ab+y)(-y-ab+ab+y+1) \right) \\
			&= - a^2 \dfrac{\varphi(a \ln(x^*)+b)}{(x^*)^{y+2}} < 0,
		\end{align*}
		since $\varphi(\cdot)>0$ and $x^*>0$. Therefore, $x^*$ is the unique extreme point and a maximum, implying that it is the global maximum of $h$. Consequently, the main claim follows. 
		
		It remains to show that $\lim_{x \to 0} |h(x)| = 0$ (since $\limsup_{x \to \infty} |h(x)| \leq C$ follows directly from the first part). We already know from the first part that $\limsup_{x \to 0} |h(x)| \leq C$. Now, we show that $\lim_{x \to 0} |h(x)| = 0$ by contradiction. Assume that there exists a sequence $x_n \xrightarrow{n \to \infty} 0$ such that $\lim_{n \to \infty} h(x_n) = c \in [-C,C]\backslash\{0\}$. Applying l'H{\^o}pital's rule and using the fact that $\varphi'(x)=-x\varphi(x)$ in the second equation, we get:
		\begin{align*}
			c &= \lim_{n \to \infty} \frac{\varphi(a \ln(x_n)+b)}{x_n^y} \\
			&= \lim_{n \to \infty} \frac{-(a \ln(x_n)+b)\varphi(a \ln(x_n)+b) \frac{a}{x_n}}{y x_n^{y-1}} \\
			&=  \lim_{n \to \infty} \tfrac{-a}{y} (a \ln(x_n)+b) h(x_n) = - \tfrac{a}{y} c \cdot \lim_{n \to \infty} (a \ln(x_n)+b) = \infty,
		\end{align*}
		since $x_n \xrightarrow{n \to \infty} 0$. This leads to a contradiction, and the claim is proven.
	\end{myproof}
	
	\begin{lemma} \label{lemma: varphi ln 0}
		Let $a,b,y \in \Real$, and $x>0$. Define $h(x):=\frac{\ln(x)\varphi(a \ln(x)+b)}{x^y}$. Then, we get that $\lim_{x \to 0} |h(x)| = 0$.
	\end{lemma}
	
	\begin{myproof}
		We get with the substitution $z=\ln(x)$:
		\begin{align*}
			\lim_{x \to 0} |h(x)| = \lim_{z \to -\infty} |{\tfrac{z\varphi(a z+b)}{e^{yz}}}|  = \lim_{z \to -\infty} |z e^{-yz} {\tfrac{1}{2\pi}} e^{-\frac{1}{2}(az+b)^2}| = {\tfrac{1}{2\pi}} \lim_{z \to -\infty} |z| e^{-\frac{1}{2}(az+b)^2-yz} = 0,
		\end{align*}
		since $\lim_{z \to -\infty} -\frac{1}{2}(az+b)^2-yz = -\infty$. 
	\end{myproof}
	
	\begin{lemma} \label{lemma: phi bounded}
		Let $a,y,x>0$, and $b \in \Real$. Then, the function $h(x):=\Phi(-a \ln(x)+b)x^{y}$ is continuous on $(0,\infty)$ and there exists a constant $C>0$ such that $|h(x)| \leq C$ for all $x \in (0,\infty)$. Furthermore, we have the following properties: $\lim_{x \to 0} |h(x)| = 0$ and $\limsup_{x \to \infty} |h(x)| \leq C$. Additionally, we consider the functions $\tilde{h}(x):=\frac{\Phi(-a \ln(x)+b)}{x^{y}}$ and $\hat{h}(x):=\frac{\Phi(a \ln(x)+b)}{x^{y}}$. For these functions, it holds that $\lim_{x \to \infty} |\tilde{h}(x)| = 0$ and $\lim_{x \to 0} |\hat{h}(x)| = 0$.
	\end{lemma}
	
	\begin{myproof}
		First, note that $h$ is obviously continuous in $(0,\infty)$, and therefore bounded (by a possibly larger $C>0$) if there exists a $C>0$ such that $\lim_{x \to 0} |h(x)| \leq C$ and $\lim_{x \to \infty} |h(x)| \leq C$. Moreover, we observe that $h(\cdot)>0$, $\tilde{h}(\cdot)>0$, and $\hat{h}(\cdot)>0$ on $(0,\infty)$. Now, we obtain for a suitable $C>0$:
		\begin{align*}
			\lim_{x \to 0} |h(x)| &= \lim_{x \to 0} \Phi(-a \ln(x)+b) x^y = 0, \\
			\limsup_{x \to \infty} |h(x)| &= \limsup_{x \to \infty} \frac{\Phi(-a \ln(x)+b)}{x^{-y}} \\
			&\leq \limsup_{x \to \infty} \frac{\varphi(-a \ln(x)+b) \frac{-a}{x}}{-yx^{-y-1}} = \tfrac{a}{y} \limsup\limits_{x \to \infty} \dfrac{\varphi(-a \ln(x)+b)}{x^{-y}} \leq C,
		\end{align*}
		since $a,y>0$ and $\lim_{z \to -\infty} \Phi(z)=0$. Note that we used the generalized rule of de l'H{\^o}pital (see, e.g., Picone \cite{picone1929sul}) in the second step of the second limit and Lemma \ref{lemma: varphi bounded} in the last step. Thus, the first claim follows. For the second claim, we have:
		\begin{align*}
			\lim_{x \to \infty} |\tilde{h}(x)| &= \lim_{x \to \infty} \frac{\Phi(-a \ln(x)+b)}{x^{y}} = 0, \\
			\lim_{x \to 0} |\hat{h}(x)| &= \lim_{x \to 0} \frac{\Phi(a \ln(x)+b)}{x^{y}} = \lim_{x \to 0} \frac{\varphi(a \ln(x)+b) \tfrac{a}{x}}{yx^{y-1}} = \tfrac{a}{y} \displaystyle \lim_{x \to 0} \frac{\varphi(a \ln(x)+b)}{x^{y}} = 0,
		\end{align*}
		since $a,y>0$ and $|\Phi(\cdot)| \leq1$. Note that we used again the rule of de l'H{\^o}pital in the second step of the second limit, and Lemma \ref{lemma: varphi bounded} in the last step. Thus, the second claim follows.
	\end{myproof}
	
	\begin{lemma} \label{lemma: cdo bounded L1}
		For all $k \geq 0$, $V_B \geq 0$, and $T>0$, there exists a constant $C>0$ such that $|\frac{\partial}{\partial V} c_{do} (V,k,V_B,T)| \leq C (e^{-\nu T} + e^{-rT})$ for all $V \geq V_B$.
	\end{lemma}
	
	\begin{myproof}
		For this proof, we need to consider two cases: when $V_B=0$ and when $V_B>0$. Let us start with $V_B=0$. Then, the Down-and-Out Call option becomes a classical Call option as $V_t>0$ for all $t \geq 0$ by the non-negativity of Geometric Brownian Motions, i.e., $c_{do}^{V_B \leq k} (V,k,V_B,T) = c(V,k,T)$. Then, the result follows analogously to the proof of Lemma \ref{lemma: cdo c1}:
		\begin{align*}
			\tfrac{\partial}{\partial V} c(V,k,T) =&\, e^{-\nu T} \Phi(d_1(\tfrac{V}{k},T))+ \frac{e^{-\nu T} \varphi(d_1(\frac{V}{k},T))}{\sigma \sqrt{T}} - \frac{ke^{-rT} \varphi(d_2(\frac{V}{k},T))}{\sigma \sqrt{T} V}.
		\end{align*} 
		Thus, Lemma \ref{lemma: varphi bounded} and $|\Phi(\cdot)| \leq 1$ immediately provide the desired result.
		
		Now, consider the case when $V_B>0$. From equations \eqref{eq: derivative cdo vb<k} and \eqref{eq: derivative cdo vb>k}, we can conclude that the claim will hold if we show that each individual term in equations \eqref{eq: derivative cdo vb<k} and \eqref{eq: derivative cdo vb>k} is bounded by $C(e^{-\nu T} + e^{-rT})$ for some constant $C>0$, both as $V \to V_B$ and as $V \to +\infty$. Once we establish this, it follows that $|\frac{\partial}{\partial V} c_{do} (V,k,V_B,T)| \leq C (e^{-\nu T} + e^{-rT})$, where $C>0$ may be larger, but still finite. For the case $V \to V_B >0$, the result holds immediately. For $V \to + \infty$, each individual term is bounded by $C(e^{-\nu T} + e^{-rT})$, with $C>0$, due to the fact that $|\Phi(\cdot)|\leq 1$, Lemma \ref{lemma: varphi bounded}, or Lemma \ref{lemma: phi bounded} since $\ln(\tfrac{V_B^2}{Vk}) = 2\ln(V_B)-\ln(V)-\ln(k)$ (resp. $\ln(\tfrac{V_B}{V}) = \ln(V_B)-\ln(V)$). 
	\end{myproof}
	
	\begin{lemma} \label{lemma: cdo limit of derivative to 0 and infty}
		For all $V_B \geq 0$, and $T>0$, it holds:
		\begin{enumerate}[(a)]
			\item $\lim_{V_B \to \infty} (\frac{\partial}{\partial V} c_{do} (V,k,V_B,T) \big|_{V=V_B}) = e^{-\nu T}(2\lambda_1 \Phi(d_1(1,T))+ \frac{2}{\sigma \sqrt{T}} \varphi(d_1(1,T)))$ for all $k \geq 0$,
			\item $\lim_{V_B \to 0} (\frac{\partial}{\partial V} c_{do} (V,k,V_B,T) \big|_{V=V_B}) = 0$ for all $k>0$,
			\item\label{lemma, item: k=0 limit} $\frac{\partial}{\partial V} c_{do} (V,0,V_B,T) \big|_{V=V_B} = e^{-\nu T}(2\lambda_1 \Phi(d_1(1,T))+ \frac{2}{\sigma \sqrt{T}} \varphi(d_1(1,T)))$.
		\end{enumerate}
	\end{lemma}
	
	\begin{myproof}
		First, we obtain from \eqref{eq: derivative cdo vb<k} and \eqref{eq: derivative cdo vb>k}:
		\begin{align}
			\tfrac{\partial}{\partial V} c_{do}^{V_B \leq k} (V,k,V_B,T) \big|_{V=V_B} =&\, e^{-\nu T} \Phi(d_1(\tfrac{V_B}{k},T))+ \frac{e^{-\nu T} \varphi(d_1(\frac{V_B}{k},T))}{\sigma \sqrt{T}} - \frac{ke^{-rT} \varphi(d_2(\frac{V_B}{k},T))}{\sigma \sqrt{T} V_B} \notag \\
			&- (1-2\lambda_1)e^{-\nu T} \Phi(d_1(\tfrac{V_B}{k},T)) + \frac{e^{-\nu T} \varphi(d_1(\frac{V_B}{k},T))}{\sigma \sqrt{T}} \notag \\
			&+ (2-2\lambda_1) \tfrac{ke^{-rT}}{V_B} \Phi(d_2(\tfrac{V_B}{k},T)) - \frac{ke^{-rT} \varphi(d_2(\frac{V_B}{k},T))}{\sigma \sqrt{T} V_B}, \notag \\
			=&\, 2 \lambda_1 e^{-\nu T} \Phi(d_1(\tfrac{V_B}{k},T))+ \frac{2 e^{-\nu T} \varphi(d_1(\frac{V_B}{k},T))}{\sigma \sqrt{T}} - \frac{2ke^{-rT} \varphi(d_2(\frac{V_B}{k},T))}{\sigma \sqrt{T} V_B} \notag \\
			&+ (2-2\lambda_1) \tfrac{ke^{-rT}}{V_B} \Phi(d_2(\tfrac{V_B}{k},T)), \label{eq: dv cdo< v=vb} \\
			\tfrac{\partial}{\partial V} c_{do}^{V_B \geq k} (V,k,V_B,T) \big|_{V=V_B} =&\, e^{-\nu T} \Phi(d_1(1,T))+ \tfrac{e^{-\nu T} \varphi(d_1(1,T))}{\sigma \sqrt{T}} - \frac{ke^{-rT} \varphi(d_2(1,T))}{\sigma \sqrt{T} V_B} \notag\\
			&- (1-2\lambda_1)e^{-\nu T} \Phi(d_1(1,T)) + \tfrac{e^{-\nu T} \varphi(d_1(1,T))}{\sigma \sqrt{T}} \notag\\
			&+ (2-2\lambda_1) \tfrac{ke^{-rT}}{V_B} \Phi(d_2(1,T)) - \frac{ke^{-rT} \varphi(d_2(1,T))}{\sigma \sqrt{T} V_B} \notag \\
			=&\, 2\lambda_1 e^{-\nu T} \Phi(d_1(1,T))+ \tfrac{2e^{-\nu T} \varphi(d_1(1,T))}{\sigma \sqrt{T}} - \frac{2ke^{-rT} \varphi(d_2(1,T))}{\sigma \sqrt{T} V_B} \notag\\
			&+ (2-2\lambda_1) \tfrac{ke^{-rT}}{V_B} \Phi(d_2(1,T)). \label{eq: dv cdo> v=vb}
		\end{align}
		Now, we conclude for the proof of parts (a) and (b):
		\begin{align*}
			\lim_{V_B \to \infty} \tfrac{\partial}{\partial V} c_{do} (V,k,V_B,T) \big|_{V=V_B} =&\, \lim_{V_B \to \infty} \tfrac{\partial}{\partial V} c_{do}^{V_B \geq k} (V,k,V_B,T) \big|_{V=V_B} \\
			=&\, 2 \lambda_1 e^{-\nu T} \Phi(d_1(1,T))+ \tfrac{2e^{-\nu T} \varphi(d_1(1,T))}{\sigma \sqrt{T}} - 0 + 0,\\
			=&\, e^{-\nu T} \left( 2 \lambda_1 \Phi(d_1(1,T)) + \tfrac{2}{\sigma \sqrt{T}} \varphi(d_1(1,T)) \right), \\
			\lim_{V_B \to 0} \tfrac{\partial}{\partial V} c_{do} (V,k,V_B,T) \big|_{V=V_B} =&\, \lim_{V_B \to 0} \tfrac{\partial}{\partial V} c_{do}^{V_B \leq k} (V,k,V_B,T) \big|_{V=V_B} \\
			=&\, 2 \lambda_1 e^{-\nu T} \Phi(d_1(0,T))+ \tfrac{2e^{-\nu T} \varphi(d_1(0,T))}{\sigma \sqrt{T}} - \tfrac{2ke^{-rT}}{\sigma \sqrt{T}} \cdot \lim\limits_{V_B \to 0} \tfrac{\varphi(d_2(\frac{V_B}{k},T))}{V_B} \\
			&+ (2-2\lambda_1) ke^{-rT} \cdot \lim\limits_{V_B \to 0} \tfrac{\Phi(d_2(\frac{V_B}{k},T))}{V_B} \\
			=&\, 0,
		\end{align*}
		due to $d_1(0,T)=-\infty$, $\lim_{z \to -\infty} \varphi(z)=0$, $\lim_{z \to -\infty} \Phi(z)=0$, Lemma \ref{lemma: varphi bounded}, and Lemma \ref{lemma: phi bounded}. Note that the first step (i.e., using the formula of the Down-and-Out Call option for $V_B \geq k$ (resp. $V_B \leq k$) when taking the limit $V_B \to \infty$ (resp. $V_B \to 0$)) is valid because $k>0$. 
		
		For part (c), when $k = 0$, it holds since $V_B \geq 0$: 
		\begin{align*}
			\tfrac{\partial}{\partial V} c_{do} (V,0,V_B,T) \big|_{V=V_B} =&\, \tfrac{\partial}{\partial V} c_{do}^{V_B \geq k} (V,0,V_B,T) \big|_{V=V_B} \\
			=&\, e^{-\nu T} \Phi(d_1(1,T))+ \tfrac{e^{-\nu T} \varphi(d_1(1,T))}{\sigma \sqrt{T}} - 0 \\
			&- (1-2\lambda_1)e^{-\nu T} \Phi(d_1(1,T)) + \tfrac{e^{-\nu T} \varphi(d_1(1,T))}{\sigma \sqrt{T}} + 0 - 0 \\
			=&\, e^{-\nu T}(2\lambda_1 \Phi(d_1(1,T))+ \tfrac{2}{\sigma \sqrt{T}} \varphi(d_1(1,T))). \qedhere
		\end{align*} 
	\end{myproof}
	
	\begin{lemma} \label{lemma: cdo unif bounded L1}
		For all $k \geq 0$, and $T>0$, there exists a constant $C>0$ such that \linebreak $|\frac{\partial}{\partial V} c_{do} (V,k,V_B,T)| \big|_{V=V_B} \leq C (e^{-\nu T} + e^{-rT})$ for all $V=V_B \in (0,\infty)$.
	\end{lemma}
	
	\begin{myproof}
		By equations \eqref{eq: dv cdo< v=vb} and \eqref{eq: dv cdo> v=vb}, the claim holds if we can show that each individual term is bounded by $C(e^{-\nu T} + e^{-rT})$ for some constant $C>0$ in the limits $V_B \to 0$ and $V_B \to +\infty$. Then, in total, we have $|\frac{\partial}{\partial V} c_{do} (V,k,V_B,T)| \big|_{V=V_B} \leq C (e^{-\nu T} + e^{-rT})$ with a possibly larger constant $C>0$. We now distinguish the cases $k=0$ and $k>0$. If $k=0$, we find that $|\frac{\partial}{\partial V} c_{do} (V,k,V_B,T)| \big|_{V=V_B}$ is constant, and hence bounded. Now, let $k>0$. For the limit as $V_B \to \infty$, the claim follows directly from equation \eqref{eq: dv cdo> v=vb} (since $V_B \geq k$ for $V_B$ sufficiently large). For the limit as $V_B \to 0$, we can assume that $|\frac{\partial}{\partial V} c_{do} (V,k,V_B,T)| \big|_{V=V_B} = |\frac{\partial}{\partial V} c_{do}^{V_B \leq k} (V,k,V_B,T)| \big|_{V=V_B}$, since $k>0$ and $V_B<k$ for $V_B$ sufficiently small. In this case, the claim follows from equation \eqref{eq: dv cdo< v=vb} using that $|\Phi(\cdot)| \leq 1$, $|\varphi(\cdot)| \leq 1$, and Lemmas \ref{lemma: varphi bounded} and \ref{lemma: phi bounded}.
	\end{myproof}
	
	\begin{lemma} \label{lemma: i1 i2 derivative}
		It holds:
		\begin{align*}
			\tfrac{\partial I_1^V(T)}{\partial V} \big|_{V=V_B} =&\, -\tfrac{2}{rTV_B} \Big( \tfrac{\lambda_2-\lambda_3}{2} + \lambda_3 \Phi (\lambda_3 \sigma \sqrt{T}) - \lambda_2 e^{-rT} \Phi(\lambda_2 \sigma \sqrt{T})\Big),\\
			\tfrac{\partial I_2^V(T)}{\partial V} \big|_{V=V_B} =&\, - \tfrac{2}{V_B} \Big( \tfrac{\lambda_2-\lambda_3}{2} - \tfrac{1}{2 \lambda_3 \sigma^2 T} + (\lambda_3+\tfrac{1}{\lambda_3 \sigma^2 T} ) \Phi (\lambda_3 \sigma \sqrt{T}) + \tfrac{\varphi (\lambda_3\sigma\sqrt{T})}{\sigma \sqrt{T}} \Big). 
		\end{align*}
	\end{lemma}
	
	\begin{myproof}
		By the definition of $I_1$ (see \eqref{eq: def of i1}), we begin by differentiating $F$ and $G$ (which are defined in equations \eqref{eq: def of F} and \eqref{eq: def of G}, respectively):
		\begin{align*}
			\tfrac{\partial F^V(T)}{\partial V} =&\, \varphi(-d_3(\tfrac{V}{V_B},T))\tfrac{-1}{\sigma \sqrt{T} V} -2\lambda_2  V_B^{2\lambda_2} V^{-2\lambda_2-1} \Phi(-d_4(\tfrac{V}{V_B},T)) + (\tfrac{V_B}{V})^{2\lambda_2} \varphi(-d_4(\tfrac{V}{V_B},T))\tfrac{-1}{\sigma \sqrt{T} V} \\
			=&\, -\tfrac{1}{V} \Big( \tfrac{\varphi(-d_3(\frac{V}{V_B},T))}{\sigma \sqrt{T}} + 2 \lambda_2 (\tfrac{V_B}{V})^{2\lambda_2} \Phi(-d_4(\tfrac{V}{V_B},T)) + (\tfrac{V_B}{V})^{2\lambda_2} \tfrac{\varphi(-d_4(\frac{V}{V_B},T))}{\sigma \sqrt{T}}  \Big), \\
			\tfrac{\partial G^V(T)}{\partial V} =&\, (-\lambda_2+\lambda_3) V_B^{\lambda_2-\lambda_3} V^{-\lambda_2+\lambda_3-1} \Phi (-d_5(\tfrac{V}{V_B},T)) + (\tfrac{V_B}{V})^{\lambda_2-\lambda_3} \varphi (-d_5(\tfrac{V}{V_B},T)) \tfrac{-1}{\sigma \sqrt{T} V} \\
			&+ (-\lambda_2-\lambda_3) V_B^{\lambda_2+\lambda_3} V^{-\lambda_2-\lambda_3-1} \Phi (-d_6(\tfrac{V}{V_B},T)) + (\tfrac{V_B}{V})^{\lambda_2+\lambda_3} \varphi (-d_6(\tfrac{V}{V_B},T)) \tfrac{-1}{\sigma \sqrt{T} V} \\
			=&\, -\tfrac{1}{V} \Big( (\lambda_2-\lambda_3) (\tfrac{V_B}{V})^{\lambda_2-\lambda_3} \Phi (-d_5(\tfrac{V}{V_B},T)) + (\tfrac{V_B}{V})^{\lambda_2-\lambda_3} \tfrac{\varphi (-d_5(\frac{V}{V_B},T))}{\sigma \sqrt{T}} \\
			&\hspace{30pt}+ (\lambda_2+\lambda_3) (\tfrac{V_B}{V})^{\lambda_2+\lambda_3} \Phi (-d_6(\tfrac{V}{V_B},T)) + (\tfrac{V_B}{V})^{\lambda_2+\lambda_3} \tfrac{\varphi (-d_6(\frac{V}{V_B},T))}{\sigma \sqrt{T}} \Big).
		\end{align*}
		Hence, using equation \eqref{eq: def of i1} and the relationships $\varphi(-d_3(1,T)) = \varphi(-d_4(1,T)) = \varphi(\lambda_2 \sigma \sqrt{T})$, $\Phi(-d_4(1,T)) = \Phi(\lambda_2 \sigma \sqrt{T})$, $\varphi (-d_5(1,T))=\varphi (-d_6(1,T))=\varphi (\lambda_3\sigma\sqrt{T})$, and $1-\Phi (-d_5(1,T))=\Phi (-d_6(1,T))=\Phi (\lambda_3 \sigma \sqrt{T})$, we obtain:
		\begin{align*}
			\tfrac{\partial I_1^V(T)}{\partial V} \big|_{V=V_B} =&\, \tfrac{1}{rT} (\tfrac{\partial G^V(T)}{\partial V} \big|_{V=V_B}-e^{-rT}\tfrac{\partial F^V(T)}{\partial V} \big|_{V=V_B}) \\
			=&\, -\tfrac{1}{rTV_B} \Big( (\lambda_2-\lambda_3)(1-\Phi (\lambda_3 \sigma \sqrt{T})) + \tfrac{\varphi (\lambda_3\sigma\sqrt{T})}{\sigma \sqrt{T}} + (\lambda_2+\lambda_3) \Phi (\lambda_3 \sigma \sqrt{T}) + \tfrac{\varphi (\lambda_3\sigma\sqrt{T})}{\sigma \sqrt{T}} \\
			&\hspace{40pt}- \tfrac{e^{-rT}\varphi (\lambda_2\sigma\sqrt{T})}{\sigma \sqrt{T}} - 2 \lambda_2 e^{-rT} \Phi(\lambda_2 \sigma \sqrt{T}) - \tfrac{e^{-rT}\varphi (\lambda_2\sigma\sqrt{T})}{\sigma \sqrt{T}} \Big) \\
			=&\, -\tfrac{2}{rTV_B} \Big( \tfrac{\lambda_2-\lambda_3}{2} + \tfrac{\varphi (\lambda_3\sigma\sqrt{T})}{\sigma \sqrt{T}} + \lambda_3 \Phi (\lambda_3 \sigma \sqrt{T}) - \tfrac{e^{-rT}\varphi (\lambda_2\sigma\sqrt{T})}{\sigma \sqrt{T}} - \lambda_2 e^{-rT} \Phi(\lambda_2 \sigma \sqrt{T})\Big),
		\end{align*}
		which is the first claim since $\tfrac{\varphi (\lambda_3\sigma\sqrt{T})}{\sigma \sqrt{T}}= \tfrac{e^{-rT}\varphi (\lambda_2\sigma\sqrt{T})}{\sigma \sqrt{T}}$. Indeed, we have $\varphi (\lambda_3\sigma\sqrt{T}) = \tfrac{1}{\sqrt{2 \pi}} e^{-\frac{1}{2}\lambda_3^2 \sigma^2 T}$ and $e^{-rT}\varphi (\lambda_2\sigma\sqrt{T}) = \tfrac{1}{\sqrt{2 \pi}} e^{-\frac{1}{2}\lambda_2^2 \sigma^2 T - rT}$. Furthermore, by the definition of $\lambda_3$, we find that $\tfrac{1}{2}\lambda_3^2 \sigma^2 T = \tfrac{1}{2} \tfrac{\lambda_2^2 \sigma^4 +2r\sigma^2}{\sigma^4} \sigma^2 T = \tfrac{1}{2} \lambda_2^2 \sigma^2 T + rT$, which completes the proof of this claim.
		
		For $I_2$ (as defined in \eqref{eq: def of i2}), we get:
		\begin{align*}
			\lambda_3 \sigma \sqrt{T} \tfrac{\partial I_2^V(T)}{\partial V} =&\, (-\lambda_2+\lambda_3)V_B^{\lambda_2-\lambda_3} V^{-\lambda_2+\lambda_3-1} \Phi(-d_5(\tfrac{V}{V_B},T))d_5(\tfrac{V}{V_B},T) \\
			&+ (\tfrac{V_B}{V})^{\lambda_2-\lambda_3} \varphi(-d_5(\tfrac{V}{V_B},T))\tfrac{-d_5(\frac{V}{V_B},T)}{\sigma \sqrt{T} V} +(\tfrac{V_B}{V})^{\lambda_2-\lambda_3} \Phi(-d_5(\tfrac{V}{V_B},T))\tfrac{1}{\sigma\sqrt{T}V} \\
			&- (-\lambda_2-\lambda_3) V_B^{\lambda_2+\lambda_3} V^{-\lambda_2-\lambda_3-1} \Phi (-d_6(\tfrac{V}{V_B},T)) d_6(\tfrac{V}{V_B},T) \\
			&- (\tfrac{V_B}{V})^{\lambda_2+\lambda_3} \varphi (-d_6(\tfrac{V}{V_B},T)) \tfrac{-d_6(\frac{V}{V_B},T)}{\sigma \sqrt{T} V} - (\tfrac{V_B}{V})^{\lambda_2+\lambda_3} \Phi (-d_6(\tfrac{V}{V_B},T)) \tfrac{1}{\sigma \sqrt{T} V}.
		\end{align*}
		Thus, using the identities for $\varphi$ and $\Phi$ from above, along with $d_{5/6}(1,T) = \pm \lambda_3 \sigma \sqrt{T}$, we obtain:
		\begin{align*}
			\tfrac{\partial I_2^V(T)}{\partial V} \big|_{V=V_B} =&\, \tfrac{1}{\lambda_3 \sigma \sqrt{T} V_B} \Big( (-\lambda_2+\lambda_3) (1-\Phi (\lambda_3 \sigma \sqrt{T}))\lambda_3 \sigma \sqrt{T} -\varphi (\lambda_3\sigma\sqrt{T})\tfrac{\lambda_3 \sigma \sqrt{T}}{\sigma \sqrt{T}} + \tfrac{1-\Phi (\lambda_3 \sigma \sqrt{T})}{\sigma\sqrt{T}}\\
			&\hspace{44pt} - (\lambda_2+\lambda_3) \Phi (\lambda_3 \sigma \sqrt{T}) \lambda_3 \sigma \sqrt{T}- \varphi (\lambda_3\sigma\sqrt{T}) \tfrac{\lambda_3 \sigma \sqrt{T}}{\sigma \sqrt{T}} - \tfrac{\Phi (\lambda_3 \sigma \sqrt{T})}{\sigma \sqrt{T}} \Big) \\
			=&\, \tfrac{1}{V_B} \Big( -\lambda_2+\lambda_3 - 2 \lambda_3 \Phi (\lambda_3 \sigma \sqrt{T}) + \tfrac{1}{\lambda_3 \sigma^2 T} - \tfrac{2\Phi (\lambda_3 \sigma \sqrt{T})}{\lambda_3 \sigma^2 T} - \tfrac{2\varphi (\lambda_3\sigma\sqrt{T})}{\sigma \sqrt{T}} \Big) \\
			=&\, - \tfrac{2}{V_B} \Big( \tfrac{\lambda_2-\lambda_3}{2} - \tfrac{1}{2 \lambda_3 \sigma^2 T} + (\lambda_3+\tfrac{1}{\lambda_3 \sigma^2 T} ) \Phi (\lambda_3 \sigma \sqrt{T}) + \tfrac{\varphi (\lambda_3\sigma\sqrt{T})}{\sigma \sqrt{T}} \Big),
		\end{align*}
		which is the second claim.
	\end{myproof}
	
	\begin{lemma} \label{lemma: explicit calculation integral}
		It holds:
		\begin{enumerate}[(a)]
			\item $\int_0^T \tfrac{e^{-\nu t}}{\sqrt{t}} \varphi(\lambda_1\sigma\sqrt{t}) \diff t = \sqrt{\tfrac{1}{\lambda_1^2\sigma^2+2\nu}}(2 \Phi(\sqrt{\lambda_1^2\sigma^2+2\nu}\sqrt{T})-1)$,
			\item $\int_0^\infty \tfrac{e^{-\nu t}}{\sqrt{t}} \varphi(\lambda_1\sigma\sqrt{t}) \diff t = \sqrt{\tfrac{1}{\lambda_1^2\sigma^2+2\nu}}$,
			\item $\int_0^T e^{-\nu t} \Phi(\lambda_1\sigma\sqrt{t}) \diff t = \tfrac{1}{2\nu} - \tfrac{e^{-\nu T} \Phi(\lambda_1\sigma\sqrt{T})}{\nu} + \tfrac{\lambda_1 \sigma}{2 \nu} \sqrt{\tfrac{1}{\lambda_1^2\sigma^2+2\nu}}(2 \Phi(\sqrt{\lambda_1^2\sigma^2+2\nu}\sqrt{T})-1)$,
			\item $\int_0^\infty e^{-\nu t} \Phi(\lambda_1\sigma\sqrt{t}) \diff t = \tfrac{1}{2\nu} + \tfrac{\lambda_1 \sigma}{2 \nu} \sqrt{\tfrac{1}{\lambda_1^2\sigma^2+2\nu}}$.
		\end{enumerate}
	\end{lemma}
	
	\begin{myproof}
		We begin by proving property (a). To do so, we define the function $\text{erf}$ as $\text{erf}(x):=\tfrac{2}{\sqrt{\pi}} \int_0^x e^{-u^2} \diff u$, which is also known as the Gaussian error function. This function has the identity $\text{erf}(x)=2\Phi(\sqrt{2}x)-1$. Now, considering $\lambda_1\geq 0$, we proceed with:
		\begin{align*}
			\int_0^T \tfrac{e^{-\nu t}}{\sqrt{t}} \varphi(\lambda_1\sigma\sqrt{t}) \diff t &= \tfrac{2}{\lambda_1\sigma} \int_0^{\lambda_1\sigma\sqrt{T}} \exp\{-\tfrac{\nu}{\lambda_1^2\sigma^2} s^2\} \varphi(s) \diff s \\
			&= \tfrac{2}{\lambda_1\sigma} \cdot \tfrac{1}{\sqrt{2\pi}} \int_0^{\lambda_1\sigma\sqrt{T}} \exp\{-s^2 (\tfrac{1}{2}+\tfrac{\nu}{\lambda_1^2\sigma^2}) \} \diff s \\
			&= \tfrac{2}{\lambda_1\sigma} \cdot \sqrt{\tfrac{2\lambda_1^2\sigma^2}{\lambda_1^2\sigma^2+2\nu}} \cdot \tfrac{1}{\sqrt{2\pi}} \int_0^{\sqrt{\frac{\lambda_1^2\sigma^2+2\nu}{2\lambda_1^2\sigma^2}}\lambda_1\sigma\sqrt{T}} \exp\{-u^2 \} \diff u \\
			&= \sqrt{\tfrac{1}{\lambda_1^2\sigma^2+2\nu}} \cdot \text{erf}(\sqrt{\tfrac{\lambda_1^2\sigma^2+2\nu}{2}}\sqrt{T}),
		\end{align*}
		where we made the substitutions $s=\lambda_1\sigma\sqrt{t}$ and $u=\sqrt{\tfrac{\lambda_1^2\sigma^2+2\nu}{2\lambda_1^2\sigma^2}} \cdot s$ in the first, resp. third step. In particular, this establishes claim (a). If $\lambda_1<0$, the proof follows similarly, as the two arising negative signs cancel out (from $\text{erf}(-x)=-\text{erf}(x)$ and the calculation of $\tfrac{\sqrt{\lambda_1^2}}{\lambda_1}$). Property (b) follows by taking the limit $T \to \infty$.
		
		Next, we proceed with the proof of property (c). Using (a), 
		we obtain:
		\begin{align*}
			\int_0^T e^{-\nu t} \Phi(\lambda_1\sigma\sqrt{t}) \diff t &= \big[ \tfrac{-1}{\nu} e^{-\nu t} \Phi(\lambda_1 \sigma \sqrt{t}) \big]_0^T + \int_0^T \tfrac{1}{\nu} e^{-\nu t} \varphi(\lambda_1 \sigma \sqrt{t}) \tfrac{\lambda_1 \sigma}{2 \sqrt{t}} \diff t \\
			&= \tfrac{1}{2\nu} - \tfrac{e^{-\nu T} \Phi(\lambda_1\sigma\sqrt{T})}{\nu} + \tfrac{\lambda_1 \sigma}{2 \nu} \cdot \sqrt{\tfrac{1}{\lambda_1^2\sigma^2+2\nu}}(2 \Phi(\sqrt{\lambda_1^2\sigma^2+2\nu}\sqrt{T})-1)
		\end{align*}
		where we derived the first equation by integration by parts. Property (d) follows directly by taking the limit as $T \to \infty$.
	\end{myproof}
	
	\begin{lemma} \label{lemma: ass reformulation derivative}
		It holds that $\tfrac{-2\frac{G}{r}A_1}{rT} + 2 \tfrac{G}{r} A_2 \geq \tau_1 \frac{G}{r}(\lambda_2+\lambda_3)$ and $\int_0^T \frac{\partial c_{do} (V,k,V_B,t)}{\partial V} \big|_{V=V_B} \diff t \geq \tau_2 \int_0^\infty \frac{\partial  c_{do} (V,k,V_B,t) }{\partial V} \big|_{V=V_B} \diff t$.
	\end{lemma}
	
	\begin{proof}
		First, we observe that, by definition, the liability value of the guaranteed payment and of the surplus participation (see term $1$ and $4$ in \eqref{eq: def liability for small t}), and the associated tax benefits (see \eqref{eq: def TB1} and \eqref{eq: def TB2}) are $0$ when $V=V_B$. Therefore, the inequalities in equations \eqref{eq: ass guarantee excess} and \eqref{eq: ass surplus excess} hold even when we differentiate with respect to $V$ and evaluate at $V=V_B$ obtaining:
		\begin{align}
			\tfrac{\partial}{\partial V}\displaystyle\int_0^T \EX^\QQ \left[\int_0^t e^{-rs} g \1_{\{\min_{r \in [0,s]} V_r \geq V_B\}} \diff s \right] \diff t \,\Big|_{V=V_B}&\geq \tfrac{\partial}{\partial V}TB_1 \,\Big|_{V=V_B}, \label{eq: ass guarantee excess derivative} \\
			\tfrac{\partial}{\partial V}\displaystyle\int_0^T c_{do} (V,k,V_B,t) \diff t \,\Big|_{V=V_B} &\geq \tau_2 \tfrac{\partial}{\partial V}\displaystyle\int_0^\infty c_{do} (V,k,V_B,t) \diff t \,\Big|_{V=V_B}. \label{eq: ass surplus excess derivative}
		\end{align}
		
		From \eqref{eq: liability value L} (with $\alpha=p=0$ and $\rho=1$), we have:
		\begin{align*}
			\int_0^T \EX^\QQ \left[\int_0^t e^{-rs} g \1_{\{\min_{r \in [0,s]} V_r \geq V_B\}} \diff s \right] \diff t = \tfrac{G}{r} - \tfrac{G}{r} ( \tfrac{1-e^{-rT}}{rT} - I_1^V(T) )  - \frac{G}{r} I_2^V(T).
		\end{align*}
		Next, using the definitions of $A_1$ and $A_2$ in \eqref{eq: a1} and \eqref{eq: a2}, along with Lemma \ref{lemma: i1 i2 derivative}, we find that $\tfrac{\partial I_1^V(T)}{\partial V} \big|_{V=V_B} = - \tfrac{2A_1}{rTV_B}$ and $\tfrac{\partial I_2^V(T)}{\partial V} \big|_{V=V_B} = - \tfrac{2A_2}{V_B}$. Additionally, note that $\tfrac{\partial (\frac{V_B}{V})^{\lambda_2+\lambda_3}}{\partial V} \big|_{V=V_B} = (-\lambda_2-\lambda_3)V_B^{\lambda_2+\lambda_3} V^{-\lambda_2-\lambda_3-1} \big|_{V=V_B} = -\tfrac{(\lambda_2+\lambda_3)}{V_B}$. Using the definition of $TB_1$ from equation \eqref{eq: firm value v}, we obtain the first result after canceling the factor of $\tfrac{1}{V_B} > 0$ from both sides in \eqref{eq: ass guarantee excess derivative}.
		
		For the second claim, we are permitted in \eqref{eq: ass surplus excess derivative} to interchange the derivative and the integral sign by the Dominated Convergence Theorem, as established by Lemma \ref{lemma: cdo unif bounded L1}. This directly leads to the conclusion.
	\end{proof}
	
	\begin{lemma} \label{lemma: term bigger 0 for vb}
		It holds:
		\begin{align*}
			1 + \rho(\lambda_2+\lambda_3)+2(1-\rho)[\tfrac{\lambda_2-\lambda_3}{2} - \tfrac{1}{2 \lambda_3 \sigma^2 T} + (\lambda_3+\tfrac{1}{\lambda_3 \sigma^2 T} ) \Phi (\lambda_3 \sigma \sqrt{T}) + \tfrac{\varphi (\lambda_3\sigma\sqrt{T})}{\sigma \sqrt{T}}] > 0.
		\end{align*}
	\end{lemma}
	
	\begin{myproof}
		It holds:
		\begin{align*}
			B:=&\,2[\tfrac{\lambda_2-\lambda_3}{2} - \tfrac{1}{2 \lambda_3 \sigma^2 T} + (\lambda_3+\tfrac{1}{\lambda_3 \sigma^2 T} ) \Phi (\lambda_3 \sigma \sqrt{T}) + \tfrac{\varphi (\lambda_3\sigma\sqrt{T})}{\sigma \sqrt{T}}] \\
			=&\, \lambda_2+\lambda_3 - 2\lambda_3 + 2\lambda_3 \Phi (\lambda_3 \sigma \sqrt{T}) - \tfrac{1}{\lambda_3 \sigma^2 T} + \tfrac{2}{\lambda_3 \sigma^2 T} \Phi (\lambda_3 \sigma \sqrt{T}) + 2 \lambda_3 \tfrac{\varphi (\lambda_3\sigma\sqrt{T})}{\lambda_3 \sigma \sqrt{T}} \\
			=&\, (\lambda_2+\lambda_3) + 2\lambda_3 (\Phi (\lambda_3 \sigma \sqrt{T})+\tfrac{\varphi (\lambda_3\sigma\sqrt{T})}{\lambda_3 \sigma \sqrt{T}}-1) + \tfrac{1}{\lambda_3 \sigma^2 T} (2\Phi (\lambda_3 \sigma \sqrt{T})-1).
		\end{align*}
		By definition, $\lambda_3 \geq 0$. Moreover, since $\lambda_3 = \tfrac{\sqrt{(\lambda_2 \sigma^2)^2+2r\sigma^2}}{\sigma^2} = \sqrt{\lambda_2^2+\tfrac{2r}{\sigma^2}}$, we see that $\lambda_3 > |\lambda_2|$ given that $r>0$ and $\sigma>0$. In particular, this implies that $\lambda_2+\lambda_3>0$. Therefore, since $\lambda_3 \sigma \sqrt{T}>0$, we have $2\Phi (\lambda_3 \sigma \sqrt{T})-1>0$. For the middle term, the situation is more intricate. Let $x > 0$ and define $h(x):=\Phi(x)+\tfrac{\varphi(x)}{x}-1$. Now, $h'(x) = \varphi(x) + \tfrac{-\varphi(x)x^2-\varphi(x)}{x^2} = - \tfrac{\varphi(x)}{x^2} <0$, meaning that $h$ is decreasing in $x$. Since $\lim_{x \to \infty} h(x) = 0$, we conclude that $h(x)>0$ for all $x \in \Real$. Therefore, we obtain the expression that $2\lambda_3 (\Phi (\lambda_3 \sigma \sqrt{T})+\tfrac{\varphi (\lambda_3\sigma\sqrt{T})}{\lambda_3 \sigma \sqrt{T}}-1)>0$, since $\lambda_3 \sigma \sqrt{T}>0$. Thus, $B>0$. 
		
		Hence, it follows:
		\begin{align*}
			&\, 1 + \rho(\lambda_2+\lambda_3)+2(1-\rho)[\tfrac{\lambda_2-\lambda_3}{2} - \tfrac{1}{2 \lambda_3 \sigma^2 T} + (\lambda_3+\tfrac{1}{\lambda_3 \sigma^2 T} ) \Phi (\lambda_3 \sigma \sqrt{T}) + \tfrac{\varphi (\lambda_3\sigma\sqrt{T})}{\sigma \sqrt{T}}] \\
			=&\, 1 + \rho(\lambda_2+\lambda_3)+ (1-\rho) B >0,
		\end{align*}
		since $\lambda_2+\lambda_3>0$, $B>0$ as discussed above, and $\rho \in [0,1]$.
	\end{myproof}
	
	\begin{lemma} \label{lemma: a1,a2,a3,a4>0}
		It holds that $A_1, A_2 >0$ and $A_3 > A_4 > 0$ with $A_1$, $A_2$, $A_3$, and $A_4$ defined as in equations \eqref{eq: a1}, \eqref{eq: a2}, \eqref{eq: a3} and \eqref{eq: a4}.
	\end{lemma}
	
	\begin{myproof}
		We begin by proving that $A_1>0$: First, we rewrite $A_1$ as $A_1 = \lambda_3 (\Phi(\lambda_3\sigma\sqrt{T})-\tfrac{1}{2}) - \lambda_2 (e^{-rT}\Phi(\lambda_2\sigma\sqrt{T})-\tfrac{1}{2})$. Now, if $\lambda_2 \geq 0$, we obtain the inequality $A_1 \geq \lambda_3 (\Phi(\lambda_3\sigma\sqrt{T})-\tfrac{1}{2}) - \lambda_2 (\Phi(\lambda_2\sigma\sqrt{T})-\tfrac{1}{2})$, since $e^{-rT}<1$. Since $r,\sigma>0$, it follows from the definition that $\lambda_3>\lambda_2$. This implies that $\Phi(\lambda_3\sigma\sqrt{T})-\tfrac{1}{2}>\Phi(\lambda_2\sigma\sqrt{T})-\tfrac{1}{2}$, because $\Phi(\cdot)$ is an increasing function. Therefore, we conclude that $A_1>0$ when $\lambda_2\geq0$. Next, we consider the case where $\lambda_2<0$. In this case, $A_1>0$ is equivalent to $\lambda_3 (\Phi(\lambda_3\sigma\sqrt{T})-\tfrac{1}{2}) > |\lambda_2| (\tfrac{1}{2} - e^{-rT}\Phi(\lambda_2\sigma\sqrt{T}))$. If $\tfrac{1}{2} - e^{-rT}\Phi(\lambda_2\sigma\sqrt{T}) \leq 0$, this inequality clearly holds, so assume that $\tfrac{1}{2} - e^{-rT}\Phi(\lambda_2\sigma\sqrt{T}) >0$. Since $r,\sigma>0$, we have $\lambda_3>|\lambda_2|$, which implies that $A_1>0$ if $\Phi(\lambda_3\sigma\sqrt{T})+ e^{-rT}\Phi(\lambda_2\sigma\sqrt{T}))-1 \geq 0$. This inequality is equivalent to 
		\begin{align} \label{eq: equivalence in a1 proof}
			e^{-rT}\Phi(-|\lambda_2|\sigma\sqrt{T})) \geq \Phi(-\lambda_3\sigma\sqrt{T}),
		\end{align}
		since $\Phi(-x)=1-\Phi(x)$ for all $x \in \Real$. Using the substitution $s = \sqrt{t^2+2rT}$, we can derive the following identity for any $x>0$, due to $e^{-t^2}$ being symmetric:
		\begin{align*}
			e^{-rT} \Phi(-x) = \tfrac{e^{-rT}}{\sqrt{2\pi}} \displaystyle\int_{-\infty}^{-x} e^{-\frac{t^2}{2}} \diff t = \tfrac{1}{\sqrt{2\pi}} \displaystyle\int_{x}^{\infty} e^{-\frac{1}{2}(t^2+2rT)} \diff t = \tfrac{1}{\sqrt{2\pi}} \displaystyle\int_{\sqrt{x^2+2rT}}^{\infty} e^{-\frac{s^2}{2}} \tfrac{s}{\sqrt{s^2-2rT}} \diff s.
		\end{align*}
		Thus, we have $e^{-rT}\Phi(-|\lambda_2|\sigma\sqrt{T})) = \tfrac{1}{\sqrt{2\pi}} \int_{\sqrt{\lambda_2^2\sigma^2 T+2rT}}^{\infty} e^{-\frac{s^2}{2}} \tfrac{s}{\sqrt{s^2-2rT}} \diff s = \tfrac{1}{\sqrt{2\pi}} \int_{\lambda_3\sigma\sqrt{T}}^{\infty} e^{-\frac{s^2}{2}} \linebreak \tfrac{s}{\sqrt{s^2-2rT}} \diff s$, since $\lambda_2^2\sigma^2 T+2rT = \lambda_3^2 \sigma^2 T$. Indeed, it holds by definition that $\lambda_3^2 \sigma^2 = \tfrac{\lambda_2^2 \sigma^4+2r\sigma^2}{\sigma^4} \sigma^2 = \lambda_2^2 \sigma^2 + 2r$. Multiplying by $T$ yields this intermediate statement. Now, a simple rewriting leads to $\Phi(-\lambda_3\sigma\sqrt{T}) = \tfrac{1}{\sqrt{2\pi}} \int_{\lambda_3\sigma\sqrt{T}}^{\infty} e^{-\frac{s^2}{2}} \diff s$. Hence, \eqref{eq: equivalence in a1 proof} is equivalent to 
		\begin{align*}
			&\hspace{-60pt}&\tfrac{1}{\sqrt{2\pi}} \displaystyle\int_{\lambda_3\sigma\sqrt{T}}^{\infty} e^{-\frac{s^2}{2}} \tfrac{s}{\sqrt{s^2-2rT}} \diff s &\geq \tfrac{1}{\sqrt{2\pi}} \displaystyle\int_{\lambda_3\sigma\sqrt{T}}^{\infty} e^{-\frac{s^2}{2}} \diff s, \\
			&\Leftrightarrow\hspace{-60pt}& \int_{\lambda_3\sigma\sqrt{T}}^{\infty} e^{-\frac{s^2}{2}} (\tfrac{s}{\sqrt{s^2-2rT}}-1) \diff s &\geq 0.
		\end{align*}
		Since $\tfrac{s}{\sqrt{s^2-2rT}} \geq 1$ (due to $s \geq \lambda_3 \sigma \sqrt{T}>0$ and $(\lambda_3\sigma\sqrt{T})^2-2rT = \lambda_2^2 \sigma^2 T >0$), the last inequality is indeed correct. Therefore, $A_1>0$ follows.
		
		The property that $A_2>0$ is directly implied by the proof of Lemma \ref{lemma: term bigger 0 for vb}, where $B=2A_2$ with $B$ defined as in the proof of Lemma \ref{lemma: term bigger 0 for vb}.
		
		Next, we show that $A_3, A_4>0$: Using Lemma \ref{lemma: explicit calculation integral} and the fact that $\Big(\tfrac{\lambda_1^2 \sigma}{\nu} + \tfrac{2}{\sigma}\Big) \sqrt{\tfrac{1}{\lambda_1^2\sigma^2+2\nu}} = \tfrac{1}{\sigma \nu} \sqrt{\lambda_1^2 \sigma^2 + 2\nu}$, we have:
		\begin{align}
			A_3 &= \int_0^\infty \Big(e^{-\nu t}(2\lambda_1 \Phi(\lambda_1 \sigma \sqrt{t})+ \tfrac{2\varphi(\lambda_1 \sigma \sqrt{t})}{\sigma \sqrt{t}} )\Big) \diff t, \label{eq: a3 equality integral} \\
			A_4 &= \displaystyle\int_0^T \Big(e^{-\nu t}(2\lambda_1 \Phi(\lambda_1 \sigma \sqrt{t})+ \tfrac{2\varphi(\lambda_1 \sigma \sqrt{t})}{\sigma \sqrt{t}} )\Big) \diff t. \label{eq: a4 equality integral}
		\end{align}
		We now show that $\lambda_1 \Phi(\lambda_1 \sigma \sqrt{t})+ \tfrac{\varphi(\lambda_1 \sigma \sqrt{t})}{\sigma \sqrt{t}}>0$ for all $\lambda_1 \in \Real$, $\sigma>0$, and $t\geq 0$, which directly implies the claim. First, if $\lambda_1 \geq 0$, the result is trivial. If $\lambda_1 < 0$, we obtain: $\lambda_1 \Phi(\lambda_1 \sigma \sqrt{t})+ \tfrac{\varphi(\lambda_1 \sigma \sqrt{t})}{\sigma \sqrt{t}} = \lambda_1 [\Phi(\lambda_1 \sigma \sqrt{t})+ \tfrac{\varphi(\lambda_1 \sigma \sqrt{t})}{\lambda_1\sigma \sqrt{t}}]$. Let $y := \lambda_1 \sigma \sqrt{t} <0$, and define $h(y):= \Phi(y)+\tfrac{\varphi(y)}{y}$. Since $\lambda_1<0$, we only need to show that $h(y)<0$ for all $y \in (-\infty,0]$. First, we obtain that $\lim_{y \to -\infty} h(y) = 0$ and $\lim_{y \to0_-} h(y) =-\infty$. Additionally, $h'(y) = \varphi(y) + \tfrac{-y\varphi(y)y-\varphi(y)}{y^2} = - \tfrac{\varphi(y)}{y^2} < 0$ which yields the claim. 
		
		Finally, we conclude that $A_3 > A_4$, which follows immediately from the fact that $\lambda_1 \Phi(\lambda_1 \sigma \sqrt{t})+ \tfrac{\varphi(\lambda_1 \sigma \sqrt{t})}{\sigma \sqrt{t}}>0$ and $T<\infty$.
	\end{myproof}
	
	\begin{lemma} \label{lem: vb neq 0}
		In our framework, the optimal bankruptcy-triggering is bounded away from $0$.
	\end{lemma}
	
	\begin{proof}
		Assume, for the sake of contradiction, that $V_B=0$. From \eqref{eq: def liability for small t}, \eqref{eq: liability value L}, \eqref{eq: firm value v}, and \eqref{equity}, it follows that both the bankruptcy costs and the value of bankruptcy vanish under this assumption. Consequently, we obtain
		\begin{align*}
			E(V) = V + TB_1(V)+TB_2(V)-VG(V)-VP(V)-VSP(V),
		\end{align*}
		where $TB_1$ and $TB_2$ denote the tax benefits associated with the guarantee and the surplus participation (cf. \eqref{eq: def TB1}, \eqref{eq: def TB2}), $VG$ the value of the guarantee, $VP$ the value of the principal payment, and $VSP$ the value of the surplus participation. By assumption, the guarantee and the surplus participation is worth at least as much as its corresponding tax benefit, i.e., $TB_1 \leq VG$ and $TB_2 \leq VSP$. Moreover, since $V_B=0 < V_S$ for all $s$, it follows that $F^V, G^V \equiv 0$ and hence $I_1^V,I_2^V \equiv 0$. Thus, we obtain that $E(V;0,T) \leq V - VP(V) =  V - P \tfrac{1-e^{-rT}}{rT}$. As $P > 0$ by assumption, this inequality implies that as $V \to V_B=0$, equity becomes negative. This contradicts the requirement that shareholders would liquidate the firm before equity turns negative. Therefore, the assumption $V_B=0$ is inconsistent, and the boundary solution cannot represent the optimal bankruptcy-triggering value. This proof also shows that $V_B$ is bounded away from $0$.
	\end{proof}
	
	\begin{lemma} \label{lemma: uniqueness preparation lemma}
		Let $R(V)$ be the right-hand side of \eqref{eq: formula for vb} with $V=V_B$. If $\lim_{V \to \infty} R(V) > 0$ (resp. $\downarrow 0$), then $R$ has a unique zero root. If $\lim_{V \to \infty} R(V) < 0$ (resp. $\uparrow 0$), then $R$ has at most two zero roots. If $R(V)=0$ for all $V$ large enough, then $R(V)=0$ for all $V \geq k$ and there is up to one zero root smaller than $k$.
		
		Notation: $\lim_{V \to \infty} R(V) \downarrow 0$ (resp. $\uparrow 0$), if $\lim_{V \to \infty} R(V) = 0$ and there exists a $\tilde{V}$ such that $R(V)<0$ (resp. $R(V)>0$) for all $V \geq \tilde{V}$.
	\end{lemma}
	
	\begin{proof}
		Taking the derivative yields $R'(V) = \tfrac{1}{V^2} H(V)$, as follows from \eqref{eq: dv2 cdo<} and \eqref{eq: dv2 cdo>}, where
		\begin{align*}
			H(V) :=&\, \Big( \tfrac{2(P-\frac{G}{r})A_1}{rT} + 2 \tfrac{G}{r} A_2 - \tau_1 \frac{G}{r}(\lambda_2+\lambda_3) \Big) \\
			&+ \tau_2 \alpha 2k \int_T^\infty e^{-rt} \Big(\lambda_2  \Phi(d_2(\min\{\tfrac{V}{k},1\},t)) + \frac{ \varphi(d_2(\min\{\frac{V}{k},1\},t))}{\sigma \sqrt{t}} \Big) \diff t \\
			&- \alpha (1-\tau_2) 2k \displaystyle\int_0^T e^{-rt} \Big(\lambda_2  \Phi(d_2(\min\{\tfrac{V}{k},1\},t)) + \frac{ \varphi(d_2(\min\{\frac{V}{k},1\},t))}{\sigma \sqrt{t}} \Big) \diff t.
		\end{align*}
		
		First, observe that $\lim_{V \to 0} H(V) = \Big( \tfrac{2(P-\frac{G}{r})A_1}{rT} + 2 \tfrac{G}{r} A_2 - \tau_1 \frac{G}{r}(\lambda_2+\lambda_3) \Big) > 0$ by Lemma \ref{lemma: ass reformulation derivative}, and that $H$ is constant for $V \geq k$. Moreover, we note that if $H(k)=0$, then $R(V)=0$ for all $V \geq k$.
		
		In what follows, we establish the stated properties by considering several cases. 
		First, in (i), we treat the case in which $H$ has no zeros in the interval $(0,k]$. 
		Next, in (ii), we analyze the case where $H$ has exactly one zero in $(0,k)$, 
		distinguishing the subcases (ii.a) $H(k) \neq 0$ and (ii.b) $H(k) = 0$. 
		Finally, in (iii), we show that $H$ cannot have more than one zero in $(0,k)$. 
		The argument is completed by noting that $H$ is constant on $[k,\infty)$.
		
		(i) If $H$ has no zero in $(0,k]$, then $H$ has no zero in $(0,\infty)$, i.e. $H > 0$, and consequently $R$ is strictly increasing in $V$. Hence, if $\lim_{V \to \infty} R(V) \leq 0$, then no zero root exists, and if $\lim_{V \to \infty} R(V) > 0$, a unique zero root exists.
		
		(ii.a) Let $H(k) \neq 0$. If $H$ has exactly one zero in $(0,k)$, then there exists $\tilde{V}$ such that $H(V) > 0$ for all $V < \tilde{V}$ and $H(V) < 0$ for all $V > \tilde{V}$. By the first part of the proof of Theorem \ref{th: VB determination} (using that $R=h_2$ with $h_2$ as in \eqref{eq: def of h2}), we have $\lim_{V \to 0} R(V) = -\infty$. If $\lim_{V \to \infty} R(V) > 0$ (resp. $\downarrow 0$), the zero root of \eqref{eq: formula for vb} is unique. Indeed, in intervals where $R$ is strictly increasing or decreasing, there can be at most one zero. Since $R$ decreases to a positive value, no additional zero can occur in this region where $R$ is decreasing. Hence, $R$ has a unique zero (existence follows by the intermediate value theorem) as $H$ is continuous and constant for $V\geq k$. If $H$ would have a zero at $V=k$, then $H(V) = 0$ for all $V \geq k$, i.e. $R$ is constant for $V \geq k$. Since we assumed here that $\lim_{V \to \infty} R(V) > 0$ (resp. $\downarrow 0$), we conclude that $R(V) > 0$ for all $V \geq k$. Thus, if $H$ has at most one zero in $(0,k)$, then $R$ has at most one local extremum in $(0,k)$\footnote{We adopt the convention that local extrema cannot occur at the boundary points of open intervals.}, which must be a maximum since $\lim_{V \to 0} R(V) = -\infty$. Therefore, $R(V) > 0$ for all $V \geq k$ implies that the zero of $R$ in $(0,\infty)$ is unique. Now, if $\lim_{V \to \infty} R(V) < 0$ (resp. $\uparrow 0$), there can be at most two zero roots of $R$ with a similar argumentation as before. (ii.b) Furthermore, if $H(k)=0$, i.e., $H(V) = 0$ and $R(V)=0$ for all $V \geq k$, then a similar argument shows that there is at most one zero root in $(0,k)$.
		
		(iii) Finally, we show that $H$ has no more than one zero in $(0,k)$. Therefore, let $\varepsilon \in (0,k)$ be small and consider $H$ on $(0,k-\varepsilon]$. Then, it holds with $\tfrac{d_2(\frac{V}{k},t)}{\sigma \sqrt{t}} = \tfrac{\ln (\frac{V}{k})}{\sigma^2 t} + \lambda_2$ and $\varphi'(x) = -x\varphi(x)$ that
		\begin{align*}
			\tfrac{\partial}{\partial V} \Big(\lambda_2  \Phi(d_2(\tfrac{V}{k},t)) + \frac{ \varphi(d_2(\frac{V}{k},t))}{\sigma \sqrt{t}} \Big) &= \lambda_2  \varphi(d_2(\tfrac{V}{k},t)) \frac{1}{\sigma \sqrt{t} V} -  \varphi(d_2(\frac{V}{k},t)) \frac{d_2(\frac{V}{k},t)}{\sigma \sqrt{t}} \frac{1}{\sigma \sqrt{t} V} \\
			&= \tfrac{(-\ln (\frac{V}{k}))}{\sigma^3 t^{3/2} V} \varphi(d_2(\frac{V}{k},t)) =: \hat{f} (t,V).
		\end{align*}
		By Lemma \ref{lemma: varphi ln 0}, $\lim_{V \to 0} \hat{f}(t,V) = 0$. With $\tilde{\varepsilon} := -\ln\!\big(\tfrac{k-\varepsilon}{k}\big) > 0$, we obtain
		\begin{align*}
			\lim_{V \to k-\varepsilon} \hat{f}(t,V) = \tfrac{\tilde{\varepsilon}}{\sigma^3 t^{3/2} (k-\varepsilon)} \tfrac{1}{\sqrt{2\pi}} e^{-\frac{1}{2} (\frac{\tilde{\varepsilon}^2}{\sigma^2 t} -2\lambda_2\tilde{\varepsilon}+\lambda_2^2 \sigma^2 t)}.
		\end{align*}
		Hence, there exists $\hat{F}(t)$ with $|\hat{f}(t,V)| \leq \hat{F}(t)$ for all $V \in (0,k-\varepsilon]$ and $\int_0^\infty \hat{F}(t)\,\mathrm{d}t < \infty$, so that the derivative and the integral can be interchanged. Therefore, for $V \in (0,k-\varepsilon]$ we have
		\begin{align*}
			H'(V) = 2k\alpha\tau_2 \int_{0}^\infty \tfrac{e^{-rt}(-\ln (\frac{V}{k}))}{\sigma^3 t^{3/2} V} \varphi(d_2(\frac{V}{k},t)) \diff t - 2k\alpha \displaystyle\int_{0}^T \tfrac{e^{-rt}(-\ln (\frac{V}{k}))}{\sigma^3 t^{3/2} V} \varphi(d_2(\frac{V}{k},t)) \diff t.
		\end{align*}
		Now define:
		\begin{align*}
			\eta (V) := \frac{2k\alpha \int_{0}^T \tfrac{e^{-rt}(-\ln (\frac{V}{k}))}{\sigma^3 t^{3/2} V} \varphi(d_2(\frac{V}{k},t)) \diff t}{2k\alpha\tau_2 \int_{0}^\infty \tfrac{e^{-rt}(-\ln (\frac{V}{k}))}{\sigma^3 t^{3/2} V} \varphi(d_2(\frac{V}{k},t)) \diff t} = \frac{ \int_{0}^T \tfrac{e^{-rt}}{t^{3/2}} \varphi(d_2(\frac{V}{k},t)) \diff t}{\tau_2 \int_{0}^\infty \tfrac{e^{-rt}}{t^{3/2}} \varphi(d_2(\frac{V}{k},t)) \diff t} =: \frac{B_1(V)}{B_2(V)}.
		\end{align*}
		Next, we show that $\eta$ is increasing, i.e. $\eta'>0$. Interchanging differentiation and integration is again justified by the same reasoning as above. We obtain $\eta' = \frac{B_1'B_2 - B_1B_2'}{B_2^2}$, so that $\eta'>0$ if and only if $B_1'B_2 - B_1B_2'>0$. Indeed,
		\begin{align*}
			B_1'(V)B_2(V) &= \tfrac{\tau_2}{\sigma V} \int_0^T \int_0^\infty \tfrac{e^{-rt}}{t^{2}} \varphi(d_2(\tfrac{V}{k},t)) (-d_2(\tfrac{V}{k},t)) \tfrac{e^{-rs}}{s^{3/2}} \varphi(d_2(\tfrac{V}{k},s)) \diff s \, \diff t, \\
			B_1(V)B_2'(V) &= \tfrac{\tau_2}{\sigma V} \int_0^T \int_0^\infty \tfrac{e^{-rt}}{t^{3/2}} \varphi(d_2(\tfrac{V}{k},t)) \tfrac{e^{-rs}}{s^{2}} \varphi(d_2(\tfrac{V}{k},s)) (-d_2(\tfrac{V}{k},s)) \diff s \, \diff t.
		\end{align*}
		Splitting the integral from $0$ to $\infty$ at $T$, the contributions over $[0,T]^2$ cancel, leaving
		\begin{align*}
			B_1'B_2-B_1B_2' 
			=&\, \tfrac{\tau_2}{\sigma V} \int_0^T \int_T^\infty \frac{e^{-rt}e^{-rs}}{s^{3/2}t^{3/2}} \varphi(d_2(\tfrac{V}{k},t)) \varphi(d_2(\tfrac{V}{k},s)) 
			\Big[\frac{d_2(\frac{V}{k},s)}{\sqrt{s}} - \frac{d_2(\frac{V}{k},t)}{\sqrt{t}}\Big] \diff s \, \diff t >0,
		\end{align*}
		since $0<t<s$ and $\psi(t) := \tfrac{d_2(\frac{V}{k},t)}{\sqrt{t}} = \tfrac{\ln(\frac{V}{k})}{t} + \lambda_2\sigma$ is increasing in $t$ (as $\ln(\tfrac{V}{k})<0$).  
		
		It follows that $H'(V) = 2k\alpha\tau_2 \int_{0}^\infty \tfrac{e^{-rt}(-\ln (\frac{V}{k}))}{\sigma^3 t^{3/2} V} \varphi(d_2(\frac{V}{k},t)) \diff t (1-\eta(V))$. Since $k$, $\alpha$, $\tau_2$, $\sigma$, $(-\ln (\frac{V}{k}))$, and $V$ are all non-negative, the sign of $H'$ is determined solely by the factor $(1-\eta)$, i.e., $\sign(H') = \sign (1-\eta)$. As $(1-\eta)$ is strictly decreasing, it follows that $H'(V)$ can undergo at most one sign change on $(0,k-\varepsilon]$, and this change must be from positive to negative. Passing to the limit as $\varepsilon \to 0$ yields that $H'$ admits at most one sign change on the full interval $(0,k)$, again necessarily from positive to negative. Consequently, $H$ may have at most one local maximum and no local minimum in $(0,k)$. Furthermore, by Lemma \ref{lemma: varphi bounded}, Lemma \ref{lemma: ass reformulation derivative}, and using the facts that $\lim_{x \to -\infty} \Phi(x)=0$ and $|\Phi|\leq 1$, we obtain that $\lim_{V \to 0} H(V) = \Big( \tfrac{2(P-\frac{G}{r})A_1}{rT} + 2 \tfrac{G}{r} A_2 - \tau_1 \tfrac{G}{r}(\lambda_2+\lambda_3) \Big) >0$. Since $H$ starts positive at the boundary $V=0$, possesses at most one local maximum, and no local minimum in $(0,k)$, it can have at most one zero in this interval.
	\end{proof}
	
	\begin{lemma} \label{lemma: long term >0}
		The inequalities $1+\rho(\lambda_2+\lambda_3)+2(1-\rho)A_2>0$ and $\tfrac{2(P-\frac{G}{r})A_1}{rT} + 2 \tfrac{G}{r} A_2 - \tau_1 \frac{G}{r}(\lambda_2+\lambda_3) > 0$ hold, where $\lambda_2,\lambda_3$ are defined in equation \eqref{eq: def lambda23}. 
	\end{lemma}
	
	\begin{myproof}
		The first part of the claim has already been established in Lemma \ref{lemma: term bigger 0 for vb} by inserting the definition of $A_2$ from equation \eqref{eq: a2}. The second part of the lemma follows directly from Lemma \ref{lemma: ass reformulation derivative}, with the additional observation that $A_1>0$ as established in Lemma \ref{lemma: a1,a2,a3,a4>0}, and that $P>0$ by definition.
	\end{myproof}

	\begin{lemma} \label{lemma: dv2 cdo bounded l1}
		It holds that $h(V_B):=\tfrac{\partial c_{do} (V,k,V_B,t)}{\partial V} \big|_{V=V_B}$ is continuously differentiable. Furthermore, for all $k\geq 0$ and $T>0$, there exists a constant $C>0$ such that $0 \leq \tfrac{\partial}{\partial V_B}[\tfrac{\partial c_{do} (V,k,V_B,T)}{\partial V} \big|_{V=V_B}] \leq C (e^{-rT}+e^{-\nu T})$ for all $V_B>0$. 
	\end{lemma}
	
	\begin{myproof}
		To establish the continuous differentiability of $h$, it suffices to verify that the derivative is continuous at $V=k$, as these properties follow directly for all other points. From equations \eqref{eq: dv cdo< v=vb} and \eqref{eq: dv cdo> v=vb} resp. \eqref{eq: dv cdo v=vb}, we observe using $2-2\lambda_1=-2\lambda_2$:
		\begin{align}
			\tfrac{\partial}{\partial V_B}[\tfrac{\partial}{\partial V} c_{do}^{V_B \leq k} (V,k,V_B,t) \big|_{V=V_B}] =&\, 2e^{-\nu t} \Big(  \tfrac{\lambda_1 \varphi(d_1(\frac{V_B}{k},t))}{\sigma \sqrt{t}V_B}+ \frac{-\varphi(d_1(\frac{V_B}{k},t))d_1(\frac{V_B}{k},t)}{\sigma^2 t V_B} \Big) \notag \\
			&+ \tfrac{2ke^{-rt}}{V_B^2} \Big(\lambda_2  \Phi(d_2(\tfrac{V_B}{k},t)) + \frac{ \varphi(d_2(\frac{V_B}{k},t))}{\sigma \sqrt{t}} \Big) \notag \\
			&- \tfrac{2ke^{-rt}}{V_B} \Big(\tfrac{\lambda_2 \varphi(d_2(\frac{V_B}{k},t))}{\sigma\sqrt{t}V_B} + \frac{ -\varphi(d_2(\frac{V_B}{k},t))d_2(\frac{V_B}{k},t)}{\sigma^2 t V_B} \Big) \notag \\
			=&\, \tfrac{2e^{-\nu t}}{\sigma\sqrt{t}V_B} \Big( \lambda_1 \varphi(d_1(\frac{V_B}{k},t))- \frac{\varphi(d_1(\frac{V_B}{k},t))(\ln(\frac{V_B}{k})+\lambda_1\sigma^2 t)}{\sigma^2 t} \Big) \notag \\
			&+ \tfrac{2ke^{-rt}}{V_B^2} \Big(\lambda_2  \Phi(d_2(\tfrac{V_B}{k},t)) + \frac{ \varphi(d_2(\frac{V_B}{k},t))}{\sigma \sqrt{t}} \Big) \notag \\
			&- \tfrac{2ke^{-rt}}{\sigma \sqrt{t}V_B^2} \Big(\lambda_2 \varphi(d_2(\frac{V_B}{k},t)) - \frac{\varphi(d_2(\frac{V_B}{k},t)) (\ln(\frac{V_B}{k})+\lambda_2\sigma^2 t)}{\sigma^2 t} \Big) \notag \\
			=&\, -\tfrac{2e^{-\nu t} \ln(\frac{V_B}{k})}{\sigma^3 t\sqrt{t}V_B} \varphi(d_1(\frac{V_B}{k},t)) \notag \\
			&+ \tfrac{2ke^{-rt}}{V_B^2} \Big(\lambda_2  \Phi(d_2(\tfrac{V_B}{k},t)) + \frac{ \varphi(d_2(\frac{V_B}{k},t))}{\sigma \sqrt{t}} (1+\tfrac{\ln(\frac{V_B}{k})}{\sigma^2t}) \Big) \notag \\
			=&\, \tfrac{2ke^{-rt}}{V_B^2} \Big(\lambda_2  \Phi(d_2(\tfrac{V_B}{k},t)) + \frac{ \varphi(d_2(\frac{V_B}{k},t))}{\sigma \sqrt{t}} \Big), \label{eq: dv2 cdo<} \\
			\tfrac{\partial}{\partial V_B} [\tfrac{\partial}{\partial V} c_{do}^{V_B \geq k} (V,k,V_B,t) \big|_{V=V_B}] =&\, \tfrac{2ke^{-rt}}{V_B^2} \Big(\lambda_2  \Phi(d_2(1,t)) + \frac{ \varphi(d_2(1,t))}{\sigma \sqrt{t}} \Big). \label{eq: dv2 cdo>}
		\end{align}
		The last step in \eqref{eq: dv2 cdo<} holds since:
		\begin{align*}
			\tfrac{2e^{-\nu t} \ln(\frac{V_B}{k})}{\sigma^3 t\sqrt{t}V_B} \varphi(d_1(\frac{V_B}{k},t)) =&\, \tfrac{2e^{-\nu t} \ln(\frac{V_B}{k})}{\sigma^3 t\sqrt{t}V_B} \cdot \tfrac{1}{\sqrt{2\pi}} e^{-\frac{1}{2 \sigma^2 t} [ \ln(\frac{V_B}{k})^2 + 2 \ln(\frac{V_B}{k}) (r-\nu+\frac{\sigma^2}{2})t+((r-\nu)^2+(r-\nu)\sigma^2+\frac{\sigma^4}{4})t^2 ]} \\
			=&\, \tfrac{2ke^{-r t} \ln(\frac{V_B}{k})}{\sigma^3 t\sqrt{t}V_B^2} \cdot e^{\ln(\frac{V_B}{k})} \cdot  e^{(r-\nu)t} \\
			&\cdot \tfrac{1}{\sqrt{2\pi}}e^{-\frac{1}{2 \sigma^2 t} [ \ln(\frac{V_B}{k})^2 + 2 \ln(\frac{V_B}{k}) (r-\nu-\frac{\sigma^2}{2})t + ((r-\nu)^2-(r-\nu)\sigma^2+\frac{\sigma^4}{4})t^2 ]} \\
			&\cdot e^{-\frac{1}{2 \sigma^2 t} [2\ln(\frac{V_B}{k})\sigma^2t+2(r-\nu)\sigma^2t^2]} \\
			=&\, \tfrac{2ke^{-r t} \ln(\frac{V_B}{k})}{\sigma^3 t\sqrt{t}V_B^2} \cdot e^{\ln(\frac{V_B}{k})+(r-\nu)t} \cdot e^{-\ln(\frac{V_B}{k})-(r-\nu)t} \\
			&\cdot \tfrac{1}{\sqrt{2\pi}} e^{-\frac{1}{2} (\frac{\ln(\frac{V_B}{k})+(r-\nu-\frac{\sigma^2}{2})t}{\sigma \sqrt{t}})^2} \\
			=&\, \tfrac{2ke^{-r t} \ln(\frac{V_B}{k})}{\sigma^3 t\sqrt{t}V_B^2} \varphi(d_2(\frac{V_B}{k},t)).
		\end{align*}
		Now, evaluating \eqref{eq: dv2 cdo<} and \eqref{eq: dv2 cdo>} at $V_B=k \neq 0$ yields:
		\begin{align*}
			\tfrac{\partial}{\partial V_B}[\tfrac{\partial}{\partial V} c_{do}^{V_B \leq k} (V,k,V_B,t) \big|_{V=V_B}] \big|_{V_B=k}
			=&\, \tfrac{2e^{-rt}}{k} \Big(\lambda_2  \Phi(d_2(1,t)) + \frac{ \varphi(d_2(1,t))}{\sigma \sqrt{t}} \Big), \\
			\tfrac{\partial}{\partial V_B} [\tfrac{\partial}{\partial V} c_{do}^{V_B \geq k} (V,k,V_B,t) \big|_{V=V_B}] \big|_{V_B=k}=&\, \tfrac{2e^{-rt}}{k} \Big(\lambda_2  \Phi(d_2(1,t)) + \frac{ \varphi(d_2(1,t))}{\sigma \sqrt{t}} \Big).
		\end{align*}
		If $k=0$, Lemma \ref{lemma: cdo limit of derivative to 0 and infty}\eqref{lemma, item: k=0 limit} implies that $\tfrac{\partial}{\partial V_B}[\tfrac{\partial}{\partial V} c_{do} (V,k,V_B,t) \big|_{V=V_B}] = 0$. This shows that $h$ is continuously differentiable.
		
		For the second claim, we need to show that the limits of $\tfrac{\partial}{\partial V_B}[\tfrac{\partial c_{do} (V,k,V_B,T)}{\partial V} \big|_{V=V_B}]$ as $V_B \to 0$ and $V_B \to \infty$ are both bounded by $C (e^{-rT}+e^{-\nu T})$ for some constant $C>0$. The continuity then gives us the result. If $k=0$, we already know that $\tfrac{\partial}{\partial V_B}[\tfrac{\partial c_{do} (V,k,V_B,T)}{\partial V} \big|_{V=V_B}] = 0$, which trivially implies the claim. Thus, we assume $k>0$. For the limit when $V_B \to 0$, we only need to consider the case $V_B \leq k$. Each term can be individually analyzed using Lemma \ref{lemma: varphi bounded}, 
		and Lemma \ref{lemma: phi bounded} (with $\ln (\tfrac{V_B}{k})=\ln(V_B)-\ln(k)$) to ensure the boundedness of $\tfrac{\partial}{\partial V_B}[\tfrac{\partial c_{do} (V,k,V_B,T)}{\partial V} \big|_{V=V_B}]$ as $V_B \to 0$. For the limit when $V_B \to \infty$, we only need to consider the case $V_B\geq k$. In this case, we have $\lim_{V_B \to \infty} \tfrac{\partial}{\partial V_B}[\tfrac{\partial c_{do} (V,k,V_B,T)}{\partial V} \big|_{V=V_B}] = 0$. Combining these results with the continuity of $\tfrac{\partial}{\partial V_B}[\tfrac{\partial c_{do} (V,k,V_B,T)}{\partial V} \big|_{V=V_B}]$ in $V_B$, we obtain the upper bound. The non-negativity of $\tfrac{\partial}{\partial V_B}[\tfrac{\partial c_{do} (V,k,V_B,T)}{\partial V} \big|_{V=V_B}]$ is a straight application of $x \Phi(x)+\varphi(x) = \int_{-\infty}^x \Phi(t) \diff t \geq 0$, $\ln(\min\{1,\tfrac{V_B}{k}\}) \leq 0$, and $\lambda_2 \sigma \sqrt{t} = d_2(1,t)$.
	\end{myproof}
	
	\begin{lemma} \label{lemma: dcdo dalpha}
		There exists an $\hat{\alpha}>0$ and a constant $C_{\hat{\alpha}}>0$ such that $|V_B'(\alpha)| \leq C_{\hat{\alpha}}$ and $|\tfrac{\partial c_{do} (V,k,V_B(\alpha),T)}{\partial \alpha}| \leq C_{\hat{\alpha}} (e^{-\nu T}+e^{-rT})$ for all $\alpha \in [0,\hat{\alpha}]$.
	\end{lemma}
	
	\begin{myproof}
		We begin by taking the derivative of $V_B(\alpha)$. For this, we use the expression from equation \eqref{eq: formula for vb}, where we define the right-hand side as $R(\alpha,V_B)$. According to the implicit function theorem (if applicable), we have $V_B'(\alpha) = - \tfrac{R_{\alpha}(\alpha,V_B(\alpha))}{R_{V_B}(\alpha,V_B(\alpha))}$, where $R_x$ denotes the partial derivative with respect to $x$. Next, we evaluate the partial derivatives involved. With $\tfrac{\partial  c_{do} (V,k,V_B,t) }{\partial V} \big|_{V=V_B}$ as in \eqref{eq: dv cdo v=vb} and $\tfrac{\partial}{\partial V_B}[\tfrac{\partial  c_{do} (V,k,V_B,t) }{\partial V} \big|_{V=V_B}]$ as in \eqref{eq: dv2 cdo<} resp. \eqref{eq: dv2 cdo>}, we obtain (where we suppress the dependency of $V_B$ on $\alpha$):
		\begin{align}
			R_\alpha(\alpha,V_B) =&\, \tau_2 \int_0^\infty \tfrac{\partial  c_{do} (V,k,V_B,t) }{\partial V} \big|_{V=V_B} \diff t - \displaystyle\int_0^T \tfrac{\partial c_{do} (V,k,V_B,t)}{\partial V} \big|_{V=V_B} \diff t, \label{eq: ralpha}\\
			R_{V_B}(\alpha,V_B) =&\, \tfrac{1}{V_B^2} \Big( \tfrac{2(P-\frac{G}{r})A_1}{rT} + 2 \tfrac{G}{r} A_2 - \tau_1 \frac{G}{r}(\lambda_2+\lambda_3) \Big) \notag \\
			&+ \tau_2 \alpha \int_0^\infty \tfrac{\partial}{\partial V_B}[\tfrac{\partial  c_{do} (V,k,V_B,t) }{\partial V} \big|_{V=V_B}] \diff t - \alpha \displaystyle\int_0^T \tfrac{\partial}{\partial V_B}[\tfrac{\partial c_{do} (V,k,V_B,t)}{\partial V} \big|_{V=V_B}] \diff t, \label{eq: rvb}
		\end{align}
		where we are allowed to interchange the integral and the derivative due to Lemma \ref{lemma: dv2 cdo bounded l1}. Next, we note that $V_B (0) > 0$ by \eqref{eq: vb alpha=0 main text} (since $\alpha=0$ corresponds to no participation). Hence, using Lemma \ref{lemma: long term >0}, we have that $R_{V_B} (0,V_B(0)) > 0$. (This property also ensures that we can apply the implicit function theorem.) 
		Consequently, the implicit function theorem guarantees the existence of an $\hat{\alpha} >0$ such that $V_B(\alpha)$ is continuously differentiable in $\alpha$ for $\alpha \in [0,\hat{\alpha}]$. Thus, we have the continuity of $R_{V_B}$ in $\alpha$. As a result, possibly after reducing $\hat{\alpha}>0$, it follows that $R_{V_B}$ is bounded from below by an $\varepsilon>0$ for all $\alpha \in [0,\hat{\alpha}]$. In particular, there exists a $C_{\hat{\alpha}}$ such that $|V_B'(\alpha)| \leq C_{\hat{\alpha}}$ for all $\alpha \in [0,\hat{\alpha}]$. Moreover, we conclude from equations \eqref{eq: cd0 vb<k formula} and \eqref{eq: cd0 vb>k formula} analogously to equations \eqref{eq: derivative cdo vb<k} and \eqref{eq: derivative cdo vb>k} (with the dependence of $V_B$ on $\alpha$ suppressed) that:
		\begin{align}
			\tfrac{\partial}{\partial \alpha} c_{do}^{V_B \leq k} (V,k,V_B,T) =&\, - e^{-\nu T}2\lambda_1 (\textstyle\frac{V_B}{V})^{2\lambda_1-1} V_B'(\alpha) \Phi (d_1 (\textstyle\frac{V_B^2}{V k},T)) \notag \\
			&- V e^{-\nu T} (\textstyle\frac{V_B}{V})^{2\lambda_1} \varphi (d_1 (\textstyle\frac{V_B^2}{V k},T)) \tfrac{2}{\sigma \sqrt{T} V_B} V_B'(\alpha) \notag \\
			&+\tfrac{k}{V}e^{-rT} (2\lambda_1-2) (\textstyle\frac{V_B}{V})^{2\lambda_1-3} V_B'(\alpha) \Phi (d_2 (\textstyle\frac{V_B^2}{V k},T)) \notag \\
			&+ke^{-rT} (\textstyle\frac{V_B}{V})^{2\lambda_1-2} \varphi (d_2 (\textstyle\frac{V_B^2}{V k},T)) \tfrac{2}{\sigma \sqrt{T} V_B} V_B'(\alpha), \label{eq: dalpha cdo<} \\
			\tfrac{\partial}{\partial \alpha} c_{do}^{V_B \geq k} (V,k,V_B,T) =&\, V e^{-\nu T} \varphi (d_1 (\textstyle\frac{V}{V_B},T)) \tfrac{-1}{\sigma \sqrt{T} V_B} V_B'(\alpha) - k e^{-rT} \varphi(d_2 (\textstyle\frac{V}{V_B},T)) \tfrac{-1}{\sigma \sqrt{T} V_B} V_B'(\alpha) \notag \\
			&- e^{-\nu T}2\lambda_1(\textstyle\frac{V_B}{V})^{2\lambda_1-1} V_B'(\alpha) \Phi (d_1 (\textstyle\frac{V_B}{V},T)) \notag \\
			&- V e^{-\nu T}(\textstyle\frac{V_B}{V})^{2\lambda_1} \varphi (d_1 (\textstyle\frac{V_B}{V},T)) \tfrac{1}{\sigma \sqrt{T} V_B} V_B'(\alpha) \notag\\
			&+\tfrac{k}{V}e^{-rT} (2\lambda_1-2) (\textstyle\frac{V_B}{V})^{2\lambda_1-3} V_B'(\alpha) \Phi (d_2 (\textstyle\frac{V_B}{V},T)) \notag\\
			&+ke^{-rT} (\textstyle\frac{V_B}{V})^{2\lambda_1-2} \varphi (d_2 (\textstyle\frac{V_B}{V},T)) \tfrac{1}{\sigma \sqrt{T} V_B} V_B'(\alpha). \label{eq: dalpha cdo>}
		\end{align}
		Next, we consider each term in the expression individually. All of these terms are bounded by $C_{\hat{\alpha}} (e^{-\nu T}+e^{-r T})$, since $V_B(\alpha)$ is bounded (by $V$ by definition) and bounded away from zero for all $\alpha \in [0,\hat{\alpha}]$ (using Lemma \ref{lem: vb neq 0} and Proposition \ref{prop: vb monoton alpha g}), with an appropriately chosen constant $C_{\hat{\alpha}}$. Therefore, the claim follows.
	\end{myproof}

	\begin{lemma} \label{lemma: dcdo dg}
		Assume 
		\begin{align*}
			\tfrac{2PA_1}{V_B(0)^2rT} &+ \tau_2 \alpha \displaystyle\int_0^\infty \tfrac{\partial}{\partial V_B}[\tfrac{\partial  c_{do} (V,k,V_B,t) }{\partial V} \big|_{V=V_B}] \big|_{V_B=V_B(0)}\diff t \\
			&- \alpha \displaystyle\int_0^T \tfrac{\partial}{\partial V_B}[\tfrac{\partial c_{do} (V,k,V_B,t)}{\partial V} \big|_{V=V_B}] \big|_{V_B=V_B(0)} \diff t \neq 0,
		\end{align*}
		where $\tfrac{\partial}{\partial V_B}[\tfrac{\partial  c_{do} (V,k,V_B,t) }{\partial V} \big|_{V=V_B}]$ is given in \eqref{eq: dv2 cdo<} resp. \eqref{eq: dv2 cdo>}.
		Then, there exists a $\hat{g}>0$ and a constant $C_{\hat{g}}>0$ such that $|V_B'(g)| \leq C_{\hat{g}}$ and $|\tfrac{\partial c_{do} (V,k,V_B(g),T)}{\partial g}| \leq C_{\hat{g}} (e^{-\nu T}+e^{-rT})$ for all $g \in [0,\hat{g}]$.
	\end{lemma}
	
	\begin{myproof}
		This proof follows from an approach similar to the proof of Lemma \ref{lemma: dcdo dalpha}. We begin by differentiating $V_B(g)$. Again, we use \eqref{eq: formula for vb} and define the right-hand side as $R(g,V_B)$. By applying the implicit function theorem (if applicable), we obtain $V_B'(g) = - \tfrac{R_{g}(g,V_B(g))}{R_{V_B}(g,V_B(g))}$, where $R_x$ denotes the partial derivative with respect to $x$. 
		With $\tfrac{\partial}{\partial V_B}[\tfrac{\partial  c_{do} (V,k,V_B,t) }{\partial V} \big|_{V=V_B}]$ as in \eqref{eq: dv2 cdo<} resp. \eqref{eq: dv2 cdo>}, we obtain (where we suppress the dependency of $V_B$ on $g$):
		\begin{align}
			R_g(g,V_B) =&\, -\tfrac{T}{V_B r} \Big( -\tfrac{2A_1}{rT} + 2A_2 - \tau_1(\lambda_2+\lambda_3) \Big), \label{eq: rg}\\
			R_{V_B}(g,V_B) =&\, \tfrac{1}{V_B^2} \Big( \tfrac{2(P-\frac{G}{r})A_1}{rT} + 2 \tfrac{G}{r} A_2 - \tau_1 \frac{G}{r}(\lambda_2+\lambda_3) \Big) \notag \\
			&+ \tau_2 \alpha \int_0^\infty \tfrac{\partial}{\partial V_B}[\tfrac{\partial  c_{do} (V,k,V_B,t) }{\partial V} \big|_{V=V_B}] \diff t - \alpha \displaystyle\int_0^T \tfrac{\partial}{\partial V_B}[\tfrac{\partial c_{do} (V,k,V_B,t)}{\partial V} \big|_{V=V_B}] \diff t, \label{eq: rvb2}
		\end{align}
		where $A_1$ and $A_2$ are defined as in \eqref{eq: a1} and \eqref{eq: a2}. By Lemma \ref{lem: vb neq 0}
		, we know that $V_B(0) \neq 0$. 
		Therefore, by assumption, we have that $R_{V_B}(0,V_B(0)) \neq 0$, which can be verified by substituting $g=0$ into $R_{V_B}(g,V_B(g))$ and comparing it with the assumption of the lemma. (This property also ensures us that we can apply the implicit function theorem.) 
		By the implicit function theorem, we conclude that there exists a $\hat{g} >0$ such that $V_B(g)$ is continuously differentiable in $g$ on $[0,\hat{g}]$. This guarantees the continuity of $R_{V_B}(g,V_B(g))$ in $g$. Furthermore, possibly after reducing $\hat{g}>0$, it follows that $|R_{V_B}(g,V_B(g))|$ is bounded from below for all $g \in [0,\hat{g}]$. In particular, there exists a constant $C_{\hat{g}}$ such that $|V_B'(g)| \leq C_{\hat{g}}$ for all $g \in [0,\hat{g}]$. Additionally, similar to Lemma \ref{lemma: dv2 cdo bounded l1} (where the dependence of $V_B$ on $g$ is omitted), we obtain:
		\begin{align}
			\tfrac{\partial}{\partial g} c_{do}^{V_B \leq k} (V,k,V_B,T) =&\, - e^{-\nu T}2\lambda_1 (\textstyle\frac{V_B}{V})^{2\lambda_1-1} V_B'(g) \Phi (d_1 (\textstyle\frac{V_B^2}{V k},T)) \notag \\
			&- V e^{-\nu T} (\textstyle\frac{V_B}{V})^{2\lambda_1} \varphi (d_1 (\textstyle\frac{V_B^2}{V k},T)) \tfrac{2}{\sigma \sqrt{T} V_B} V_B'(g) \notag \\
			&+\tfrac{k}{V}e^{-rT} (2\lambda_1-2) (\textstyle\frac{V_B}{V})^{2\lambda_1-3} V_B'(g) \Phi (d_2 (\textstyle\frac{V_B^2}{V k},T)) \notag \\
			&+ke^{-rT} (\textstyle\frac{V_B}{V})^{2\lambda_1-2} \varphi (d_2 (\textstyle\frac{V_B^2}{V k},T)) \tfrac{2}{\sigma \sqrt{T} V_B} V_B'(g), \label{eq: dg cdo<} \\
			\tfrac{\partial}{\partial g} c_{do}^{V_B \geq k} (V,k,V_B,T) =&\, V e^{-\nu T} \varphi (d_1 (\textstyle\frac{V}{V_B},T)) \tfrac{-1}{\sigma \sqrt{T} V_B} V_B'(g) - k e^{-rT} \varphi(d_2 (\textstyle\frac{V}{V_B},T)) \tfrac{-1}{\sigma \sqrt{T} V_B} V_B'(g) \notag \\
			&- e^{-\nu T}2\lambda_1(\textstyle\frac{V_B}{V})^{2\lambda_1-1} V_B'(g) \Phi (d_1 (\textstyle\frac{V_B}{V},T)) \notag \\
			&- V e^{-\nu T}(\textstyle\frac{V_B}{V})^{2\lambda_1} \varphi (d_1 (\textstyle\frac{V_B}{V},T)) \tfrac{1}{\sigma \sqrt{T} V_B} V_B'(g) \notag\\
			&+\tfrac{k}{V}e^{-rT} (2\lambda_1-2) (\textstyle\frac{V_B}{V})^{2\lambda_1-3} V_B'(g) \Phi (d_2 (\textstyle\frac{V_B}{V},T)) \notag\\
			&+ke^{-rT} (\textstyle\frac{V_B}{V})^{2\lambda_1-2} \varphi (d_2 (\textstyle\frac{V_B}{V},T)) \tfrac{1}{\sigma \sqrt{T} V_B} V_B'(g). \label{eq: dg cdo>}
		\end{align}
		We now consider each term in the expression individually. All of these terms are bounded by $C_{\hat{g}} (e^{-\nu T}+e^{-r T})$, since $V_B$ is bounded (by $V$) and bounded away from zero for all $g \in [0,\hat{g}]$ (using Lemma \ref{lem: vb neq 0} and Proposition \ref{prop: vb monoton alpha g}), with an appropriately chosen constant $C_{\hat{g}}$. Therefore, the claim follows.
	\end{myproof}
	
	\begin{lemma}\label{lemma: new}
		The partial derivatives $\tfrac{\partial}{\partial V_B} c_{do} (V,k,V_B,T)$, $\tfrac{\partial}{\partial \alpha} c_{do} (V,k,V_B,T)$, and $\tfrac{\partial}{\partial g} c_{do} (V,k,V_B,T)$ are continuous with respect to the variable of differentiation.
	\end{lemma}
	
	\begin{proof}
		According to the proofs of Lemmas \ref{lemma: dcdo dalpha} and \ref{lemma: dcdo dg}, the functions $V_B'(\alpha)$ and $V_B'(g)$ are continuous. Using equations \eqref{eq: dalpha cdo<}, \eqref{eq: dalpha cdo>}, \eqref{eq: dg cdo<}, and \eqref{eq: dg cdo>}, we observe that the continuity of the respective partial derivatives follows provided continuity holds at the point $V_B=k$. 
		To verify this, we compare the corresponding expressions at $V_B=k$ and note that continuity at $V_B = k$ is ensured if the following two conditions are satisfied: (i) $(\textstyle\frac{V_B}{V})^{2\lambda_1} \varphi (d_1 (\textstyle\frac{V_B}{V},T)) = \varphi (d_1 (\textstyle\frac{V}{V_B},T))$ and (ii) $(\textstyle\frac{V_B}{V})^{2\lambda_1-2} \varphi (d_2 (\textstyle\frac{V_B^2}{V k},T))= \varphi(d_2 (\textstyle\frac{V}{V_B},T))$. 
		
		We begin with the proof of (i). From \eqref{eq: def d12 lambda1}, it follows that $d_1 (x,T) = \tfrac{\ln x}{\sigma \sqrt{t}} + \sigma\lambda_1\sqrt{t}$. Then, using this expression, we compute:
		\begin{align*}
			(\textstyle\frac{V_B}{V})^{2\lambda_1} \varphi (d_1 (\textstyle\frac{V_B}{V},T)) &= e^{2\lambda_1 (\ln V_B - \ln V)} \cdot \tfrac{1}{\sqrt{2\pi}} e^{-\frac{1}{2} (\frac{(\ln V_B - \ln V)^2}{\sigma^2 t} + 2\lambda_1(\ln V_B - \ln V) +\sigma^2 \lambda_1^2 t)} \\
			&= \tfrac{1}{\sqrt{2\pi}} e^{-\frac{1}{2} (\frac{(\ln V_B - \ln V)^2}{\sigma^2 t} - 2\lambda_1(\ln V_B - \ln V) +\sigma^2 \lambda_1^2 t)} \\
			&= \tfrac{1}{\sqrt{2\pi}} e^{-\frac{1}{2} (\frac{(\ln V - \ln V_B)^2}{\sigma^2 t} + 2\lambda_1 \sigma \sqrt{t}\frac{(\ln V - \ln V_B)}{\sigma \sqrt{t}} +\sigma^2 \lambda_1^2 t)} = \varphi (d_1 (\textstyle\frac{V}{V_B},T)).
		\end{align*}
		
		The proof of (ii) follows analogously, using the facts that $\lambda_1-1=\lambda_2$ (see \eqref{eq: def lambda23}) and $d_2 (x,T) = \tfrac{\ln x}{\sigma \sqrt{t}} + \sigma\lambda_2\sqrt{t}$.
	\end{proof}
	
	\section{Proof of the main results} \label{proofs}
	In this section, we present the proofs for all theorems and propositions discussed in the main body of the paper.
	
	\begin{myproof}[Proof of Theorem \ref{th: VB determination}]
		By Lemma \ref{lem: vb neq 0}, the optimal bankruptcy-triggering value is bounded away from $0$. 
		
		First, we compute the left-hand side of \eqref{eq: smooth pasting condition} explicitly by substituting the closed-form expressions for the liability value and the firm value given in \eqref{eq: liability value L} and \eqref{eq: firm value v}, respectively:
		\begin{align*}
			\tfrac{\partial E(V;V_B,T)}{\partial V} \big|_{V=V_B} =&\, \tfrac{\partial v(V;V_B)}{\partial V} \big|_{V=V_B} - \tfrac{\partial L(V;V_B,T)}{\partial V} \big|_{V=V_B} \\
			=&\, \tfrac{\partial V}{\partial V} \big|_{V=V_B} + \tfrac{\partial \tau_1 \frac{G}{r} (1-(\frac{V_B}{V})^{\lambda_2+\lambda_3})}{\partial V} \big|_{V=V_B} + \tfrac{\partial \tau_2 \alpha \int_0^\infty c_{do} (V,k,V_B,t) \diff t}{\partial V} \big|_{V=V_B} \\
			&- \tfrac{\partial \rho V_B (\frac{V_B}{V})^{\lambda_2+\lambda_3}}{\partial V} \big|_{V=V_B}-\tfrac{\partial \frac{G}{r}}{\partial V} \big|_{V=V_B}-\tfrac{\partial ( P - \frac{G}{r} ) ( \frac{1-e^{-rT}}{rT} - I_1^V(T))}{\partial V} \big|_{V=V_B} \\
			&-\tfrac{\partial ( (1-\rho)V_B - \frac{G}{r} ) I_2^V(T)}{\partial V} \big|_{V=V_B} -\tfrac{\partial \alpha \int_0^T c_{do} (V,k,V_B,t) \diff t}{\partial V} \big|_{V=V_B} \\
			=&\, 1 - \tau_1 \tfrac{G}{r} \tfrac{\partial (\frac{V_B}{V})^{\lambda_2+\lambda_3}}{\partial V} \big|_{V=V_B} + \tau_2 \alpha \int_0^\infty \tfrac{\partial  c_{do} (V,k,V_B,t) }{\partial V} \big|_{V=V_B} \diff t \\
			&- \rho V_B  \tfrac{\partial (\frac{V_B}{V})^{\lambda_2+\lambda_3}}{\partial V} \big|_{V=V_B}-0+( P - \tfrac{G}{r} )\tfrac{\partial I_1^V(T)}{\partial V} \big|_{V=V_B} \\
			&-( (1-\rho)V_B - \tfrac{G}{r} )\tfrac{\partial I_2^V(T)}{\partial V} \big|_{V=V_B} - \alpha \int_0^T \tfrac{\partial c_{do} (V,k,V_B,t)}{\partial V} \big|_{V=V_B} \diff t,
		\end{align*}
		where we used Lemma \ref{lemma: cdo bounded L1} to interchange the integral with the derivative by the Leibniz integral rule (in the measure theoretic version). 
		
		Now, in a first step, we demonstrate the existence of a solution $V_B$ by setting the previous equation equal to zero. The definitions of $A_1$ and $A_2$ in \eqref{eq: a1} and \eqref{eq: a2}, along with Lemma \ref{lemma: i1 i2 derivative}, imply that $\tfrac{\partial I_1^V(T)}{\partial V} \big|_{V=V_B} = - \tfrac{2A_1}{rTV_B}$ and $\tfrac{\partial I_2^V(T)}{\partial V} \big|_{V=V_B} = - \tfrac{2A_2}{V_B}$. In the initial step, we will disregard the participating part by setting $\alpha=0$. Then, using Lemma \ref{lemma: i1 i2 derivative} and noting that $\tfrac{\partial (\frac{V_B}{V})^{\lambda_2+\lambda_3}}{\partial V} \big|_{V=V_B} = (-\lambda_2-\lambda_3)V_B^{\lambda_2+\lambda_3} V^{-\lambda_2-\lambda_3-1} \big|_{V=V_B} = -\tfrac{(\lambda_2+\lambda_3)}{V_B}$, we have:
		\begin{align}
			0 &= 1 + \tfrac{\tau_1 \frac{G}{r}(\lambda_2+\lambda_3)}{V_B} + \tfrac{\rho V_B(\lambda_2+\lambda_3)}{V_B}-\tfrac{2(P-\frac{G}{r})A_1}{rTV_B} + \tfrac{2((1-\rho)V_B-\frac{G}{r})A_2}{V_B} \notag \\
			&= 1 + \rho(\lambda_2+\lambda_3)+2(1-\rho)A_2 - \tfrac{1}{V_B} \Big( \tfrac{2(P-\frac{G}{r})A_1}{rT} + 2 \tfrac{G}{r} A_2 - \tau_1 \frac{G}{r}(\lambda_2+\lambda_3) \Big) =: h_1(V_B). \label{eq: def h1}
		\end{align}
		Solving this equation for $V_B$ yields:
		\begin{align} \label{eq: vb alpha=0}
			V_B^* = \frac{\tfrac{2(P-\frac{G}{r})A_1}{rT} + 2 \tfrac{G}{r} A_2 - \tau_1 \frac{G}{r}(\lambda_2+\lambda_3)}{1 + \rho(\lambda_2+\lambda_3)+2(1-\rho)A_2},
		\end{align}
		where Lemma \ref{lemma: long term >0} guarantees that the nominator and the denominator are positive, ensuring that $V_B^*$ is positive and well-defined. 
		
		Specifically, we find that $V_B^*>0$ solves $h_1(V_B)=0$. Since $h$ is increasing in $V_B$, we have $h_1(V) < 0$ for all $V<V_B^*$ and $h(V) > 0$ for all $V>V_B^*$. In particular, by Lemma \ref{lemma: term bigger 0 for vb},we obtain:
		\begin{align}
			\lim_{V \to 0} h_1(V) &= - \infty, \label{eq: h1 v to 0} \\
			\lim_{V \to \infty} h_1(V) &= 1 + \rho(\lambda_2+\lambda_3)+2(1-\rho)A_2 > 0. \notag
		\end{align}	
		Next, we incorporate the term with the participation component. Consequently, $\tfrac{\partial E(V;V_B,T)}{\partial V} \big|_{V=V_B}=0$ is equivalent to
		\begin{align} \label{eq: def of h2}
			0 &= h_1(V_B) + \tau_2 \alpha \int_0^\infty \tfrac{\partial  c_{do} (V,k,V_B,t) }{\partial V} \big|_{V=V_B} \diff t - \alpha \displaystyle\int_0^T \tfrac{\partial c_{do} (V,k,V_B,t)}{\partial V} \big|_{V=V_B} \diff t =: h_2(V_B).
		\end{align}
		Note that we plug in $V=V_B$ into both functions $h_1$ and $h_2$. Since the Dominated Convergence Theorem allows us to interchange the integral and the limit (with its prerequisite demonstrated in Lemma \ref{lemma: cdo bounded L1}), we can apply Lemma \ref{lemma: cdo limit of derivative to 0 and infty}. Furthermore, by utilizing the fact that $d_1(1,t)) = (\tfrac{r-\nu}{\sigma}+\tfrac{\sigma}{2})\sqrt{t} = \lambda_1 \sigma \sqrt{t}$, and incorporating the equalities from \eqref{eq: a3 equality integral} and \eqref{eq: a4 equality integral}, we can conclude that:
		\begin{align}
			\lim_{V_B \to 0} h_2(V_B) =&\, - \infty, \notag \\
			\lim_{V_B \to \infty} h_2(V_B) =&\, 1 + \rho(\lambda_2+\lambda_3)+2(1-\rho)A_2 + \tau_2 \alpha \int_0^\infty \Big(e^{-\nu t}(2\lambda_1 \Phi(\lambda_1 \sigma \sqrt{t})+ \tfrac{2\varphi(\lambda_1 \sigma \sqrt{t})}{\sigma \sqrt{t}} )\Big) \diff t \notag \\
			&- \alpha \int_0^T \Big(e^{-\nu t}(2\lambda_1 \Phi(\lambda_1 \sigma \sqrt{t})+ \tfrac{2\varphi(\lambda_1 \sigma \sqrt{t})}{\sigma \sqrt{t}} )\Big) \diff t \notag \\
			=&\, 1 + \rho(\lambda_2+\lambda_3)+2(1-\rho)A_2 + \tau_2 \alpha A_3 - \alpha A_4. \label{eq: limit h2 inf}
		\end{align}

		Now, formula \eqref{eq: formula for vb} follows directly from \eqref{eq: def of h2}, and we obtain formula \eqref{eq: dv cdo v=vb} from \eqref{eq: dv cdo< v=vb} and \eqref{eq: dv cdo> v=vb}, 
		where we used that $\lambda_1-1=\lambda_2$. 
		%
		%
		If the solution of \eqref{eq: formula for vb} exceeds $V_0$, the insurance company declares bankruptcy immediately. Therefore, we can equivalently set $V_B = V_0$ in this case, without affecting the timing of the bankruptcy declaration, while ensuring that $V_B$ represents the asset value (before subtracting the bankruptcy costs) at the time of bankruptcy.
		
		It still remains to discuss which solution of \eqref{eq: formula for vb} is the best one and the existence. First, we note that the optimal bankruptcy-triggering value $V_B$ has to be a minimum of $V \to E(V,V_B)$ as a maximum violates the condition $E(V) \geq 0$ for all $V \geq V_B$ since $E(V_B,V_B) = 0$ (see also the discussion at the beginning of Section \ref{subsection: derivation of vb}). Moreover, remember that for the optimal-bankruptcy triggering value $V_B$, the function $V \to E(V,V_B)$ cannot be negative. 
		
		Now, if $\lim_{V \to \infty} h_2(V) >0$ (resp. $\downarrow 0$), then Lemma \ref{lemma: uniqueness preparation lemma} implies the uniqueness of the zero root of \eqref{eq: formula for vb}, which trivially corresponds to a minimum of $V \to E(V,V_B)$ (as $\lim_{V \to 0} h_2(V) = - \infty$). If $\lim_{V \to \infty} h_2(V) < 0$ (resp. $\uparrow 0$), then Lemma \ref{lemma: uniqueness preparation lemma} implies that there exist up to two zero roots of \eqref{eq: formula for vb}. If $h_2$ has no zero root, then $E'_V (V,V) < 0$ for all $V>0$ as $\lim_{V \to 0} h_2(V) = - \infty$. Hence, every choice of a bankruptcy-triggering value would lead to a negative equity at a value larger than this point which contradicts our assumption of non-negative equity and immediate bankruptcy is optimal. Thus, we can set $V_B=V_0$, i.e., immediate bankruptcy\footnote{As \eqref{eq: formula for vb} admits no solution in this case, the smooth-pasting condition cannot be applied. 
			The absence of a solution initially corresponds to $V_B = \infty$, which implies immediate bankruptcy regardless of the initial value $V_0$. 
			Since choosing $V_B = V_0$ also results in immediate bankruptcy, we may set $V_B = V_0$ without affecting the outcome.}. If $h_2$ has one zero root $V_B^*$, then we have again the unique solution. If $h_2$ has two zero roots, by \eqref{eq: limit h2 inf} only the first zero root corresponds to a minimum of the function $V \to E(V,V_B)$. In the remaining case that $\lim_{V \to \infty} h_2(V) =0$ and that $h_2$ is constant from a certain $V$ on, Lemma \ref{lemma: uniqueness preparation lemma} implies that $h_2(V) =0$ for all $V \geq k$. Lemma \ref{lemma: uniqueness preparation lemma} entails that there is at most one zero root smaller than $k$ which corresponds analogously to before to a minimum of the function $V \to E(V,V_B)$. The zero roots larger than $k$ can be ignored, as we may take the minimum of the optimal solution and $V_0$ (since $k \geq V_0$ by assumption). Summarizing, we have shown the statement that the minimum of the smallest solution of the smooth pasting condition and $V_0$ is the optimal bankruptcy-triggering value. In particular, note that at most one solution of \eqref{eq: formula for vb} is a possible bankruptcy-triggering value.
	\end{myproof}
	
	\begin{myproof}[Proof of Proposition \ref{prop: vb monoton tau T}]
		We first note from the previous proof and Lemma \ref{lemma: uniqueness preparation lemma} that even when parameters change, no additional relevant zero roots can occur or vanish as long as a solution exists. 
		If $\tau_2$ increases, we observe that the right-hand side of \eqref{eq: formula for vb} also increases. From the proof of Theorem \ref{th: VB determination}, we know that the smallest solution of \eqref{eq: formula for vb} corresponds to a sign transition from ``$-$'' to ``$+$'' (or to a local minimum, with the function being positive to the right-hand side of the zero root). Therefore, as the graph shifts upwards, the zero root decreases showing that $V_B$ is increasing in $\tau_2$ if a zero root exists. If no zero root exists, it holds that $V_B=V_0$. Now, if $\tau_2$ decreases, no zero root can occur. If $\tau_2$ increases and equation \eqref{eq: formula for vb} gets a zero root, this leads to the case that $V_B \leq V_0$ (as $V_B$ is the minimum of this zero root and $V_0$).
		
		We know that the positive value $\tfrac{2(P-\frac{G}{r})A_1}{rT} + 2 \tfrac{G}{r} A_2 - \tau_1 \frac{G}{r}(\lambda_2+\lambda_3)$ decreases if $\tau_1$ increases. (Note that positivity is ensured by Lemma \ref{lemma: long term >0}.) Therefore, the zero root $V_B$ must also decrease in order to maintain the equality (if existent, otherwise similar to above), ensuring that the right-hand side of equation \eqref{eq: formula for vb} remains zero. 
		
		The argument for the contract maturity $T$ follows analogously, but in the opposite direction in both cases. Indeed, an increasing $T$ leads to lower values on the right-hand side of equation \eqref{eq: formula for vb} in the participation component, while simultaneously increasing $\tfrac{2(P-\frac{G}{r})A_1}{rT} + 2 \tfrac{G}{r} A_2 - \tau_1 \frac{G}{r}(\lambda_2+\lambda_3)$ under the assumption that $P - \tfrac{G}{r} \leq 0$. Therefore, as $T$ increases, the bankruptcy-triggering value increases as well. 
	\end{myproof}
	
	\begin{myproof}[Proof of Proposition \ref{prop: vb monoton alpha g}]
		
		To show that $V_B$ is monotonically increasing in $\alpha$, let $0 \leq \alpha_1 < \alpha_2$, and we aim to prove that $V_B(\alpha_1) < V_B(\alpha_2)$. By Theorem \ref{th: VB determination}, we know that:
		\begin{align}
			0 =&\, 1 + \rho(\lambda_2+\lambda_3)+2(1-\rho)A_2 - \tfrac{1}{V_B(\alpha_1)} \Big( \tfrac{2(P-\frac{G}{r})A_1}{rT} + 2 \tfrac{G}{r} A_2 - \tau_1 \frac{G}{r}(\lambda_2+\lambda_3) \Big) \notag \\
			&+ \tau_2 \alpha_1 \int_0^\infty \tfrac{\partial  c_{do} (V,k,V_B(\alpha_1),t) }{\partial V} \big|_{V=V_B(\alpha_1)} \diff t - \alpha_1 \displaystyle\int_0^T \tfrac{\partial c_{do} (V,k,V_B(\alpha_1),t)}{\partial V} \big|_{V=V_B(\alpha_1)} \diff t. \label{eq: solution formula vb(alpha1)}
		\end{align}
		Using Lemma \ref{lemma: ass reformulation derivative}, we get, in particular, that:
		\begin{align*}
			0 >&\, 1 + \rho(\lambda_2+\lambda_3)+2(1-\rho)A_2 - \tfrac{1}{V_B(\alpha_1)} \Big( \tfrac{2(P-\frac{G}{r})A_1}{rT} + 2 \tfrac{G}{r} A_2 - \tau_1 \frac{G}{r}(\lambda_2+\lambda_3) \Big) \notag \\
			&+ \tau_2 \alpha_2 \int_0^\infty \tfrac{\partial  c_{do} (V,k,V_B(\alpha_1),t) }{\partial V} \big|_{V=V_B(\alpha_1)} \diff t - \alpha_2 \displaystyle\int_0^T \tfrac{\partial c_{do} (V,k,V_B(\alpha_1),t)}{\partial V} \big|_{V=V_B(\alpha_1)} \diff t.
		\end{align*}
		Loosely speaking, Lemma \ref{lemma: ass reformulation derivative} implies that, if $V_B$ remains constant, every point of the right-hand side's graph decreases as $\alpha$ increases. As in the previous proof, we know that there exists at most one relevant zero root. Hence, to maintain equality, it must hold that $V_B(\alpha_1) < V_B(\alpha_2)$. This establishes the claim, since $V_B$ is defined as the minimum of this zero root (if it exists) and $V_0$.
		
		Second, we establish left-continuity in $\alpha$. A discontinuity can occur only if, as $\alpha$ increases, equation \eqref{eq: formula for vb} ceases to admit a zero root. Otherwise, there is at most one relevant zero root and the expression is continuous in all parameters. In this case, the right-hand side of \eqref{eq: formula for vb} is negative for all admissible values of $V_B$. By definition, this implies $V_B = V_0$, which is the maximal possible value of $V_B$. By monotonicity, such a jump can occur at most once. Since $V_B$ is non-decreasing in $\alpha$, left-continuity follows.
		
		From a graphical perspective, as before, every point of the graph of the right-hand side decreases as $\alpha$ increases. The critical situation arises when the local maximum becomes negative. At this critical value, a solution still exists; hence, together with the monotonicity of $V_B$ in $\alpha$, this yields left-continuity.

		Third, the existence of the right-limits follows directly from the monotonicity.
		
		The proof for the guaranteed payment proceeds similarly. We observe that, under the assumption on the value of the guaranteed payments, the term $\tfrac{2(P-\frac{G}{r})A_1}{rT} + 2 \tfrac{G}{r} A_2 - \tau_1 \frac{G}{r}(\lambda_2+\lambda_3)$ is increasing in $G$ (as analyzed in the proof of Theorem \ref{th: VB determination}), while the other terms depend on $G$ only through $V_B$. Therefore, an analogous reasoning leads to the same conclusion about the behavior of $V_B$ with respect to $G$.
	\end{myproof}
	
	
	\begin{myproof}[Proof of Proposition \ref{prop: alpha bar competitive market}]
		As in the proof of Theorem~\ref{th: VB determination}, we begin by noting that 
		the non-negativity of the mapping \(V \mapsto E(V)\) is implied by the condition 
		\(\lim_{V \to \infty} E(V;V_B) \ge 0\) for any fixed \(V_B\). 
		Fix \(V_B\). As \(V \to \infty\), the probability of bankruptcy converges to zero, 
		and therefore the limit of the equity value becomes independent of \(V_B\). 
		Using \eqref{equity}, \eqref{eq: firm value v}, and \eqref{eq: liability value L}, 
		we obtain:
		\begin{align*}
			\lim_{V \to \infty} E(V;V_B) = \lim_{V \to \infty} [ V + \tau_1 \tfrac{G}{r} + \tau_2 \alpha \int_0^\infty c (V,k,t) \diff t - \tfrac{G}{r} - ( P - \tfrac{G}{r} )  \tfrac{1-e^{-rT}}{rT} - \alpha \int_0^T c (V,k,t) \diff t ].
		\end{align*}
		Thus, if $\lim_{V \to \infty} [ V + \tau_2 \alpha \int_0^\infty c (V,k,t) \diff t - \alpha \int_0^T c (V,k,t) \diff t ] = \infty$, then indeed $\lim_{V \to \infty} E(V;V_B) \geq 0$. To verify this condition, we observe that for all \(s \in (0,\infty]\), 
		the dominated convergence theorem applies since \(r,\nu >0\) and it holds that
		\(\lim_{V \to \infty} \Phi(d_{1/2}(\tfrac{V}{k},t)) = 1\). Hence,
		\begin{align*}
			\lim_{V \to \infty} \int_0^s c (V,k,t) \diff t &=  \lim_{V \to \infty} \int_0^s  V e^{-\nu t} \Phi (d_1 (\tfrac{V}{k},t)) - k e^{-rt} \Phi(d_2 (\tfrac{V}{k},t)) \diff t \\
			&= \lim_{V \to \infty} [V \int_0^s e^{-\nu t} \diff t - k \int_0^s e^{-r t} \diff t] = \lim_{V \to \infty} [V \tfrac{1}{\nu} (1-e^{-\nu s}) - k \tfrac{1}{r} (1-e^{-rs})].
		\end{align*}
		Consequently, $\lim_{V \to \infty} [ V + \tau_2 \alpha \int_0^\infty c (V,k,t) \diff t - \alpha \int_0^T c (V,k,t) \diff t ] = \infty$ holds provided that $1+ \tau_2 \alpha \tfrac{1}{\nu} - \alpha \tfrac{1}{\nu} (1-e^{-\nu T}) > 0$, which is equivalent to $\alpha < \bar{\alpha}$, with $\bar{\alpha}$ given in \eqref{eq: baralpha}.
	\end{myproof}
	
	
	\begin{myproof}[Proof of Corollary \ref{cor: unique solution of vb determination}]
		We obtain from the proof of Theorem \ref{th: VB determination} by using Lemma \ref{lemma: uniqueness preparation lemma} that \eqref{eq: formula for vb} has a unique solution if $\lim_{V_B \to \infty} h_2(V_B)>0$, where $h_2$ is defined in \eqref{eq: def of h2}. Applying \eqref{eq: limit h2 inf} yields that if $1 + \rho(\lambda_2+\lambda_3)+2(1-\rho)A_2 + \tau_2 \alpha A_3 - \alpha A_4>0$, \eqref{eq: formula for vb} has a unique solution. This inequality is equivalent to $\alpha<\tilde{\alpha}$, with $\tilde{\alpha}$ given in \eqref{eq: tildealpha} implying the claim.
	\end{myproof}
	
	\begin{myproof}[Proof of Corollary \ref{cor: vb determination}]
		Let us assume that $V_B \geq k$, i.e., $\min\{\tfrac{V_B}{k},1\}=1$. Note that $d_1(1,T) = \lambda_1 \sigma \sqrt{T}$ and $d_2(1,T) = \lambda_2 \sigma \sqrt{T}$. Then, we get, using \eqref{eq: dv cdo v=vb} and Lemma \ref{lemma: explicit calculation integral} (once applied in the original version, and once applied with $\lambda_1$ replaced by $\lambda_2$, and $\nu$ replaced by $r$):
		\begin{align}
			\int_0^\infty \tfrac{\partial  c_{do} (V,k,V_B,t) }{\partial V} \big|_{V=V_B} \diff t =&\, \tfrac{\lambda_1}{\nu} + \tfrac{\lambda_1^2 \sigma}{\nu} \sqrt{\tfrac{1}{\lambda_1^2\sigma^2+2\nu}} + \tfrac{2}{\sigma} \sqrt{\tfrac{1}{\lambda_1^2\sigma^2+2\nu}} \notag \\
			&- \tfrac{k}{V_B} \Big( \tfrac{\lambda_2}{r} + \tfrac{\lambda_2^2 \sigma}{r} \sqrt{\tfrac{1}{\lambda_2^2\sigma^2+2r}} + \tfrac{2}{\sigma} \sqrt{\tfrac{1}{\lambda_2^2\sigma^2+2r}} \Big) \notag \\
			=&\, A_3 - \tfrac{k}{V_B} A_5, \label{eq: integral infty del cdo v=vb}\\
			\int_0^T \tfrac{\partial  c_{do} (V,k,V_B,t) }{\partial V} \big|_{V=V_B} \diff t =&\, \tfrac{\lambda_1}{\nu} - \tfrac{2\lambda_1 e^{-\nu T} \Phi(\lambda_1\sigma\sqrt{T})}{\nu} + \tfrac{\lambda_1^2 \sigma}{\nu} \sqrt{\tfrac{1}{\lambda_1^2\sigma^2+2\nu}} (2\Phi(\sqrt{\lambda_1^2 \sigma^2 + 2\nu}\sqrt{T})-1) \notag \\
			&+ \tfrac{2}{\sigma} \sqrt{\tfrac{1}{\lambda_1^2\sigma^2+2\nu}} (2\Phi(\sqrt{\lambda_1^2 \sigma^2 + 2\nu}\sqrt{T})-1) \notag \\
			&- \tfrac{k}{V_B} \Big( \tfrac{\lambda_2}{r} - \tfrac{2\lambda_2 e^{-r T} \Phi(\lambda_2\sigma\sqrt{T})}{r} + \tfrac{\lambda_2^2 \sigma}{r} \sqrt{\tfrac{1}{\lambda_2^2\sigma^2+2r}} (2\Phi(\sqrt{\lambda_2^2 \sigma^2 + 2r}\sqrt{T})-1) \notag \\
			&\hspace{30pt}+ \tfrac{2}{\sigma} \sqrt{\tfrac{1}{\lambda_2^2\sigma^2+2r}} (2\Phi(\sqrt{\lambda_2^2 \sigma^2 + 2r}\sqrt{T})-1) \Big) \notag \\
			=&\, A_4 - \tfrac{k}{V_B} A_6, \label{eq: integral T del cdo v=vb}
		\end{align}
		where we used that $\Big(\tfrac{\lambda_1^2 \sigma}{\nu} + \tfrac{2}{\sigma}\Big) \sqrt{\tfrac{1}{\lambda_1^2\sigma^2+2\nu}} = \tfrac{1}{\sigma \nu} \sqrt{\lambda_1^2 \sigma^2 + 2\nu}$ and $\Big(\tfrac{\lambda_2^2 \sigma}{r} + \tfrac{2}{\sigma}\Big) \sqrt{\tfrac{1}{\lambda_2^2\sigma^2+2r}} = \tfrac{1}{\sigma r} \sqrt{\lambda_2^2 \sigma^2 + 2r}$. Next, we substitute \eqref{eq: integral infty del cdo v=vb} and \eqref{eq: integral T del cdo v=vb} into \eqref{eq: formula for vb} and solve for $V_B$, yielding the formula for $\hat{V}_B$. In particular, if $\hat{V}_B\geq k$ and $\alpha<\tilde{\alpha}$, it is the unique solution to \eqref{eq: formula for vb} using Corollary \ref{cor: unique solution of vb determination}. Therefore, the claim follows.
	\end{myproof}
	
	\begin{myproof}[Proof of Proposition \ref{prop: alphastar exists}]
		If $\bar{\alpha} > 1$, the proposition follows directly from the continuity of $v$ in $\alpha$. Therefore, let us assume that $\bar{\alpha} \leq 1$. First, we note that for $\alpha>\bar{\alpha}$, the assumption that $V \to E(V)$ is non-negative does not hold, as then $\lim_{V \to \infty} E(V;V_B)=-\infty$ for all $V_B$ fixed. This, however, is excluded by assumption. Thus, we can restrict our analysis to the case where $\alpha \in [0,\bar{\alpha}]$.	Now the proof of Theorem \ref{th: VB determination} and of Proposition \ref{prop: vb monoton alpha g} imply that the only relevant solution of \eqref{eq: formula for vb} is continuous in its parameters as long as it exists. For $\alpha<\bar{\alpha}$, a solution always exists since then $\lim_{V \to \infty} E(V;V_B)=\infty$ for all $V_B$ fixed. Thus, the left-continuity of $V_B$ in $\alpha$ and the continuity of $v$ in $\alpha$ completes the proof.
	\end{myproof}
	
	\begin{myproof}[Proof of Theorem \ref{th: alpha*}]
		Before we begin the actual proof, we first show the identities used in the theorem: From the proof of Lemma \ref{lemma: dcdo dalpha}, we find that $V_B'(\alpha) = - \tfrac{R_{\alpha}(\alpha,V_B(\alpha))}{R_{V_B}(\alpha,V_B(\alpha))}$ for $\alpha \in [0,\hat{\alpha}]$, where $R_{\alpha}$ is defined as in \eqref{eq: ralpha}, $R_{V_B}$ is defined as in \eqref{eq: rvb}, and $\hat{\alpha}>0$ is defined as in Lemma \ref{lemma: dcdo dalpha}. Substituting $\alpha=0$ gives us \eqref{eq: partial vb0} and \eqref{eq: formula dvb dalpha}. Finally, \eqref{eq: partial cdo} follows immediately from \eqref{eq: dv2 cdo<} and \eqref{eq: dv2 cdo>}.
		
		For the main claim, we consider the total value $v$ of the insurance company as given in \eqref{eq: firm value v} and take the derivative with respect to $\alpha$. We restrict ourselves to $\alpha \in [0,\hat{\alpha}]$, with $\hat{\alpha}$ defined as in Lemma \ref{lemma: dcdo dalpha}. Under these conditions, we can interchange the integral and the derivative (where we suppress the dependency of $V_B$ on $\alpha$):
		\begin{align}
			\tfrac{\partial}{\partial \alpha} v(V;V_B) =&\, - \tfrac{\tau_1 \frac{G}{r} (\lambda_2+\lambda_3)}{V} (\tfrac{V_B}{V})^{\lambda_2+\lambda_3-1} V_B'(\alpha) - \rho (\lambda_2+\lambda_3+1) (\tfrac{V_B}{V})^{\lambda_2+\lambda_3} V_B'(\alpha) \notag\\
			&+ \tau_2 \int_0^\infty c_{do} (V,k,V_B,t) \diff t + \alpha \tau_2 \int_0^\infty \tfrac{\partial c_{do} (V,k,V_B,T)}{\partial \alpha} \diff t \notag \\
			=&\, -V_B'(\alpha) (\tfrac{V_B}{V})^{\lambda_2+\lambda_3} \left( \tfrac{\tau_1 \frac{G}{r} (\lambda_2+\lambda_3)}{V_B} + \rho (\lambda_2+\lambda_3+1) \right)  \notag\\
			&+ \tau_2 \int_0^\infty c_{do} (V,k,V_B,t) \diff t + \alpha \tau_2 \int_0^\infty \tfrac{\partial c_{do} (V,k,V_B,T)}{\partial \alpha} \diff t. \label{eq: dv dalpha}
		\end{align}
		Setting this equation equal to zero yields \eqref{eq: solution alpha*}. 
		
		Next, we show the existence of a $\bar{\tau}$ as described in the theorem. To do this, we evaluate the above formula at $\alpha=0$, which gives us:
		\begin{align}
			\tfrac{\partial}{\partial \alpha} v(V;V_B) \big|_{\alpha=0} =&\, -V_B'(0) (\tfrac{V_B(0)}{V})^{\lambda_2+\lambda_3} \left( \tfrac{\tau_1 \frac{G}{r} (\lambda_2+\lambda_3)}{V_B(0)} + \rho (\lambda_2+\lambda_3+1) \right) \notag \\
			&+ \tau_2 \int_0^\infty c_{do} (V,k,V_B(0),t) \diff t. \label{eq: dv dalpha alpha0}
		\end{align}
		Now, it is optimal to offer a surplus participation for the insurance company if $\tfrac{\partial}{\partial \alpha} v(V;V_B) \big|_{\alpha=0}>0$. 
		
		Therefore, let us analyze equation \eqref{eq: dv dalpha alpha0}: Lemma \ref{lemma: long term >0} implies that $V_B(0)>0$ (see \eqref{eq: vb alpha=0 main text}), and that the denominator of $V_B'(0)$ (see \eqref{eq: partial vb0}) is strictly positive. Moreover, $\Big( \tfrac{\tau_1 \frac{G}{r} (\lambda_2+\lambda_3)}{V_B(0)} + \rho (\lambda_2+\lambda_3+1) \Big)>0$, since $\lambda_3>|\lambda_2|$ by definition, 
		the price of a Down-and-Out Call Option, $c_{do}$, is non-negative, $V_B(0)$ (see \eqref{eq: vb alpha=0 main text}) is independent of $\tau_2$, and 
		for $\tau_2=1$ the term $(-V_B'(0))$ (see \eqref{eq: partial vb0}) is positive because $T < \infty$. Consequently, we obtain that $\tfrac{\partial}{\partial \alpha} v(V;V_B) \big|_{\alpha=0}>0$ for $\tau_2=1$, and that $\tfrac{\partial}{\partial \alpha} v(V;V_B) \big|_{\alpha=0}$ is continuous in $\tau_2$. This leads to the existence of a $\bar{\tau} \in (0,1)$ such that $\tfrac{\partial}{\partial \alpha} v(V;V_B) \big|_{\alpha=0}>0$ for $\tau_2 \in (\bar{\tau},1)$, implying the claim.
		
		Finally, it remains to demonstrate the existence of $\bar{\bar{\tau}}$. To establish this, we have by \eqref{eq: dv dalpha}:
		\begin{align*}
			\tfrac{\partial}{\partial \alpha} v(V;V_B) \big|_{\tau_2=0} =  -V_B'(\alpha) (\tfrac{V_B}{V})^{\lambda_2+\lambda_3} \left( \tfrac{\tau_1 \frac{G}{r} (\lambda_2+\lambda_3)}{V_B} + \rho (\lambda_2+\lambda_3+1) \right) < 0,
		\end{align*}
		since $\lambda_3>|\lambda_2|$ by definition and $V_B'(\alpha)>0$ (since $V_B$ is increasing in $\alpha$ by Proposition \ref{prop: vb monoton alpha g}). Hence, it follows that if $\tau_2=0$, the optimal choice is $\alpha^*=0$. Moreover, the inequality
		\begin{align*}
			-V_B'(\alpha) (\tfrac{V_B}{V})^{\lambda_2+\lambda_3} \left( \tfrac{\tau_1 \frac{G}{r} (\lambda_2+\lambda_3)}{V_B} + \rho (\lambda_2+\lambda_3+1) \right) < 0
		\end{align*}
		persists independently of the specific choice of $\tau_2$. Finally, since $\int_0^\infty c_{do} (V,k,V_B,t) \diff t + \linebreak \alpha \int_0^\infty \tfrac{\partial c_{do} (V,k,V_B,T)}{\partial \alpha} \diff t$ in \eqref{eq: dv dalpha} is bounded (by Lemmas \ref{lem: int cdo < infty} and \ref{lemma: dcdo dalpha}), it follows that
		there exists a value $\bar{\bar{\tau}} > 0$ such that $\tfrac{\partial}{\partial \alpha} v(V;V_B)<0$ for all $0 \leq \tau_2 < \bar{\bar{\tau}}$, implying that $\alpha^*=0$ for all $0 \leq \tau_2 < \bar{\bar{\tau}}$. Thus, the claim is established.
	\end{myproof}
	
	
	\begin{myproof}[Proof of Proposition \ref{prop: gstar exists}]
		
		Since $v$ is continuous in $g$ (because $V_B$ is continuous in $g$), it suffices to show that the supremum is attained in a compact interval of $g$. First, assume that $-\tfrac{2A_1}{rT} + 2 A_2 - \tau_1 (\lambda_2+\lambda_3) \neq 0$. By Lemma \ref{lemma: cdo unif bounded L1}, it follows that the term $\tau_2 \alpha \int_0^\infty \tfrac{\partial  c_{do} (V,k,V_B,t) }{\partial V} \big|_{V=V_B} \diff t - \alpha \int_0^T \tfrac{\partial c_{do} (V,k,V_B,t)}{\partial V} \big|_{V=V_B} \diff t$ is uniformly bounded in $V_B$, and thus in $g$. From equation \eqref{eq: formula for vb}, we know that $V_B \xrightarrow{g \to \infty} \infty$ (using that $g = \tfrac{G}{T}$). Therefore, by the continuity of $V_B$ in $g$, there exists a $\bar{g} \geq 0$ such that $V_B(g) \geq V_0$ for all $g \geq \bar{g}$. If $V_B \geq V_0$, however, the insurance company declares bankruptcy immediately. In this case, the insurance company's value is given by $v(V;V_B) = V - \rho V$ for all $V_B \geq V_0$, according to the first equation in \eqref{eq: firm value v}, where $TB_1 = TB_2 = 0$ and $BC = \rho V$ (see \eqref{eq: def TB1}, \eqref{eq: def TB2}, and \eqref{eq: def BC}). Therefore, $v$ is constant for $V_B \geq V_0$, and we can consequently set $V_B=V_0$. Hence, we can restrict $g$ to the interval $[0,\bar{g}]$, completing the proof.
	\end{myproof}
	
	\begin{myproof}[Proof of Theorem \ref{th: g*}]
		Note that this proof is similar to the proof of Theorem \ref{th: alpha*}.
		
		We begin by showing the identities stated in the theorem: Equation \eqref{eq: vb0 g} follows directly from \eqref{eq: formula for vb}. Next, from the proof of Lemma \ref{lemma: dcdo dg}, we obtain that $V_B'(g) = - \tfrac{R_{g}}{R_{V_B}}$, where $R_{g}$ is defined in \eqref{eq: rg} and $R_{V_B}$ is defined in \eqref{eq: rvb2}. Evaluating this expression at $g=0$ and inserting $V_B(0)$ gives us the desired results in \eqref{eq: formula dvb dg} and \eqref{eq: partial vb0 g}. Finally, the equations \eqref{eq: partial cdo g} and \eqref{eq: partial cdo g0} are immediate consequences of \eqref{eq: dv2 cdo<} and \eqref{eq: dv2 cdo>}.
		
		For the main claim \eqref{eq: solution g*}, we again utilize the total value $v$ of the insurance company, as defined in \eqref{eq: firm value v}, and differentiate it with respect to $g$. We consider $g \in [0,\hat{g}]$, where $\hat{g}$ is defined as in Lemma \ref{lemma: dcdo dg}, such that we can interchange the integral and the derivative (where we suppress the dependency of $V_B$ on $g$):
		\begin{align}
			\tfrac{\partial}{\partial g} v(V;V_B) =&\, \tfrac{\tau_1T}{r} (1-(\tfrac{V_B}{V})^{\lambda_2+\lambda_3})  - \tfrac{\tau_1 gT (\lambda_2+\lambda_3)}{Vr} (\tfrac{V_B}{V})^{\lambda_2+\lambda_3-1} V_B'(g) \notag \\ 
			&- \rho (\lambda_2+\lambda_3+1) (\tfrac{V_B}{V})^{\lambda_2+\lambda_3} V_B'(g) + \alpha \tau_2 \int_0^\infty \tfrac{\partial c_{do} (V,k,V_B,T)}{\partial g} \diff t. \label{eq: dv dg}
		\end{align}
		Setting this equation equal to zero results in the equation \eqref{eq: solution g*}.
		
		It remains to demonstrate that offering a guarantee rate is indeed optimal. To do so, we evaluate this formula at $g=0$, which leads to:
		\begin{align}
			\tfrac{\partial}{\partial g} v(V;V_B) \big|_{g=0} =&\, \tfrac{\tau_1 T}{r} (1-(\tfrac{V_B(0)}{V})^{\lambda_2+\lambda_3}) - \rho (\lambda_2+\lambda_3+1) (\tfrac{V_B(0)}{V})^{\lambda_2+\lambda_3} V_B'(0) \notag \\
			&+ \alpha \tau_2 \int_0^\infty \tfrac{\partial c_{do} (V,k,V_B,T)}{\partial g} \Big|_{g=0} \diff t >0, \label{eq: dv dg0}
		\end{align}
		by assumption. Note that we applied Lemma \ref{lemma: dcdo dg} to move the point estimation inside the integral. Therefore, it follows that $g^*>0$.
	\end{myproof}
	
	\begin{myproof}[Proof of Proposition \ref{prop: alphastar and gstar exists}]
		This proposition follows directly from Propositions \ref{prop: alphastar exists} and \ref{prop: gstar exists}. Note that the proof of Proposition \ref{prop: alphastar exists} is independent of $g$, and the proof of Proposition \ref{prop: gstar exists} holds for all $\alpha \leq 1$.
	\end{myproof}
	
	\section{Additional results}
	
	
	In this section, we present an additional result that is referenced in the main text.

	\begin{proposition} \label{prop: solution g* existence}
		Let 
		$\tau_2$ be sufficiently large such that there exists an $\varepsilon>0$ with $R_{V_B} \geq \varepsilon$ for all $(g,V_B) \in [0,\infty) \times [0,V_0]$ (defined as in \eqref{eq: rvb2}). Then, under Assumption \ref{ass: g*>0}, equation \eqref{eq: solution g*} always admits a solution.
	\end{proposition}
	
	\begin{proof}
		First, we observe that if $\tau_2=1$, then $R_{V_B} > 0$ for all $(g,V_B) \in [0,\infty) \times [0,V_0]$, as follows from \eqref{eq: rvb2}, together with Lemmas \ref{lemma: long term >0} and \ref{lemma: dv2 cdo bounded l1}.
		
		Second, note that the right-hand side of \eqref{eq: solution g*} is nothing else than $\tfrac{\partial}{\partial g} v(V;V_B)$ (see \eqref{eq: dv dg}). Thus, we show the existence of an $\tilde{g} \in (0,\infty)$ such that $\tfrac{\partial}{\partial g} v(V;V_B) = 0$ for $g=\tilde{g}$ to prove the lemma.
		
		Under the assumption made in this lemma, we can choose $\hat{g} = \infty$ in the proof of Lemma \ref{lemma: dcdo dg}, and consequently, also in the proof of Theorem \ref{th: g*}. By Proposition \ref{prop: vb monoton alpha g}, $V_B$ is increasing in $g$, implying that $V_B'$ is non-negative. Moreover, we know that $V_B \leq V$ by assumption (in Subsection \ref{subsection: model setup}). Then, from the definition of $V_B'(g)$ (see \eqref{eq: formula dvb dg}), we conclude that $\lim_{g \to \infty} g V_B'(g) = \lim_{g \to \infty} V_B(g)$ and $V_B'(g) \xrightarrow{g \to \infty} 0$, because $G = gT$ and $V_B$ is increasing and bounded. Moreover, it holds that $V_B(g) \xrightarrow{g \to \infty} V_0$. Indeed, from \eqref{eq: formula for vb}, it follows that a solution of this equation diverges to infinity as $G \to \infty$, since $\tfrac{2(P-\frac{G}{r})A_1}{rT} + 2 \tfrac{G}{r} A_2 - \tau_1 \frac{G}{r}(\lambda_2+\lambda_3) \xrightarrow{G \to \infty} \infty$. The reason is that it follows from \eqref{eq: formula for vb} that, if $V_B$ remained bounded, all other terms remain bounded as well, implying that the right hand side cannot be $0$ for sufficiently large $G$. Since, by construction, we set $V_B = V_0$, corresponding to immediate bankruptcy, whenever the only relevant solution of \eqref{eq: formula for vb} exceeds $V_0$, it follows that $V_B(g) \xrightarrow{g \to \infty} V_0$. Furthermore, Lemma \ref{lemma: dcdo dg} implies that we can interchange the integral sign of the last term of \eqref{eq: solution g*} with the limit of $g \to \infty$. Then, \eqref{eq: dg cdo<} and \eqref{eq: dg cdo>} imply that the last term of \eqref{eq: solution g*} converges to $0$ for $g \to \infty$ as $V_B'(g) \xrightarrow{g \to \infty} 0$. Therefore, we have $\tfrac{\partial}{\partial g} v(V;V_B) \xrightarrow{g \to \infty} -\tfrac{\tau_1 (\lambda_2+\lambda_3)}{rT} < 0$ (see \eqref{eq: solution g*}). Finally, using the continuity of $v(V;V_B)$ in $g$ and \eqref{eq: dv dg0}, the intermediate value theorem ensures the desired result.
	\end{proof}	
	
	\newpage
	\footnotesize
	\bibliography{bibliography}
	\footnotesize
	\bibliographystyle{plain}
	
	\footnotesize
\end{document}